\documentclass[a4paper,UKenglish,cleveref, autoref, thm-restate]{lipics-v2021}

\title{Complete Relational Logic for Infinite-Dimensional Quantum Programs with Unbounded Assertions} 

\titlerunning{Complete Relational Logic for Infinite-Dimensional Quantum Programs} 

\author{Gilles Barthe}{Max-Planck Institute for Security and Privacy (MPI-SP), Germany \and IMDEA Software Institute, Spain}{gilles.barthe@mpi-sp.org}{https://orcid.org/0000-0002-3853-1777}{Corresponding author.}

\author{Minbo Gao}{Key Laboratory of System Software (Chinese Academy of Sciences), Beijing, China \and Institute of Software, Chinese Academy of Sciences, Beijing, China \and University of Chinese Academy of Sciences, Beijing, China}{gaomb@ios.ac.cn}{https://orcid.org/0009-0006-2976-548X}{}

\author{Jam Kabeer Ali Khan}
{Max-Planck Institute for Security and Privacy (MPI-SP), Germany}
{jamkhan@connect.hku.hk}
{https://orcid.org/0009-0007-8505-9667}
{}

\author{Matthijs Muis}{Max-Planck Institute for Security and Privacy (MPI-SP), Germany \and Eidgen\"ossische Technische Hochschule (ETH) Z\"urich, Switzerland}
{mamuis@ethz.ch}
{https://orcid.org/0009-0001-3421-5793}
{}

\author{Ivan Renison}
{Max-Planck Institute for Security and Privacy (MPI-SP), Germany}
{ivan.renison@mi.unc.edu.ar}
{https://orcid.org/0009-0006-6216-3763}
{}

\author{Keiya Sakabe}{Ludwig-Maximilians-Universit\"at München, Germany}{keiya.sakabe@lmu.de}{https://orcid.org/0009-0003-8894-4400}{}

\author{Michael Walter}{Ludwig-Maximilians-Universit\"at München, Germany \and Munich Center for Quantum Science and Technology (MCQST), Germany}{michael.walter@lmu.de}{https://orcid.org/0000-0002-3073-1408}{}

\author{Yingte Xu}{Max-Planck Institute for Security and Privacy (MPI-SP), Germany}{yingte.xu@mpi-sp.org}{https://orcid.org/0000-0001-9071-7862}{}

\author{Tianshi Yu}{Key Laboratory of System Software (Chinese Academy of Sciences), Beijing, China \and Institute of Software, Chinese Academy of Sciences, Beijing, China}{yuts@ios.ac.cn}{https://orcid.org/0009-0007-2498-3346}{}

\author{Li Zhou}{Key Laboratory of System Software (Chinese Academy of Sciences), Beijing, China \and Institute of Software, Chinese Academy of Sciences, Beijing, China}{zhouli@ios.ac.cn}{https://orcid.org/0000-0002-9868-8477}{Corresponding author.}

\authorrunning{G. Barthe et al.} 

\Copyright{Gilles Barthe, Minbo Gao, Jam Kabeer Ali Khan, Matthijs Muis, Ivan Renison, Keiya Sakabe, Michael Walter, Yingte Xu, Tianshi Yu, Li Zhou} 



\category{} 

\funding{MG, TY, and LZ were supported by National Key Research and Development Program of China
(Grant No.\ 2023YFA1009403).
MW and KS acknowledge support by the European Research Council (ERC Grant~101040907), the German Research Foundation (556164098, EXC-2111 – 390814868, EXC-2092 - 390781972), and the German Federal Ministry of Research, Technology and Space (QuSol, 13N17173).
}

\nolinenumbers 

\EventEditors{Claudia Faggian and Joost-Pieter Katoen}
\EventNoEds{2}
\EventLongTitle{41st Annual Symposium on Logic in Computer Science (LICS 2026)}
\EventShortTitle{LICS 2026}
\EventAcronym{LICS}
\EventYear{2026}
\EventDate{July 20--23, 2026}
\EventLocation{Lisbon, Portugal}
\EventLogo{}
\SeriesVolume{380}
\ArticleNo{43}

\usepackage{mathtools} 
\usepackage{amssymb}    
\usepackage{amsthm}     
\usepackage{physics}    
\usepackage{ebproof}    
\usepackage{pifont}
\usepackage{longtable}
\usepackage{hyperref}
\usepackage{interval}
\intervalconfig{soft open fences}
\usepackage{booktabs}
\usepackage{stmaryrd}
\usepackage[short]{optidef}
\usepackage{microtype}
\usepackage{tikz}
\usetikzlibrary{shapes,arrows.meta,positioning,decorations.pathreplacing,calc}

\newcommand{\yes}{$\ding{52}$}
\newcommand{\no}{$\ding{54}$}

\DeclarePairedDelimiter\rbra{\lparen}{\rparen}

\DeclarePairedDelimiter\cbra{\{}{\}}

\DeclarePairedDelimiter\ave{\langle}{\rangle}

\newcommand{\states}[1]{%
  \if\relax\detokenize{#1}\relax
    \mathcal{D}^1
  \else
    \mathcal{D}^1(#1)
  \fi
}

\newcommand{\cstates}{\mathcal{S}}

\newcommand{\pstates}[1]{%
  \if\relax\detokenize{#1}\relax
    \mathcal{D}^{\le 1}
  \else
    \mathcal{D}^{\le 1}(#1)
  \fi
}

\newcommand{\allmaps}[1]{%
  \if\relax\detokenize{#1}\relax
    \mathcal{F}
  \else
    \mathcal{F}(#1)
  \fi
}

\newcommand{\bbmaps}[1]{%
  \if\relax\detokenize{#1}\relax
    \mathcal{F}_{\mathrm{bel}}
  \else
    \mathcal{F}_{\mathrm{bel}}(#1)
  \fi
}

\newcommand{\bndmaps}[1]{%
  \if\relax\detokenize{#1}\relax
    \mathcal{F}_{\mathrm{bnd}}
  \else
    \mathcal{F}_{\mathrm{bnd}}(#1)
  \fi
}

\newcommand{\posmaps}[1]{%
  \if\relax\detokenize{#1}\relax
    \mathcal{F}_{\mathrm{\geq 0}}
  \else
    \mathcal{F}_{\mathrm{\geq 0}}(#1)
  \fi
}

\newcommand{\couplings}[2]{%
  \if\relax\detokenize{#1}\relax
    \if\relax\detokenize{#2}\relax
      \mathcal{C}
    \fi
  \else
    \mathcal{C}(#1,#2)
  \fi
}

\newcommand{\dualset}[1]{%
  \if\relax\detokenize{#1}\relax
      \mathcal{G}
  \else
    \mathcal{G}(#1)
  \fi
}
\newcommand{\E}{\mathbb{E}}

\newcommand{\set}[2] {\left\{\, #1 \colon #2 \,\right\}}

\newcommand{\supp} {\operatorname{supp}}
\newcommand{\spanspace} {\operatorname{span}}

\newcommand{\Ran} {\operatorname{Ran}}
\newcommand{\mul} {\mathsf{Mul}}

\newcommand{\tdom}{\mathsf{Dom}_{\mathfrak{t}}}

\newcommand{\cH} {\mathcal{H}}

\let\geq=\ge
\let\leq=\le
\newcommand{\bN} { {\mathbb{N}}}   
\newcommand{\bC} { {\mathbb{C}}}

\newcommand{\bR} { {\mathbb{R}}}

\newcommand{\<}{{\langle}}
\renewcommand{\>}{{\rangle}}
\newcommand{\Dom}{{\operatorname{Dom}}}

\newcommand{\trunc}[2]{\textup{trunc} \left( #1, #2 \right)}

\newcommand{\distr}{\states{}}
\newcommand{\sdistr}{\pstates{}}

\newcommand{\ereals}{\overline{\mathbb{R}}}
\newcommand{\cvars}{\mathsf{cVar}}
\newcommand{\qvars}{\mathsf{qVar}}

\newcommand{\cqstate}{\Delta}
\newcommand{\cqpredicate}{Q}
\newcommand{\cqstates}[2]{\states{#1,#2}}
\newcommand{\cqpstates}[2]{\pstates{#1,#2}}
\newcommand{\cqbndmaps}[2]{\bndmaps{#1,#2}}
\newcommand{\cqbbmaps}[2]{\bbmaps{#1,#2}}

\newcommand{\qbndmaps}[1]{\bndmaps{#1}}
\newcommand{\qbbmaps}[1]{\bbmaps{#1}}

\newcommand{\linearmaps}[1]{\mathcal{L}(#1)}
\newcommand{\qstate}[1]{\states{#1}}
\newcommand{\qpstate}[1]{\pstates{#1}}

\newcommand{\ithen}{\mathbf{then}}
\newcommand{\ielse}{\mathbf{else}}
\newcommand{\ife}{\mathbf{fi}}
\newcommand{\ifte}[3]{\ifb ~ #1  ~ \ithen ~ #2 ~\ielse ~  #3 ~\ife}

\newcommand{\rndarrow}{%
  \stackrel{\,\raisebox{-.25ex}[.25ex]%
   {\tiny $\mathdollar$}}{\raisebox{-.25ex}[.25ex]{$\leftarrow$}}}
   
\newcommand{\whiled}[2]{\while ~#1  ~ \wdo~ #2 ~ \wod}
\newcommand{\sample}[2]{#1\rndarrow #2}
\newcommand{\assign}[2]{#1 := #2}

\newcommand{\cJ}{\mathcal{J}}
\newcommand{\coupling}[2]{\langle #1, #2 \rangle}
\newcommand{\DualSet}[1]{\mathcal{Y}\rbra*{#1}}

\newcommand{\opt}[1]{\operatorname{Opt}_P^{#1}}
\newcommand{\dopt}[1]{\operatorname{Opt}_D^{#1}}

\newcommand{\qvar}{\textup{qVar}}
\newcommand{\oq}{\overline{q}}
\newcommand{\skp}{\mathbf{skip}}
\newcommand{\guard}{\Box}
\newcommand{\ifb}{\mathbf{if}}
\newcommand{\while}{\mathbf{while}}
\newcommand{\wdo}{\mathbf{do}}
\newcommand{\wod}{\mathbf{od}}

\newcommand{\abort}{\textbf{abort}}
\newcommand{\imeas}{\mathbf{meas}}
\newcommand{\whiledtit}[3]{\while_{#3}\ {#1}\ \wdo\ {#2}}

\newcommand{\sem}[1]{\llbracket #1 \rrbracket} 
\newcommand{\AST}{\ensuremath{\mathsf{AST}}} 
\newcommand{\wprecond}[2]{\textup{wp}[#1](#2)} 

\newcommand{\triple}[3]{\{#1\} \ #2 \ \{#3\}} 
\newcommand{\rtriple}[4]{\triple{#1}{#2 \sim #3}{#4}} 

\renewcommand{\vec}[1]{\ensuremath{\mathbf{#1}}}
\renewcommand{\u}{\vec{u}}
\renewcommand{\v}{\vec{v}}
\newcommand{\w}{\vec{w}}

\newcommand{\shift}[2]{\mathrm{shift}(#1,#2)}

\newcommand{\cstate}{s}
\newcommand{\cqsim}[2]{{(\;\!\!\!|{#1},{#2}|\!)}}
\newcommand{\true}{\mathsf{True}}
\newcommand{\false}{\mathsf{False}}

\newcommand{\RuleRef}[1]{%
{\hyperref[rule:#1]{\textsc{[#1]}}}%
}

\newcommand{\RuleLabel}[1]{%
  {\small\textsc{[#1]}}%
  \label{rule:#1}%
}

\begin{document}

\maketitle
\begin{abstract}

We present sound and complete relational program logics for infinite-dimensional quantum and classical-quantum programs. The logics model assertions as self-adjoint unbounded linear relations, which simultaneously support quantitative and qualitative reasoning. Our main theoretical results include new convergence theorems and infinite-dimensional duality theorems for infinite-dimensional quantum states, which we use to establish completeness.
\end{abstract}

\section{Introduction}
There is a rich landscape of quantum programming languages, see e.g.\,\cite{heim2020quantum}. Languages
may notably differ in their underlying paradigm, e.g.\, functional vs
imperative, or classical vs. quantum control flow. Another difference, 
which is central to this paper, is whether programs operate over
finite-dimensional or (discrete) infinite-dimensional Hilbert spaces,
or equivalently if the underlying set of values is finite or countably
infinite\footnote{Throughout the paper, we only consider separable
infinite-dimensional Hilbert spaces, and simply refer them as
infinite-dimensional Hilbert spaces.}. While developing sound
foundations for the latter setting is significantly more complex,
infinite-dimensional quantum programs are needed to capture many
applications, including some examples of quantum cryptography and 
quantum walks. Yet another difference is whether programs are purely
quantum and only carry quantum computations or are classical-quantum
and carry both classical and quantum computations---note that if the
set of classical values is countable then classical-quantum programs
can be embedded into infinite-dimensional quantum programs. The case
of classical-quantum programs is important, as it provides a natural
formalism to capture many examples in quantum machine learning,
quantum and post-quantum cryptography\footnote{Quantum cryptography is the 
branch of cryptography that uses quantum phenomena to develop 
cryptographic primitives that are (conditionally) secure against 
quantum adversaries. In contrast, post-quantum cryptography develops 
classical, or non-quantum,  cryptographic primitives, that are 
(conditionally) secure against quantum adversaries. In both cases,
conditional security means security down to computational assumptions.}.

The purpose of this work is to develop foundations and program logics
for reasoning about infinite-dimensional quantum and classical-quantum
programs. Specifically, our goal is to build relational program
logics, i.e.\, program logics that can reason about executions of two
programs. Such logics are particularly well-suited to reason about
security of cryptographic schemes and generalization properties of
machine-learning, and have been extensively studied in the context of
probabilistic programs~\cite{BartheGZ09,erhl}, and finite-dimensional
quantum programs~\cite{Unr19a,GHY19,qotl}. However, extending these
logics to the infinite-dimensional setting poses two significant
challenges.

\subparagraph*{Challenge 1: assertion language.}
 A main challenge in quantum program verification is to define an assertion language that captures a large set of properties. For non-relational program logics, the common choice is to adopt
quantitative assertions~\cite{DP06}. Instead, it is common for relational program logics to use qualitative, e.g.\, projective, assertions, as they simplify lockstep reasoning, i.e.\, reasoning
about two programs that have the same control flow and perform the same number of samplings. Unfortunately, qualitative assertions also limit non-lockstep reasoning, i.e.\, reasoning about two programs whose control-flow or sampling alignment differ. Thus, the ideal solution is to combine qualitative and quantitative assertions, as explored in~\cite{erhl} in the probabilistic setting and in~\cite{qotl} for the finite-dimensional quantum setting. However, the approach developed in~\cite{qotl} relies on an ad hoc construction to allow assertions to take infinite values, 
which is required to provide adequate support for lockstep reasoning. Unfortunately, this 
construction requires several ad hoc definitions for basic operations on assertions and it is not clear how to extend it to our more general setting.

\subparagraph*{Challenge 2: complete proof system.}
A main challenge in relational program logic is to define proof systems that are practical, sound and complete. On the one hand, practical proof systems are based on coupling-based, qualitative
logics. On the other hand, qualitative, coupling-based logics are inherently incomplete for non-lockstep reasoning. Recently, \cite{qotl} shows how the tension can be overcome by leveraging duality theorems from quantum optimal transport. Informally, duality theorems provide a means to decompose any relational judgment into an equivalent (universally quantified) conjunction of unary judgments, so that completeness of the relational proof system can be derived from completeness of the non-relational proof system. Unfortunately, duality theorems for quantum optimal transport, see \Cref{sec:relwork}, are 
limited to the finite-dimensional case, and to a subset of assertions that rules out infinite-valued assertions used to support lockstep reasoning.

\paragraph*{Contributions}
This paper proposes the first sound and complete relational logics for infinite-dimensional 
quantum and classical-quantum programs and a general form of assertions. To this end, we make the following contributions.

First, we develop a new theory of assertions based on linear relations~\cite{behrndt2020boundary,cross1998multivalued}. Linear relations provide a uniform 
framework that captures both qualitative and quantitative assertions, and allows to set the values of assertions to infinite values, as required for supporting lockstep reasoning. In order to support our theory of assertions, we introduce carefully crafted notions of trace, L\"owner order, convergence, and use these notions to reestablish classic theorems, e.g.\, the Monotone Convergence Theorem and Fatou's Lemma, which are needed in the meta-theory of our program logics.

Second, we prove duality theorems for infinite-dimensional quantum states and classical-quantum states. Our duality theorems extend prior works in two dimensions: first, they consider a richer notion of quantum states (infinite-dimensional or classical-quantum); second, they consider our more general notion of assertions. The main technical ingredient of the proofs of the theorems is a dimension-independent perturbation bound. Informally, the perturbation bound controls the approximation error of the optimal transport cost when truncating infinite-dimensional states to larger and larger finite-dimensional projections. Thanks to a careful analysis that makes the bound dimension-independent, one can show that duality is preserved when taking the limit from finite-dimensional states to infinite-dimensional states, and to use the finite-dimensional duality theorem of~\cite{qotl} for the base case.

Third, we develop sound and complete program logics for infinite-dimensional quantum programs and classical-quantum programs. Our program logics distill the essence of the completeness results, by featuring a core set of rules. Notably, they use a weakest precondition rule or a lifting rule which allows connections with existing relational and unary logics. Moreover, we demonstrate the applicability of our logics by formalizing two illustrative examples from quantum walks and quantum machine learning.

\paragraph*{Organization of Paper}
\Cref{sec:probprog} and \Cref{sec:finite quantum} are introductory sections that present the essence of prior work on probabilistic and finite-dimensional quantum programs. Then, \Cref{Unbounded_Section} introduces our general notion of assertions based on linear relations and \Cref{sec:duality} establishes new duality theorems for infinite-dimensional quantum states and general assertions. Finally, \Cref{sec:infinite programs} and \Cref{sec:cqprog} leverage the results of \Cref{Unbounded_Section} and \Cref{sec:duality} to provide sound and complete program logics for infinite-dimensional quantum programs and classical-quantum programs. 
\begin{figure}
\begin{equation*}
    \begin{aligned}
      S & ::= \  \skp \\
      & \mid  \assign{x}{e}  \\
      & \mid \sample{x}{\mu}  \\
              & \mid  S_1; S_2  \\
              &\mid \ifte{e}{S_1}{S_2} \\
              &\mid \whiled{e}{S}
\end{aligned}
\begin{aligned}
  S & ::= \  \skp  \\
  & \mid q \coloneqq \u \\
  & \mid \oq \coloneqq U [\oq] \\
  & \mid  S_1; S_2  \\
  &\mid \ifb\ (\guard v\cdot M[\oq] = v \to S_v)\ \ife \\
&\mid \while\ M[\oq]=1\ \wdo\ S\ \wod 
\end{aligned}
\end{equation*}
\caption{\textsf{pWhile} and \textsf{qWhile} languages.  Here $x$ ranges over a finite set $\cvars$ of classical variables, $e$ ranges over expressions, $\mu$ ranges over distribution expressions; $q$ ranges over a finite set $\qvar$ of quantum variables, $\oq$ ranges over a distinct list of $\qvar$, $\u$ ranges over pure states, $U$ ranges over unitaries, $M$ ranges over measurements. Expressions take values over a set $\mathcal{V}$.
  }\label{fig:syntax}
\end{figure}

\section{Probabilistic Programs}\label{sec:probprog}
This section sets the stage for our work, 
by presenting how the duality theorem is used for reasoning about probabilistic programs. All the results in this section are known---sometimes in a different but equivalent form. The sole exception is the completeness theorem for unbounded assertions, which is new.

\subparagraph*{Preliminaries.}
Let $\ereals=\mathbb{R}\cup\{-\infty,+\infty\}$, and for a set $X$, let
$\allmaps{X}=\{X\rightarrow\ereals\}$ denote the set of extended-real-valued functions on $X$. We define the subset of lower-bounded functions $\bbmaps{X}=\{ f\in\allmaps{X}
\mid -\infty < \inf(f)\}$, and the set of bounded functions $\bndmaps{X}=\{ f\in\bbmaps{X} \mid
\sup(f) < +\infty \}$.
The set $\allmaps{X}$ is equipped with a partial order defined pointwisely: for $f_1, f_2\in \allmaps{X}$, we write $f_1\sqsubseteq f_2$ if $f_1(x)\le f_2(x)$ for all $x\in X$.
Given functions $f_1\in \allmaps{X_1}$ and $f_2\in \allmaps{X_2}$, we define their sum $f_1\boxplus f_2 \in \allmaps{X_1\times X_2}$ by $(f_1\boxplus f_2)(x_1, x_2)\triangleq f_1(x_1)+f_2(x_2)$.

We let $\distr(X)$ denote the set of discrete distributions over a set
$X$. For $\mu\in \distr(X)$, and $f\in \bbmaps{X}$, the expectation can always be defined by $\E_{\mu}[f] = \sum_{x\in X} \mu(x) f(x)$.
Moreover, we define the set of couplings $\mu_1\in\distr(X_1)$
and $\mu_2\in\distr(X_2)$ as the unique set
$\couplings{\mu_1}{\mu_2}\subseteq \distr(X_1\times X_2)$ such that
$\mu\in\couplings{\mu_1}{\mu_2}$ iff for all $x_1,x_2$, $\mu_1(x_1)=
\sum_{x_2} \mu(x_1,x_2)$ and $\mu_2(x_2)= \sum_{x_1}
\mu(x_1,x_2)$.

\subparagraph*{Duality theorem.}
Optimal transport~\cite{Vil03,Villani08,RR98} aims to minimize the transport cost between two distributions. A main result in optimal transport is the Kantorovich-Rubinstein theorem~\cite{Kan42}, which states that the
optimal transport cost (defined on the left of the equation) coincides with the optimal cost of the dual problem (on the right of the equation).
\begin{theorem}[Kantorovich-Rubinstein Duality Theorem]
\label{thm:kr-duality}
Let $Q\in\bbmaps{X_1\times X_2}$, and let $\dualset{Q}\subseteq
  \bndmaps{X_1}\times\bndmaps{X_2}$ such that
$(Q_1,Q_2)\in \dualset{Q}$ iff
      $Q_1 \boxplus Q_2\sqsubseteq Q$. Then
\[ \inf_{\mu\in \couplings{\mu_1}{\mu_2}} \E_{\mu}[Q] =\sup_{(Q_1,Q_2)\in\dualset{Q}}
  \{\E_{\mu_1}[Q_1]+\E_{\mu_2}[Q_2]\}. \]
\end{theorem}
In the remainder of the section, we use the duality theorem to establish soundness and completeness of relational program logics. 

\subparagraph*{Programs and assertions.}
Programs are written in the \textsf{pWhile} language; its syntax is defined in \Cref{fig:syntax}. Programs operate over 
states, where the set of states is defined as
$\cstates=\cvars \rightarrow\mathcal{V}$. Here $\cvars$ denotes a finite set of variables and $\mathcal{V}$ denotes a \emph{countably infinite} set of values. 

Each program $S$ has an interpretation $\sem{S}: \cstates \rightarrow \sdistr(\cstates)$. We say that $S$ is almost surely terminating, or AST,
if the output of a computation is always a full distribution, i.e.\, $\sem{S}:\cstates \rightarrow \distr(\cstates)$.

Relational assertions are elements of $\bbmaps{\cstates\times\cstates}$, 
i.e.\, bounded by below maps from pairs of states to \emph{extended} reals. Relational assertions are naturally ordered by $\sqsubseteq$,
which is the point-wise lifting of the usual order on extended reals.
Note that our choice of assertions is less standard than the usual \emph{non-negative} 
maps from pairs of states to (extended) reals. However, the two notions are equivalent for AST programs via a shift argument, which will be given below after the program logic is introduced.

\subparagraph*{Program logic.}
Judgments are of the form $\rtriple{P}{S_1}{S_2}{Q}$,
\textbf{where $S_1,S_2$ are AST programs}\footnote{The restriction to AST programs applies to all logics in the paper.} and $P,Q\in\bbmaps{\cstates\times\cstates}$.
We say that $\rtriple{P}{S_1}{S_2}{Q}$ is valid, written $\models
\rtriple{P}{S_1}{S_2}{Q}$, if for every $s_1,s_2\in\cstates$, there
exists $\mu\in \couplings{\sem{S_1}~s_1}{\sem{S_2}~s_2}$ such that
$\E_{\mu}[Q] \leq P(s_1,s_2)$.

\begin{remark}[Bounded by below and non-negative assertions]
  By definition of validity, we have: for every fixed constant $M$,  if $S_1, S_2$ are AST, then $\vDash \rtriple{P}{S_1}{S_2}{Q}$ iff $\vDash \rtriple{P+M}{S_1}{S_2}{Q+M}$. Therefore, we can shift all bounded by below assertions to non-negative ones.
\end{remark}

\Cref{fig:prob:rules} presents the rules of the logic. The \RuleRef{duality}
rule internalizes the duality theorem. The \RuleRef{conseq} rule is the usual
rule for consequence, and is used for strengthening pre-conditions or
weakening post-conditions. Finally, the \RuleRef{wp} rule states that one can prove
the validity of the weakest precondition, defined as
$\wprecond{S}{Q}=\lambda s. \E_{\sem{S}~s}[Q]$.
\begin{theorem}[Soundness and Completeness of core rules]\label{thm:prob:complete}
A judgment is valid iff it can be derived with the rules \RuleRef{duality},
\RuleRef{conseq}, and \RuleRef{wp}.
\end{theorem}
\begin{proof}
Soundness follows from the definition of weakest precondition and from
the duality theorem. For completeness, assume that $\models
\rtriple{P}{S_1}{S_2}{Q}$. By definition of validity and of weakest
precondition, and by linearity of expectations and \Cref{thm:kr-duality}, 
we have $\wprecond{S_1}{Q_1}\boxplus
\wprecond{S_2}{Q_2}\sqsubseteq P$ for every $(Q_1,Q_2)\in\dualset{Q}$.  By the \RuleRef{wp} and \RuleRef{conseq} rules, it
follows that $\vdash \rtriple{P}{S_1}{S_2}{Q_1\boxplus Q_2}$. One concludes from
applying the \RuleRef{duality} rule.
\end{proof}

\begin{figure}[ht]
  
  $
    \begin{prooftree}
      \hypo{ P' \sqsupseteq P}
      \hypo{ \vdash \rtriple{P}{S_1}{S_2}{Q}}
       \hypo{ Q \sqsupseteq Q'}
      \infer3{\vdash \rtriple{P'}{S_1}{S_2}{Q'}}
    \end{prooftree}
  $\quad
    \RuleLabel{conseq} \\ [0.5em]
  $
    \begin{prooftree}
      \hypo{ 
        \forall (Q_1, Q_2)\in\dualset{Q},\vdash
          \rtriple{P}{S_1}{S_2}{Q_1 \boxplus Q_2}
     }
      \infer1{\vdash \rtriple{P}{S_1}{S_2}{Q}}
    \end{prooftree}
  $ \quad
    \RuleLabel{duality} \\ [0.5em]
  $
    \begin{prooftree}
      \hypo{
    Q_1,Q_2\in\bndmaps{\cstates}
      }
      \infer1{\vdash \rtriple{\wprecond{S_1}{Q_1}\boxplus \wprecond{S_2}{Q_2}}{S_1}{S_2}{
          Q_1 \boxplus Q_2}}
    \end{prooftree}
  $\quad
    \RuleLabel{wp}
 \caption{Core proof rules}
  \label{fig:prob:rules}
\end{figure}

\subparagraph*{Lifting.}
An alternative is to introduce a lifting rule that connects unary reasoning to relational reasoning. That is, we can consider a unary judgment of the form  $\triple{P}{S}{Q}$, where $P,Q\in\bndmaps{\cstates}$. A unary judgment $\triple{P}{S}{Q}$ is valid
iff for every $s\in\cstates$, $\E_{\sem{S}s}[Q] \leq P(s)$. Then, one can add a lifting rule \RuleRef{lift} that combines two unary judgments into a relational one, and a unary rule \RuleRef{wp-u} for weakest pre-conditions:

\begin{center}
  \begin{minipage}{0.4\textwidth}
$\begin{prooftree}
      \hypo{
    Q\in\bndmaps{\cstates}
      }
      \infer1{\vdash \triple{\wprecond{S}{Q}}{S}{Q}}
    \end{prooftree}$\quad \RuleLabel{wp-u}
\end{minipage}
\hfill
\begin{minipage}{0.55\textwidth}
$\begin{prooftree}
      \hypo{\vdash \triple{P_1}{S_1}{Q_1}}
     \hypo{\vdash \triple{P_2}{S_2}{Q_2}}
      \infer2{\vdash \rtriple{P_1 \boxplus P_2}{S_1} {S_2}{Q_1 \boxplus Q_2}}
    \end{prooftree}
$\quad \RuleLabel{lift}
\end{minipage}
\end{center}

\subparagraph*{One-sided and two-sided rules.}
In practice, relational logics do not have \RuleRef{wp} as a core rule. Rather, they rely on two-sided rules, in which the two programs have the same top-level construct and execute in lock-step, and on one-sided rules, in which one program uses a given top-level construct, and the other one is either arbitrary or $\skp$. One recovers completeness by proving by induction on the structure of the program that the \RuleRef{wp} rule is derivable from the one-sided rules. The one-sided and two-sided rules of our logic are standard~\cite{erhl}.

\newcommand{\smalleqn}[1]{\mbox{$#1$}}

\subparagraph*{Running example (probabilistic case).}
For every $\phi\subseteq\cstates\times\cstates$ and $Q\in\bbmaps{\cstates\times\cstates}$, let 
$\phi \mid Q \in \bbmaps{\cstates\times\cstates}$
be defined by the clause:
$$(\phi \mid Q)(s,s') = \left\{\begin{array}{ll} Q(s,s') & \mbox{if }
(s,s')\in \phi \\ +\infty & \mbox{otherwise.} \end{array}\right.$$
Then, $\vdash \rtriple{ \top \mid 0 }{ \textsc{Bern}}{\textsc{Unif}^2}{x_1 = x_2\land x'_2 \mid 0}$ (see \Cref{fig:placeholder}), asserting that two programs yield the same distributions on $x_1$ and $x_2\land x'_2$.
The two programs are not aligned w.r.t.\, sampling: \textsc{Bern} performs one sampling and \textsc{Unif}$^2$ performs two samplings. For such examples, the \RuleRef{duality} rule is required. 
The derivation applies \RuleRef{conseq}, \RuleRef{wp}, and \RuleRef{duality}, reducing the goal to the entailment
\begin{align*}
0 \sqsupseteq\ & \wprecond{\textsc{Bern}}{f_1}\boxplus \wprecond{\textsc{Unif}^2}{f_2}
\end{align*}
for every $f_1, f_2$ such that $f_1\boxplus f_2\sqsubseteq x_1 = x_2\land x'_2 \mid 0$, i.e., $f_1(x_1) + f_2(x_2,x_2') \le 0$ if $x_1 = x_2\land x'_2$. Plugging in the concrete weakest preconditions, we conclude since the right-hand side simplifies to
$$ \lambda x_1.\big(\smalleqn{\frac{3}{4}}f_1(0) + \smalleqn{\frac{1}{4}}f_1(1)\big)\boxplus 
     \lambda x_2x_2'.\big(\smalleqn{\frac{1}{4}\sum_{ij}f_2(i,j)}\big) \le 0.$$

\begin{figure*}
    \centering
    \small

    \[
    \begin{array}{@{}c@{\qquad}c@{\qquad}c@{}}
    \begin{alignedat}{2}
        \textsc{Bern} \triangleq\quad
        & \sample{x_1}{\mathsf{Bern}\!\left(\frac{1}{4}\right)}
    \end{alignedat}
    &
    \begin{alignedat}{2}
        \textsc{Unif}^2 \triangleq\quad
        & \sample{x_2}{\{0,1\}}; \\
        & \sample{x'_2}{\{0,1\}}
    \end{alignedat}
    &
    \begin{alignedat}{2}
        \textsc{QBern} \triangleq\quad
        & q_1 \coloneqq \frac{\sqrt{3}}{2}|0\>+\frac{1}{2}|1\>; \\
        & \ifb\ (M[q_1]=1)\ \ithen\ \skp;
    \end{alignedat}
    \end{array}
    \]

    \vspace{-0.4em}

    \[
    \begin{array}{@{}c@{\qquad\qquad}c@{}}
    \begin{alignedat}{2}
        \textsc{QUnif}^2 \triangleq\quad
        & q_2 \coloneqq \frac{1}{\sqrt{2}}(|0\>+|1\>); \\
        & \ifb\ (M[q_2]=1)\ \ithen \\
        & \quad q_2' \coloneqq \frac{1}{\sqrt{2}}(|0\>+|1\>); \\
        & \quad \ifb\ (M[q_2']=0)\ \ithen\ q_2 \coloneqq |0\>
    \end{alignedat}
    &
    \begin{alignedat}{2}
        \textsc{CQUnif}^2 \triangleq\quad
        & q_2 \coloneqq \frac{1}{\sqrt{2}}(|0\>+|1\>); \\
        & x_2 \gets \imeas\; M[q_2]; \\
        & q_2' \coloneqq \frac{1}{\sqrt{2}}(|0\>+|1\>); \\
        & x_2' \gets \imeas\; M[q_2'];
    \end{alignedat}
    \end{array}
    \]

    \caption{Bernoulli sampling and uniform boolean sampling, and their quantum analogue and classical-quantum analogue. Here, we take $M$ as the standard computational basis measurement.}
    \label{fig:placeholder}
\end{figure*}

\subparagraph*{Comparison with ERHL~\cite{erhl}.}
Our logic is closely related to ERHL~\cite{erhl}, although here are some differences between the two proof systems. First, ERHL does not assume programs to be almost-surely terminating (AST); rather, AST assumptions are added as side-conditions as required for the soundness of the rules. Note that~\cite{erhl} uses a weaker notion of coupling, called $\star$ or partial couplings, to accommodate reasoning about programs that do not have the same probability of termination. Second, ERHL only considers non-negative assertions. As a consequence, the duality theorem needs to be adapted so that $\dualset{Q}$ only contains positive functions---this can be achieved with a little bookkeeping, as shown in~\cite{qotl}. However, note that the duality rule from~\cite{qotl} is restricted to bounded post-conditions. \Cref{thm:prob:complete} shows that this restriction can be lifted by invoking a
more general duality theorem.

\section{Finite-Dimensional Quantum Programs}
\label{sec:finite quantum}
This section continues to set the stage for the main results by showing how the duality theorem is used for reasoning about finite-dimensional quantum programs, that is, programs operating on a finite-dimensional Hilbert space such as those of qubits.
In this section, we work exclusively with bounded assertions, which ensures that all expectations remain finite.

\subparagraph*{Preliminaries.}

Let $\mathcal{H}$ be a finite-dimensional complex Hilbert space with inner product $\ave{\cdot,\cdot}$ and induced norm $\|\u\|\triangleq\sqrt{\ave{\u,\u}}$. Let $\linearmaps{\cH}$ denote the set of linear operators on $\mathcal{H}$. We define the operator norm as $\|A\|\triangleq \sup_{\|u\|=1} \|Au\|$ and the trace as $\tr(A)\triangleq \sum_{\u} \ave{\u,A\u}$ over an orthonormal basis $\{\u\}$ of $\cH$. For composite systems, the partial trace $\tr_1: \linearmaps{\cH_1 \otimes \cH_2} \to \linearmaps{\cH_2}$ is the unique linear map satisfying $\tr_1(A \otimes B) = \tr(A) B$ for all $A \in \linearmaps{\cH_1}$ and $B \in \linearmaps{\cH_2}$, with $\tr_2$ defined similarly. Intuitively, the partial trace returns the marginal operator of a subsystem.

$A\in \linearmaps{\cH}$ is called self-adjoint if $\ave{\u,A\u} = \ave{A\u,\u}$ for all $\u\in \cH$.
We reload the notation $\qbndmaps{\cH}$ for the set of \emph{self-adjoint} operators on $\mathcal{H}$\footnote{Since $\cH$ is finite-dimensional, any $A \in \linearmaps{\cH}$ is bounded.}. We use $\sqsubseteq$ to denote the L\"{o}wner order on $\qbndmaps{\cH}$, i.e., for $A, B\in \qbndmaps{\cH}$, $A\sqsubseteq B$ iff $B-A$ is positive semidefinite.
For $A\in \qbndmaps{\cH_1}$ and $B\in \qbndmaps{\cH_2}$, we define $A\boxplus B\in \qbndmaps{\cH_1\otimes \cH_2}$ as  $A\boxplus B \triangleq A\otimes I + I\otimes B$. 

We use $\qstate{\cH}$ to denote the set of density operators on $\cH$, i.e., positive semidefinite operators with trace $1$,
and $\qpstate{\cH}$ for partial density operators with trace $\leq 1$.
Moreover, we write $\couplings{\rho_1}{\rho_2}\subseteq \qstate{\mathcal{H}_1\otimes \mathcal{H}_2}$
for the set of couplings of $\rho_1\in\qstate{\mathcal{H}_1}$ and $\rho_2\in\qstate{\mathcal{H}_2}$,
i.e., all $\rho$ such that $\tr_2(\rho)=\rho_1$ and $\tr_1(\rho)=\rho_2$.

\subparagraph*{Duality theorem.}
Quantum optimal transport~\cite{jan24optimal} aims to minimize the transport cost between two density operators. 
The following theorem generalizes \Cref{thm:kr-duality} to finite-dimensional quantum systems with bounded cost. 
\begin{theorem}[Kantorovich-Rubinstein Duality Theory for Finite Dimensional Quantum Systems~\cite{qotl}]%
\label{thm:quantum-Kantorovich-duality-finite-dim}
     Let $\cH_1$ and $\cH_2$ be two finite-dimensional Hilbert spaces, 
    $\rho_1\in \qstate{\cH_1}$, $\rho_2\in \qstate{\cH_2}$ be two density operators, and $Q\in \qbndmaps{\cH_1\otimes \cH_2}$. 
    Let $\dualset{Q}\subseteq \qbndmaps{\cH_1}\times \qbndmaps{\cH_2}$ such that $(Q_1, Q_2)\in \dualset{Q}$ iff $ Q_1\boxplus Q_2 \sqsubseteq Q$.
    Then, 
    \begin{equation*}
      \inf_{\rho\in\couplings{\rho_1}{\rho_2}} \tr (Q\rho) 
    = \sup_{(Q_1, Q_2)\in \dualset{Q}} \tr(Q_1\rho_1) + \tr(Q_2\rho_2).
    \end{equation*}
\end{theorem}

This duality theorem allows us to reason about couplings without explicitly constructing joint states, and thus serves as the semantic justification of the \RuleRef{duality} rule in our relational program logic.

\subparagraph*{Programs and Assertions.}
Programs are written in the \textsf{qWhile} language, whose syntax is given in \Cref{fig:syntax}. 
Programs operate 
over states, where the set of states is defined 
as $\pstates{\mathcal{H}}$, and
$\cH =\bigotimes_{q\in \qvars} \cH_q (\mathcal{V}_{q})$. Here, $\qvars$ denotes a finite set of quantum variables, $\mathcal{V}_q$ is the set of possible classical values for $q$, and $\cH(\mathcal{V})$ is the Hilbert space spanned by the orthonormal basis $\{|v\>: v \in \mathcal{V}\}$. In this section, we assume that all $\mathcal{V}_q$ are finite; it follows that $\mathcal{H}$ is finite-dimensional.

Each program $S$ has an interpretation
$\sem{S}\in \mathcal{QO}(\mathcal{H})$, i.e., a completely positive trace non-increasing map from $\qpstate{\mathcal{H}}$ to $\qpstate{\mathcal{H}}$. We say that $S$ is almost surely terminating, or AST, if $\sem{S}$ is trace-preserving, i.e., $\sem{S}$ : $\qstate{\mathcal{H}} \rightarrow \qstate{\mathcal{H}}$.

Relational assertions are elements of $\qbndmaps{\mathcal{H}\otimes\mathcal{H}}$, i.e.\, the self-adjoint (aka, Hermitian, for finite-dimensional space) operators on $\mathcal{H}\otimes\mathcal{H}$. Relational assertions are naturally ordered by L\"{o}wner order $\sqsubseteq$.
Note that our choice of assertions is less standard than using bounded positive operators.
However, the two notions are equivalent for AST programs, since any assertion can be shifted to a positive one, and the expectation of the offset behaves uniformly in pre- and post-expectations.

\subparagraph*{Program logic.}
Judgments are of the form $\rtriple{P}{S_1}{S_2}{Q}$, where $S_1,S_2$ are AST programs and 
$P,Q\in\qbndmaps{\mathcal{H}\otimes\mathcal{H}}$.
We say that $\rtriple{P}{S_1}{S_2}{Q}$ is valid, written $\models\rtriple{P}{S_1}{S_2}{Q}$, 
 if for every $\rho \in \mathcal{D}^1(\mathcal{H} \otimes \mathcal{H})$, 
  there exists a coupling $\sigma\in \couplings{\sem{S_1}(\tr_2(\rho))}{\sem{S_2}(\tr_1(\rho))}$ such that
  $\tr(P \rho) \geq \tr(Q \sigma)$.

The core rules are formally the same as those in \Cref{fig:prob:rules}, with the notations $\sqsubseteq$, $\dualset{}$ and $\boxplus$ redefined as above. The only exception is the weakest precondition $\wprecond{S}{Q} \in \qbndmaps{\mathcal{H}}$, which must be redefined independently and satisfies $\tr(\wprecond{S}{Q}\rho) = \tr(Q\sem{S}(\rho))$ for all $\rho\in\qstate{\mathcal{H}}$. 
\begin{theorem}[Soundness and Completeness of core rules for qWhile, cf.~\cite{qotl}]\label{thm:qotl-complete}
A judgment is valid iff it can be derived with the rules \RuleRef{duality},
\RuleRef{conseq}, and \RuleRef{wp}.
\end{theorem}
Similar to the probabilistic setting, the \RuleRef{wp} rule reduces split post-condition to unary weakest-precondition reasoning, and the \RuleRef{duality} rule converts the existence of a coupling satisfying $Q$ into universally quantified split post-condition justified by
\Cref{thm:quantum-Kantorovich-duality-finite-dim}.

\subparagraph*{Comparison with QOTL~\cite{qotl}.}

QOTL uses extended self-adjoint operators, namely self-adjoint operators with a possible
$+\infty$-eigenspace, as assertions to capture also qualitative properties of quantum states. However, its duality rule is restricted to bounded assertions, and consequently the  completeness result coincides with \Cref{thm:qotl-complete}.
A key motivation for introducing such extended assertions in QOTL is the development of useful two-sided proof rules based on measurement conditions.
In this setting, extended self-adjoint operators are used as predicates: the $+\infty$-eigenspace serves as the qualitative component that enforces lockstep execution of two programs, while the finite part forms the quantitative component describing the properties of interest.
Finally, QOTL also introduces partial couplings rather than full couplings, which allows to establish sound but not complete one-sided and two-sided rules that do not always require programs to be AST.

\section{Linear Relations}\label{Unbounded_Section}

This section develops the analytic foundations required for extended operators on infinite-dimensional quantum systems.
In particular, we introduce a framework that supports the approximation and convergence arguments needed for later duality theorems.
From now on, we assume that the Hilbert space $\cH$ is separable, i.e., it admits a countable orthonormal basis.

\subsection{Challenges with Infinite Dimensions}

We begin by recalling basic notions that remain consistent with the finite-dimensional setting.
For bounded operators (with respect to operator norm), self-adjointness, positivity, and hence the L\"owner order are defined in the same way.
We write $\qbndmaps{\cH}$ for the set of self-adjoint bounded operators.
Among bounded operators, trace-class operators---those with a finite, basis-independent trace---closely parallel the finite-dimensional case; for example, the partial trace is well defined on them.
(Partial) density operators are positive trace-class operators and are therefore defined in the same way.

\subparagraph*{Infinity.}
There are naturally two distinct sources of ``infinity'' that arise in the infinite-dimensional
setting:
\begin{itemize}
    \item \emph{Infinity from approximation (unboundedness).}
    This occurs when a sequence of bounded values grows without bound.
    For example, the operator $A = \sum_n n \ket{n}\!\bra{n}$ has countably many eigenvalues, but an infinite operator norm.
    This is the standard notion of an \emph{unbounded operator} in functional analysis, characterized by a dense domain.
    \item \emph{Infinity from values (singularity).}
    This represents a ``hard'' constraint, corresponding to the value $+\infty$.
    For instance, an assertion requiring a state to be orthogonal to a subspace $S$ assigns infinite cost to any component supported on $S$.
    Such constraints correspond to operators with an eigenspace associated with the eigenvalue $+\infty$.
\end{itemize}

Unboundedness captures growing quantitative costs, such as the expected distance in an infinite quantum walk, while $+\infty$-valued assertions encode qualitative constraints~\cite{qotl}. A unified treatment of these two forms of infinity is essential to achieve the desired expressiveness and completeness of program logics. This unified treatment is rather simple in the probabilistic setting. Intuitively,  the two sources of infinity are  unified by pointwise lifting the codomain of assertions to the extended reals: for $f,g\in\bbmaps{X}$, $(f+g)(x)$ is finite if and only if both $f(x)$ and $g(x)$ are
finite. However, in the quantum setting, this simple picture breaks down as explained below.

\subparagraph*{Challenge 1: failure of operator sums.}
The naive analogue of pointwise addition is the operator sum
\[
(A+B)\u \triangleq A\u +B\u,
\]
which is defined only when both $A\u$ and $B\u$ exist.
However, for self-adjoint operators, addition need not preserve self-adjointness~\cite{reed1972methods}, and the spectral decomposition theorem may fail.


This suggests shifting attention to quadratic forms, which directly compute the expectations: for bounded $A$, $\mathfrak{t}_A(\u)=\langle\u,A\u\rangle$.
At this level, addition is pointwise defined by
$\mathfrak{t}_{A+B}=\mathfrak{t}_A+\mathfrak{t}_B$, 
which remains well-behaved even for infinite-valued quadratic forms.
Notably, $\mathfrak{t}_A(\u)$ may be finite even when $A\u$ is undefined.

\subparagraph*{Challenge 2: representation of quadratic forms.}
While quadratic forms correctly capture both unbounded growth and $+\infty$-valued singularity,
they are abstract and lack a direct algebraic representation.
Standard unbounded operators must be densely defined and therefore cannot encode value infinity.
QOTL~\cite{qotl} addresses this by separating finite operators from infinite projections, but this
leads to ad hoc definitions of algebraic operations, where interactions between finite and infinite
parts must be specified manually.

To reconcile these issues, we adopt the formalism of \emph{linear relations}, which represent
operators as linear subspaces of $\cH\times\cH$.
This framework uniformly accommodates unbounded growth and value infinity; moreover, self-adjoint linear relations admit a spectral theorem with values in $\ereals$, providing a principled treatment of infinite spectral parts.
This formalism serves primarily as a unifying representation rather than introducing additional semantic structure.

\subsection{Definitions and Basic Properties}

\subparagraph*{Preliminaries}
A linear relation is a subspace (a.k.a. graph; not necessarily closed) of $\mathcal{H} \times \mathcal{H}$.
The adjoint of a linear relation $A$, denoted $A^*$, is the subspace:
\[
    A^* \triangleq \{ (\u, \v) \in \mathcal{H} \times \mathcal{H} \mid \forall\, (\u', \v') \in A, \langle \v, \u' \rangle = \langle \u, \v' \rangle \}.
\]
A linear relation $A$ is \textit{self-adjoint} if $A = A^*$. It is known that self-adjointness automatically ensures the relation is closed \cite[Proposition 1.3.2]{behrndt2020boundary}.
We say a self-adjoint linear relation $A$ is bounded below by $m \in \mathbb{R}$ if $\<\u,\v\>\geq m\<\u,\u\>$ holds for any pair $(\u,\v)\in A$. If $m\geq 0$, then $A$ is called positive, and it admits the square root, denoted $\sqrt{A}$, the unique positive linear relation such that
$A = \{(\u,\v)\mid (\u,\w)\in \sqrt{A}\text{~and~}(\w,\v)\in \sqrt{A}\}$~\cite[Theorem 1.5.9]{behrndt2020boundary}.
For a linear relation $A$ bounded below by $m$, we can \emph{shift} $A$ to a positive linear relation: $\shift{A}{m}\triangleq\{(\u,\v-m\u)\mid (\u,\v)\in A\}$.
In the following, we only focus on bounded-below self-adjoint linear relations, or simply linear relations, denoted by $\bbmaps{\cH}$. 

\newcommand{\cl}[1]{%
  \overline{\smash{#1}\vphantom{i}}\vphantom{#1}%
}

Now, we explain how linear relations capture infinity.
Given $A\in\bbmaps{\cH}$ bounded below by $m$, let $\tdom(A)$ be the finite part, known as the quadratic domain, and $\mul(A)$ be the multivalued part:
\begin{align*}
    \tdom(A)& \triangleq \{\u\in\cH:~\exists\, \v\in\cH~\text{s.t.~}(\u,\v)\in \sqrt{\shift{A}{m}}\};\\
    \mul(A)& \triangleq \{\v\in\cH:~(0,\v)\in A\}.
\end{align*}
$\tdom(A)$ is the subspace of states with finite expectations, but is not necessarily closed under the usual norm (i.e., the norm induced by the inner product). Infinity is captured in two ways: the unboundedness is encoded in the difference $\cl{\tdom(A)} \backslash \tdom(A)$ where $\cl{\,\cdot\,}$ denotes the closure, while the singularity corresponds to $\mul(A)$ which is a closed subspace.
According to \cite[Theorem 1.5.1]{behrndt2020boundary}, the domain is orthogonal to the multivalued part, i.e., $\tdom(A) \perp \mul(A)$, and the orthogonal decomposition $\cH = \cl{\tdom(A)} \oplus \mul(A)$ holds. This allows us to associate $A$ (bounded below by $m$) with a semi-bounded closed quadratic form (or simply quadratic form) $\mathfrak{t}_{A}$ :
\begin{equation*}
    \mathfrak{t}_{A}(\u) \triangleq \left\{
    \begin{aligned}
    &\<\v,\v\>+m\<\u,\u\>,&& \text{if\,} \left(\u,\v\right)\in \sqrt{\shift{A}{m}}\text{\,\&\,}\v\in\mul(A)^{\perp},\\
    &+\infty,&&\text{otherwise}
    \end{aligned}
    \right.
\end{equation*}
which is automatically well-defined due to the uniqueness of such $\v$ if $\u \in \tdom(A)$.
The quadratic form treats the two sources of infinity uniformly, by assigning the value $+\infty$ to both cases, while reserving finite values strictly for $\tdom(A)$.
Kato's First Representation Theorem~\cite{kato2013perturbation,behrndt2020boundary} legitimizes this unification by establishing a one-to-one correspondence between semi-bounded closed quadratic forms and linear relations, thereby giving the algebraic operations of linear relations by point-wise lifting on their associated quadratic forms. This additionally suggests treating both sources of infinity (and in fact, infinite expectation has only these two sources) uniformly and indistinguishably, from the perspective of  expectation.

The spectral theorem is another key property that allows a linear relation to be identified with a spectral measure on the extended reals. This viewpoint is particularly convenient for defining functional calculus, such as giving the explicit form of square roots, and for formulating convergence theorems.
\begin{proposition}[Spectral theorem, \cite{schmudgen2012unbounded,behrndt2020boundary}]\label{prop:spectral-thm}
    There is a one-to-one correspondence between a linear relation $A\in\bbmaps{\cH}$ bounded below by $m$ and spectral measures $E_A$ on the interval $[m,+\infty]$ in the extended real line. The correspondence is established via the quadratic form, i.e., for any $\u \in \cH$, its quadratic form value is given simply by the moment of the spectral measure:
    $$\mathfrak{t}_A(\u) = \int_{[m,+\infty]} \lambda \, d\<\u, E_A(\lambda)\u\>.$$
\end{proposition}

The spectral theorem provides an alternative way to understand infinity: the projection onto the infinite eigenspace is exactly the projection onto the multivalued part of $A$, i.e., $E_A(\{+\infty\}) = \mul(A)$; the integral on the finite part, $\int_{[m,+\infty)} \lambda \, d\<\u, E_A(\lambda)\u\>$ recovers the quadratic form of standard unbounded operators, i.e., the finite part together with unbounded growth.

With all these basic ingredients, we are ready to extend some commonly used operations for linear relations. The first one is the extended L\"owner order, which is ordered by expectations ranging over all states:

\begin{definition}[Extended L\"owner Order]
    \label{def:extended_lowner_order_Text}
    Let $A, B$ be two linear relations. We define the partial order $A \sqsubseteq B$ via their associated quadratic forms:
    \[ A \sqsubseteq B \iff \mathfrak{t}_A(\u) \le \mathfrak{t}_B(\u), \quad \forall\, \u \in \mathcal{H}. \]
    This inequality holds pointwise on the extended real line $\mathbb{R}\cup\{+\infty\}$.
\end{definition}
Note that $A \sqsubseteq B$ implies the inclusion of the domains $\tdom(B) \subseteq \tdom(A)$. 
Recalling \Cref{prop:spectral-thm}, we can also define $A \sqsubseteq B$ based on the spectral measures.
Furthermore, we can show that $(\bbmaps{\cH},\sqsubseteq)$ indeed forms an $\omega$-complete partial order.

Next we turn to the expectation (i.e., trace) of linear relations acting on density operators, which is defined via the standard measure-theoretic formulations of quantum mechanics (see e.g., \cite{holevo2011probabilistic, reed1972methods}):
\begin{definition}[Extended Trace]\label{ExTr}
    We define the extended trace $\Tr:\bbmaps{\cH}\times\pstates{\cH}\rightarrow \bR\cup\{+\infty\}$ as:


    \[
    \Tr(A\rho)= \int_{[m,+\infty]} \lambda \, d\tr(E_A(\lambda)\rho),
    \]
    if $A$ is lower bounded by $m\in\mathbb{R}$. The
    $\tr(\cdot)$ on the right-hand side denotes the standard trace, noting that $E_A(\lambda)\rho$ is in fact a trace-class operator.
\end{definition}
The extended trace can also be equivalently defined from the quadratic form, given the spectral decomposition of $\rho$. The extended trace remains finite if and only if both of the following conditions hold:
\begin{enumerate}
     \item The support of $\rho$ is in $\tdom(A)$, and
    \item The sum of the quadratic form values converges. 
\end{enumerate}
Intuitively, this formalizes a \textit{support check}: any overlap with the ``infinite energy'' subspace results in infinite cost, while for valid states, it reduces to the standard expectation value.


The last construction is the finite truncation of a linear relation, which provides a bounded approximation of every linear relation, including those that are unbounded or with value infinity.
Suppose $A=\int_{[m,+\infty]}\lambda  dE_{A}(\lambda)$ and $n\in\bN$ with $n\ge m$. For any $\u\in\cH$, we denote
\[
\trunc{A}{n}\u=\int_{m}^{n} \lambda  dE_{A}(\lambda)\u+\int_{[n,+\infty]} n dE_{A}(\lambda)\u.
\]
Note that on every $E_A((k,+\infty])$ we define $\trunc{A}{n}$ to have spectrum $n$ if $n\le k$, thus $\trunc{A}{n}$ is a bounded self-adjoint \textit{operator}, i.e., its operator norm is bounded by $\max(|m|,n)$. 
Note that our truncation method differs from that of physicists, which focuses on unbounded growth rather than
value singularities and typically yields a finite-dimensional subspace.
In contrast, we only require the spectrum of the truncated operator to be a bounded set.

\subparagraph*{Degenerate cases of linear relations.}
As a consistency check, we show how the notion of linear relations introduced above degenerates into
standard operator-based semantics in familiar cases.

\begin{enumerate}
    \item \emph{Bounded operators.}
    If $A$ is a bounded self-adjoint operator, the associated linear relation is its graph:
    $(\u,\v)\in A$ if and only if $\v = A\u$.
    In this case, the spectrum is bounded and contained in $\mathbb{R}$,
    the extended trace coincides with the standard trace $\tr(A\rho)$, and no domain issues arise.

    \item \emph{Unbounded operators.}
    If $A$ is a (densely defined) unbounded self-adjoint operator, the corresponding linear relation is again its graph:
    $(\u,\v)\in A$ if and only if $\v = A\u$,
    where only those $\u$ that belong to the operator domain, i.e., those for which $A\u\in\cH$ is well defined, are collected.
    Thus, unbounded operators correspond exactly to single-valued linear relations with trivial multivalued part.

    \item \emph{Value singularities (hard constraints).}
    In QOTL~\cite{qotl}, an extended operator is represented as a bounded operator
    $A_{\mathrm{fin}}$ together with a projection (or closed subspace) $\mathsf{P}_\infty$ corresponding to
    the value infinity.
    This semantics is naturally captured by a multivalued linear relation defined by
    $(\u,\v)\in A$ if and only if $\u\in \mathsf{P}_{\infty}^{\perp}$ and
    $\v = A_{\mathrm{fin}}\u + \w$ and $\w\in \mathsf{P}_{\infty}$ (i.e., allowing direct sum of $A_{\mathrm{fin}}\u$ and $\mathsf{P}_{\infty}$).
    The multivalued part of the relation is precisely $\mathsf{P}_{\infty}$, corresponding to the
    $+\infty$-eigenspace.
\end{enumerate}

\subparagraph*{Guarded linear relations.}
Finally, as an example, we introduce the notion of \emph{guarded linear relations}, which serves as the quantum analogue to guarded assertions in probabilistic programming.
In the quantum setting, a ``guard'' is represented by a projection $\mathsf{P}$. Given a linear relation $A$ (as an observable), we wish to define the guarded linear relation $P\mid A$ that agrees with $A$ on the closed subspace $\mathsf{P}$, but imposes infinity on any state overlapping with $\mathsf{P}^{\perp}$.
Thus, 
$$\mathsf{P}\mid A=\{(\u, \v + \w)\mid \u\in \mathsf{P},\ (\u,\v)\in A,\ \text{and}\ \w\in P^\bot\}.$$
We can check its quadratic form exactly satisfies what we want:
\begin{equation*}
    \mathfrak{t}_{\mathsf{P} \mid A}(\u) = \left\{
    \begin{aligned}
    &\mathfrak{t}_A(\u),&&\text{if\,}\u\in \mathsf{P};\\
    &+\infty,&&\text{otherwise}.
    \end{aligned}\right.
\end{equation*}
Similarly, its extended trace is consistent with $A$ if $\rho$ ``satisfies'' $P$:
\[
\Tr((\mathsf{P}\mid A) \rho) = \left\{
    \begin{aligned}
    &\Tr(A\rho),&&\text{if\,}\ker(\rho)^{\perp}\subseteq \mathsf{P};\\
    &+\infty,&&\text{otherwise}.
    \end{aligned}\right.
\]

\subsection{Convergence Theorems}

Now we present convergence theorems for linear relations using truncations, which will be used to prove the unbounded duality theorem (\Cref{thm:quantum-Kantorovich-duality-unbounded}).
Specifically, these convergence theorems are instrumental to derive a duality theorem for unbounded quantum costs by applying a truncation-and-limit argument based on the bounded quantum duality theorem.

\begin{theorem}[Convergence Theorems]
\label{thm:main-convergence-for-trunc}
    Let $A$ be a self-adjoint, bounded below linear relation on a Hilbert space $\mathcal{H}$, $A_n=\trunc{A}{n}$, and $\rho\in\pstates{\cH}$.
    We have
    \begin{itemize}
        \item  (Monotone Convergence Theorem for Truncations) 
     $
		\lim\limits_{n\rightarrow\infty} \tr(A_n\rho)=\Tr(A\rho).
    $
    \item (Generalized Quantum Fatou's Lemma for Truncations) Suppose $\rho_n\in\pstates{\cH}$ converges to $\rho$ in trace norm, then $
		\liminf\limits_{n\rightarrow\infty} \tr(A_n\rho_n)\ge\Tr(A\rho).
    $
    \item  (Lower Semi-continuity of Expectations) If $\rho_n\to\rho$ under the trace norm, then \[
	\liminf_{n\rightarrow\infty}\Tr(A\rho_n)\geq\Tr(A\rho).
    \]
    Moreover, if $A$ is bounded, then the equality holds.
    \end{itemize}
    Note that $A_n$ is always bounded; therefore, $A_n\rho$ is trace-class, and the standard trace $\tr$ can be used instead of $\Tr$. 
\end{theorem}

These constructions provide the analytic basis for the truncation and convergence arguments used in the subsequent duality results.

\section{Infinite-Dimensional Duality Theorem}\label{sec:duality}

\begin{figure*}[t]
\centering
\resizebox{\textwidth}{!}{
\begin{tikzpicture}[>=Stealth,x=1cm,y=1cm,font=\small]

\begin{scope}[shift={(0,0)}]

  \draw[->] (0,0) -- (3.2,0) node[right] {$S_1$};
  \draw[->] (0,0) -- (0,-2.4) node[below] {$S_2$};

  \draw (0,-2.0) rectangle (3,0);
  \node at (1.5,-1.1) {$\mu$};

  \fill[gray!20] (0,0.15) rectangle (3,0.4);
  \node[above] at (1.5,0.4) {$\mu_1$};
  \fill[gray!20] (-0.35,-2.0) rectangle (-0.1,-0.2);
  \node[left] at (-0.35,-1.1) {$\mu_2$};

  \draw[->] (1.5,-2.3) -- (1.5,-2.9)
     node[midway,right] {Truncation};

  \draw (0,-5.5) rectangle (3,-3.2);

  \fill[gray!20] (0.6,-4.7) rectangle (2.4,-3.6);
  \draw[thick]   (0.6,-4.7) rectangle (2.4,-3.6);

  \draw[decorate,decoration={brace,mirror,raise=3pt}]
    (0.6,-4.7) -- (2.4,-4.7)
    node[midway,below=4pt] {$S_1'$};
  \draw[decorate,decoration={brace,raise=3pt}]
    (0.6,-4.7) -- (0.6,-3.6)
    node[midway,left=4pt] {$S_2'$};

  \node[right, align=left] at (0.7,-4.1)
    {$\mu(S_1'\!\times\! S_2')$\\
    $\ge 1-2\delta$};

\end{scope}

\begin{scope}[shift={(6,-1.5)}]

  \coordinate (Ptop) at (0,0.9);
  \coordinate (Pbot) at (0,-3.4);

  \coordinate (Dtop) at (8.2,0.9);
  \coordinate (Dbot) at (8.2,-3.4);

  \coordinate (Mid)  at (4.1,-2.35);

  \node[align=center] at (Ptop) {primal\\[0.3em]
    $\displaystyle
      \inf_{\mu}\,\mathbb{E}_{\mu}[C]
    $
  };

  \node[align=center] at (Pbot) {finite primal\\[0.3em]
    $\displaystyle
      \inf_{\mu'}\,\mathbb{E}_{\mu'}\!\bigl[C|_{S_1'\times S_2'}\bigr]
    $
  };

  \draw[->] (0,0.25) -- (0,-2.75);
  \node[left=10pt, align=center] at (0,-1.25) {Truncation};
  \node[right=10pt, align=left] at (0,-1.25) {Error:\\
    $2\delta \lVert C\rVert_{\infty}$\\
    (additive)};

  \node[align=center] at (Dtop) {dual\\[0.3em]
    $\displaystyle
      \sup_{f_1,f_2}
      \left(\mathbb{E}_{\mu_1}[f_1]+\mathbb{E}_{\mu_2}[f_2]\right)
    $
  };

  \node[align=center] at (Dbot) {finite dual\\[0.3em]
    $\displaystyle
      \sup_{f_1',f_2'}
      \left(\mathbb{E}_{\mu_1'}[f_1']+\mathbb{E}_{\mu_2'}[f_2']\right)
    $
  };

  \draw[->] (8.2,-2.75) -- (8.2,0.25);
  \node[left=12pt, align=center] at (8.2,-1.25)
    {``pointwise partial\\ minimization''\\[0.2em]
     \textbf{Obstacle for}\\
     \textbf{quantum case}};
  \node[right=10pt, align=left] at (8.2,-1.25)
    {Error:\\
     $(1-2\delta)$ multiplicative\\
     $O(\delta \lVert C\rVert_{\infty})$ additive};

\node[align=center] (FD) at ($(Pbot)!0.5!(Dbot)+(0,0.45)$) {finite duality};
\draw ($(FD.south)+(-0.8,-0.08)$) -- ($(FD.south)+(0.8,-0.08)$);
\draw ($(FD.south)+(-0.8,-0.20)$) -- ($(FD.south)+(0.8,-0.20)$);

\end{scope}
\end{tikzpicture}
}
\caption{Proof strategy for the infinite-dimensional duality for bounded costs. Left: coupling $\mu$ on $S_1\times S_2$ and truncation to a large-mass block $S_1'\times S_2'$. Right: primal error from truncation and finite duality obtained by improving admissible dual potentials, with multiplicative error $(1-2\delta)$ and additive $O(\delta)$. 
In the probabilistic setting this improvement uses pointwise partial minimization; in the quantum setting it fails and is replaced by a dimension-independent perturbation bound.
}
\label{fig:OT-duality-proof-sketch}
\end{figure*}

We now present our main duality theorem for infinite-dimensional quantum systems and our general notion of assertions, together with a proof sketch. This result generalizes \Cref{thm:kr-duality,thm:quantum-Kantorovich-duality-finite-dim} to the infinite-dimensional setting. The notions of partial trace and couplings for density operators extend naturally from the finite-dimensional case.

\begin{theorem}[Kantorovich Duality for Infinite-Dimensional Quantum Systems with Bounded by Below Cost]%
    \label{thm:quantum-Kantorovich-duality-unbounded}
     Let $\mathcal{H}_1$ and $\mathcal{H}_2$ be two Hilbert spaces, 
    $\rho_1\in \qstate{\mathcal{H}_1}$, $\rho_2\in \qstate{\mathcal{H}_2}$ be two density operators, and $Q\in \qbbmaps{\cH_1\otimes \cH_2}$ be a bounded-by-below self-adjoint linear relation.
Let $\dualset{Q}\subseteq \qbndmaps{\cH_1}\times \qbndmaps{\cH_2}$ be such that $(Q_1, Q_2)\in \dualset{Q}$ iff $Q_1\boxplus Q_2\sqsubseteq Q$.
    Then, 
    \begin{equation*}
         \inf_{\rho\in\couplings{\rho_1}{\rho_2}} \Tr (Q\rho)
         = \sup_{(Q_1,Q_2)\in \dualset{Q}} \Tr(Q_1\rho_1) + \Tr(Q_2\rho_2).
    \end{equation*}
\end{theorem}



\subparagraph*{Proof strategy.} Proving Theorem~\ref{thm:quantum-Kantorovich-duality-unbounded} requires addressing both infinite dimensionality and unbounded costs. 
In the probabilistic setting, this is proved via a standard strategy by first establishing the infinite-dimensional duality for bounded costs and then applying the convergence theorem to push from bounded to unbounded costs.
The infinite-dimensional duality for bounded costs is further refined via four steps (see~\Cref{fig:OT-duality-proof-sketch}): truncation, finite-dimensional duality, dual-pair improvement, and error control~\cite{Vil03}. A critical technique is pointwise partial minimization in the third step as $\bar{f}_1(x) \triangleq \inf_{y\in S_2'} C(x,y) -f_2'(y)$ which ``improves'' dual pairs and preserves feasibility, which, unfortunately, fails in the quantum setting as explained below:


\subparagraph*{Challenge: Non-linearity prevents pointwise partial minimization.} In the quantum setting, 
dual variables are now self-adjoint operators: they must be \emph{linear} functionals of density operators. The linearity prevents applying an analogue of classical partial minimization: this is because taking an infimum over one subsystem like 
$$f(\rho_1) \triangleq \inf_{\rho_2\in \qstate{\cH_2}} \inf_{\rho\in \couplings{\rho_1}{\rho_2}}\Tr(Q\rho) - \Tr(Q_2\rho_2) $$
is a nonlinear operation that generally fails to yield a valid self-adjoint operator, i.e., there does not exist $Q_1$ such that $f(\rho) = \Tr(Q_1\rho)$ for all $\rho$, on the remaining subsystem. 

\subparagraph*{Solution: A dimension-independent perturbation bound.}
We address this via a core lemma that extracts near-optimal, uniformly norm-controlled dual solutions from the finite-dimensional semidefinite program. The lemma gives a dimension-independent bound that enables lifting these solutions to the infinite-dimensional setting, replacing classical pointwise improvement with an operator-norm-controlled construction.

\subsection{Perturbation Bound}
This subsection is devoted to establish our core technical tool, a dimension-independent perturbation lemma, as stated below. 
\begin{lemma}
\label{lem:perturbation-analysis}
     Let $\mathcal{H}_1$ and $\mathcal{H}_2$ be two finite-dimensional Hilbert spaces, 
    $\rho_1\in \qstate{\mathcal{H}_1}$, $\rho_2\in \qstate{\mathcal{H}_2}$ be two density operators. Let $Q\in \qbndmaps{\mathcal{H}_1\otimes \mathcal{H}_2}$ be a bounded positive operator.
    Let $\dualset{Q}\subseteq \qbndmaps{\cH_1}\times \qbndmaps{\cH_2}$ such that $(A_1, A_2)\in \dualset{Q}$ iff $A_1\boxplus A_2\sqsubseteq Q $.
     We denote $\dopt{Q}\coloneqq\sup_{(A_1, A_2)\in \dualset{Q}} \Tr(A_1\rho_1) + \Tr(A_2\rho_2).$
    Then, for any $\varepsilon >0$, 
    there exist bounded operators $A_1$ and $A_2$ such that $(A_1, A_2)\in \dualset{Q}$, 
    $\tr(A_1\rho_1) + \tr(A_2\rho_2) \ge \dopt{Q} -\varepsilon$,
    and $\max \{\norm{A_1}, \norm{A_2}\}\le \norm{Q} + \frac{2}{\varepsilon}\norm{Q}^2$.
\end{lemma}

\begin{proof}
The proof is purely constructive, which we demonstrate as follows.

\textbf{Step 1:} \emph{Start from an approximate maximizer.} Without loss of generality, we can find self-adjoint operators $B_1$, $B_2$ such that $\tr(B_1\rho_1) + \tr(B_2\rho_2)\ge \dopt{Q}-\frac{\varepsilon}{2}$, with $\lambda_{\max}(B_1) \le \norm{Q}$ and $\lambda_{\max}(B_2) = 0$ by performing the constant shift $B_1\gets B_1- cI$, $B_2 \gets B_2+cI$. 

\textbf{Step 2:} \emph{Clip eigenvalues to obtain a uniform norm bound.} In the following, we choose the basis of $\mathcal{H}_1$ and $\mathcal{H}_2$ such that $B_1$ and $B_2$ are diagonal.
We then define $A_1$ and $A_2$ to be the diagonal matrices on that basis with the following entries:
\begin{equation*}
        (A_1)_{kk}  \triangleq \max\cbra*{-\norm{Q}-\frac{2}{\varepsilon}\norm{Q}^2, (B_1)_{kk}-\frac{\varepsilon}{2} }, \quad 
        (A_2)_{\ell \ell}  \triangleq \max\cbra*{-\norm{Q}-\frac{2}{\varepsilon}\norm{Q}^2, (B_2)_{\ell \ell}}. 
\end{equation*}
Note that our construction always implies $A_1\sqsupseteq B_1 - \frac{\varepsilon}{2}I$ and $A_2\sqsupseteq B_2$.

\textbf{Step 3:} \emph{Control the loss in objective value.} The above construction gives $\tr(A_1\rho_1)\ge \tr(B_1\rho_1) - \frac{\varepsilon}{2}$ and 
$\tr(A_2\rho_2)\ge \tr(B_2\rho_2)$, hence $ \tr(A_1\rho_1) + \tr(A_2\rho_2) \ge \dopt{Q} -\varepsilon$
as desired.

\textbf{Step 4:} \emph{Check feasibility via a Schur complement argument.} The remaining part is to show $A_1\otimes I+ I\otimes A_2\sqsubseteq Q$. 


This is equivalent to show $Q - (A_1\otimes I + I\otimes A_2) \sqsupseteq 0 $. As the latter term $(A_1\otimes I + I\otimes A_2) $ is constructed from $(B_1\otimes I + I\otimes B_2) $ and diagonal (with diagonal entries of the form $(A_1)_{kk}+(A_2)_{\ell \ell}$),
it is natural to prove the positivity using the following proposition of Schur complement.

\begin{lemma}[Schur complement{~\cite[Theorem 1.12]{Horn2005}}]\label{lemm:schur-complement-psd-condition}
Consider a Hermitian block matrix
\begin{align*}
    H = \begin{bmatrix}
    H_{11} & H_{12} \\
    H_{12}^\dagger & H_{22}
    \end{bmatrix}.
\end{align*}
 s.t. $H_{11} \sqsupseteq \alpha I$, $H_{22} \sqsupseteq \gamma I$, and $\norm{H_{12}}^2 \leq \gamma \alpha$, with $\gamma,\alpha>0$. Then $H \sqsupseteq 0$.
\end{lemma}

To use the Schur complement argument, let $\mathcal{J}$ denote the set of indices where we did not clip, i.e.,
$    \mathcal{J} \triangleq \cbra*{(k,\ell)\mid (A_1)_{kk} =  (B_1)_{kk}-\frac{\varepsilon}{2} \text{ and }
        (A_2)_{\ell \ell} = (B_2)_{\ell \ell}}$.

If $(k,\ell) \in \cJ$, then clearly
$     (A_1)_{kk} + (A_2)_{\ell \ell} =  (B_1)_{kk} + (B_2)_{\ell \ell}-\frac{\varepsilon}{2}$
while if $(k,\ell) \not\in \cJ$, it is easy to show
$(A_1)_{kk} + (A_2)_{\ell \ell} \le -\frac{2}{\varepsilon} \norm{Q}^2$.

Sorting the product basis such that the coordinates in~$\cJ$ come last, and writing
\begin{equation*}
        Q = \begin{bmatrix}
        Q_{\cJ^c\cJ^c} & Q_{\cJ^c\cJ} \\
        Q_{\cJ\cJ^c} & Q_{\cJ\cJ}
    \end{bmatrix}, \quad
    B_1 \otimes I + I \otimes B_2 = \begin{bmatrix}
        \hat D_{\cJ^c\cJ^c} & \\
        & \hat D_{\cJ\cJ}
    \end{bmatrix},
\end{equation*}
we see that in order to prove $A_1 \otimes I + I \otimes A_2 \sqsubseteq Q$, it suffices to show the following:
\begin{equation*}
        \begin{bmatrix}
        Q_{\cJ^c\cJ^c} + \frac{2}{\varepsilon} \norm{Q}^2 I &
        Q_{\cJ^c\cJ} \\
        Q_{\cJ\cJ^c} &
        Q_{\cJ\cJ} - \hat D_{\cJ\cJ} + \frac{\varepsilon}{2} I
    \end{bmatrix}
    \sqsupseteq 0.
\end{equation*}
  
We now use \Cref{lemm:schur-complement-psd-condition}.
First we note that the upper left block $ \sqsupseteq \frac{2}{\varepsilon} \norm{Q}^2 I$, because $Q_{\cJ^c\cJ^c} \sqsupseteq 0$.
It is also clear that $\norm{Q_{\cJ^c\cJ}}^2 \leq \norm{Q}^2$.
On the other hand, since $B_1\otimes I + I\otimes B_2 \sqsubseteq Q$, we know that
\begin{align*}
    \begin{bmatrix}
        Q_{\cJ^c\cJ^c} - \hat D_{\cJ^c\cJ^c} &
        Q_{\cJ^c\cJ} \\
        Q_{\cJ\cJ^c} &
        Q_{\cJ\cJ} - \hat D_{\cJ\cJ}
    \end{bmatrix}
    \sqsupseteq 0,
\end{align*}
and hence $Q_{\cJ\cJ} - \hat D_{\cJ\cJ} + \frac{\varepsilon}{2} I\sqsupseteq \frac{\varepsilon}{2} I$.
Now the result follows by an application of \Cref{lemm:schur-complement-psd-condition}, with $\alpha=\frac{2}{\varepsilon}\norm{Q}^2$ and $\gamma=\frac{\varepsilon}{2}$.
\end{proof}

\subsection{\texorpdfstring{Proof Sketch of \Cref{thm:quantum-Kantorovich-duality-unbounded}}{Proof Sketch of Theorem (Quantum Kantorovich Duality Unbounded)}}

We now prove our main \Cref{thm:quantum-Kantorovich-duality-unbounded} by first establishing the result for bounded positive cost using the perturbation bound developed above, and then extending the argument to the general case via the convergence theorem for linear relations that are bounded below.
The argument is presented as a proof sketch, highlighting how each step contributes to controlling approximation errors and enabling passage to the limits of infinite-dimensional and unbounded costs.






\subparagraph*{Results for bounded positive cost.}
Following \Cref{fig:OT-duality-proof-sketch}, the bounded positive case is proved by four steps:


\textbf{Step 1:} \emph{Truncation.} Let $\delta > 0$ be a fixed constant in the following.
Given a density operator $\rho_1$ on $\mathcal{H}_1$, consider its spectral decomposition $\sum_j \lambda_{1j}\ket{e_{1j}}\bra{e_{1j}}$.
Since $\tr(\rho_1)=1$, there exists an integer $N_1$ such that
$\sum_{j=1}^{N_1}\lambda_{1j} \ge 1-\delta$.
Let $K_1$ be the finite-dimensional subspace spanned by $\left\{\ket{e_{1j}}\right\}_{j=1}^{N_1}$.
Applying the same construction to $\rho_2$ yields a finite-dimensional subspace $K_2$ of $\cH_2$. We project the primal problem onto $K_1\otimes K_2$.
A careful analysis shows that this truncation leads to an additive error of at most $2\delta \|Q\|$ in the primal optimal value $\inf_{\rho}\tr(Q\rho)$ between the original and truncated problems.


\textbf{Step 2:} \emph{Finite-dimensional duality.} 
The quantum Kantorovich duality theorem for finite-dimensional Hilbert spaces
(\Cref{thm:quantum-Kantorovich-duality-finite-dim}) applies to the truncated problem.
As a result, the primal and dual optimal values coincide in the finite-dimensional setting.



\textbf{Step 3:} \emph{Improving over an admissible pair.}
By applying the core lemma (\Cref{lem:perturbation-analysis}), which we discussed in the previous subsection, we obtain a near-optimal dual feasible pair $Q_1'$ and $Q_2'$ for the truncated dual problem, satisfying $\max\{\norm{Q_1'}, \norm{Q_2'}\}\le \norm{Q} + \frac{1}{\varepsilon}\norm{Q}^2$.
Using these operators, we define the improved dual feasible solution as follows:
    \begin{equation*}
        Q_1 \triangleq Q_1' - \frac{\varepsilon}{2} \mathsf{P}_{K_1} -  \rbra*{\frac{\norm{Q}^2}{\varepsilon}+\norm{Q_2'}} \mathsf{P}_{K_1^{\perp}}, \quad
        Q_2 \triangleq Q_2' - \frac{\varepsilon}{2} \mathsf{P}_{K_2} -  \rbra*{\frac{\norm{Q}^2}{\varepsilon}+\norm{Q_1'}} \mathsf{P}_{K_2^{\perp}},
    \end{equation*}    
    where $\mathsf{P}_{K}$ stands for the projection onto the space $K$.
By choosing $\varepsilon=\sqrt{\delta}$, a careful analysis shows that the dual optimal values differ by at most    
 $  3\sqrt{\delta}+4\delta \norm{Q}+6\sqrt{\delta}\norm{Q}^2$
from that of the truncated dual problem.


\textbf{Step 4:} \emph{Combining all and taking a limit.} 
From the preceding arguments, we conclude that for any $\delta>0$, the optimal primal and dual values differ by at most
$3\sqrt{\delta}+6\delta \|Q\|+6\sqrt{\delta}\|Q\|^2$.
Since $Q$ is bounded, letting $\delta\to 0$ yields vanishing error, thereby establishing the duality theorem for bounded positive cost operators.


\subparagraph*{Generalizing to bounded-by-below cases via the convergence theorem.} For a self-adjoint linear-relation $Q$ that is bounded below, we consider its truncation sequence $Q^{(n)}\triangleq\trunc{Q}{n}$. Note that each $Q^{(n)}$ is a bounded self-adjoint operator and $Q^{(n)}\sqsubseteq Q$.
Recall that $\dopt{Q}$ denotes the dual optimal value (\Cref{lem:perturbation-analysis}).
By monotonicity of the dual problem under truncation, we have $\dopt{Q} \ge \sup_{n} \dopt{Q^{(n)}} $.
By the bounded duality theorem and the convergence theorem (\Cref{thm:main-convergence-for-trunc}), 
we obtain $$\sup_{n} \dopt{Q^{(n)}} = \sup_{n} \opt{Q^{(n)}} \ge \liminf_{n} \opt{Q^{(n)}} \ge \opt{Q}.$$
Combining this with weak duality, $\dopt{Q}\le \opt{Q}$, yields the desired duality result for bounded-by-below cost.

This completes the proof of our generalized quantum Kantorovich duality theorem. In the next section, we apply this duality to infinite-dimensional quantum relational logic.

\section{Infinite-Dimensional Quantum Programs}\label{sec:infinite programs}
This section extends the results of Section~\ref{sec:finite quantum} and establishes a sound and complete relational program logic for infinite-dimensional quantum programs and our general form of assertions.

\subparagraph*{Preliminaries.} 
We write $\mathcal{QO}(\mathcal{H})$ for the set of quantum operations, i.e., the completely positive trace-nonincreasing linear maps from $\qpstate{\cH}$ to $ \qpstate{\cH}$, where completely positive means for every \emph{finite}-dimensional Hilbert space $\cH_{\mathrm{aux}}$ and every density operator $\rho\in\qstate{\cH\otimes\cH_{\mathrm{aux}}}$, $(\mathcal{E}\otimes I_{\cH_{\mathrm{aux}}})(\rho)$ is positive.

\subparagraph*{Syntax and semantics.}
Programs are written in the \textsf{qWhile} language, with syntax defined in \Cref{fig:syntax}.
In contrast to finitely-valued quantum programs, we allow variables to range over possibly countably infinite sets of values, such as the natural numbers or integers.
Accordingly, programs operate over states $\qstate{\mathcal{H}}$, where
$\mathcal{H} = \bigotimes_{q \in \qvars} \mathcal{H}_q(\mathcal{V}_q)$ is an infinite-dimensional separable Hilbert space.

Each program $S$ has an interpretation $\sem{S}\in \mathcal{QO}(\mathcal{H})$~\cite{Ying11}. 
We say that $S$ is almost surely terminating, or AST, if $\sem{S}$ is trace-preserving, i.e., $\sem{S}$: $\qstate{\mathcal{H}} \rightarrow \qstate{\mathcal{H}}$.

\subparagraph*{Assertions.}
Relational assertions are elements of $\qbbmaps{\mathcal{H}\otimes\mathcal{H}}$, i.e.\, the extended self-adjoint operators bounded by below on $\mathcal{H}\otimes\mathcal{H}$, and are ordered by the extended L\"{o}wner order $\sqsubseteq$.

The weakest precondition $\wprecond{S}{Q} \in \qbndmaps{\mathcal{H}}$ for a program $S$ is defined only for bounded assertions $Q \in \qbndmaps{\mathcal{H}}$, and satisfies
$\tr(\wprecond{S}{Q}\,\rho) = \tr(Q\,\sem{S}(\rho))$
for all $\rho \in \qstate{\mathcal{H}}$.
Our definition of $\wprecond{S}{Q}$ mildly extends prior work~\cite{Ying11,Ying16} by allowing bounded $Q$, since any bounded operator can be normalized and shifted to lie between $0$ and $I$.

\subparagraph*{Program logic.}
Judgments are of the form $\rtriple{P}{S_1}{S_2}{Q}$, where $S_1,S_2$ are AST programs and 
$P,Q\in\qbbmaps{\mathcal{H}\otimes\mathcal{H}}$.
We say that $\rtriple{P}{S_1}{S_2}{Q}$ is valid, written $\models\rtriple{P}{S_1}{S_2}{Q}$, 
if for every $\rho \in \qstate{\cH\otimes \cH}$, 
there exists a coupling $\sigma$ in 
  $\couplings{\sem{S_1}(\tr_2(\rho))}{\sem{S_2}(\tr_1(\rho))}$ such that
  $\Tr(P \rho) \geq \Tr(Q \sigma)$.
  
The rules in \Cref{fig:prob:rules} are still sufficient to establish a sound and complete proof system, with appropriate adaptations of definitions to account for linear relations.
To justify the soundness of the \RuleRef{duality} rule, we proceed as follows. First observe that the judgment
$\rtriple{P}{S_1}{S_2}{Q}$ is valid if and only if, for all
$\rho \in \mathcal{D}^1(\mathcal{H} \otimes \mathcal{H})$,
$\Tr(P\rho) \ge \inf_{\sigma \in \couplings{\sigma_1}{\sigma_2}} \Tr(Q\sigma)$,
where $\sigma_1 \triangleq \sem{S_1}(\tr_2(\rho))$ and
$\sigma_2 \triangleq \sem{S_2}(\tr_1(\rho))$ are the respective outputs of
$S_1$ and $S_2$; the infimum is attained since the set of couplings is compact~\cite[Theorem 1.4]{qstrasseninf}.
Next, applying \Cref{thm:quantum-Kantorovich-duality-unbounded} to rewrite the
right-hand side yields
\begin{align*}
\Tr(P\rho)
&\ge \sup_{(Q_1,Q_2)\in\dualset{Q}}
   \bigl(\Tr(Q_1\sigma_1) + \Tr(Q_2\sigma_2)\bigr) \\
&= \sup_{(Q_1,Q_2)\in\dualset{Q}}
   \inf_{\sigma\in\couplings{\sigma_1}{\sigma_2}}
   \Tr\bigl((Q_1\boxplus Q_2)\sigma\bigr),
\end{align*}
where the second equality follows from the fact that the expectation of a split
assertion $Q_1\boxplus Q_2$ depends only on the marginals and is independent of
the choice of coupling
$\sigma\in\couplings{\sigma_1}{\sigma_2}$, which establishes the soundness of the \RuleRef{duality} rule.
The \RuleRef{duality} rule thus reduces a bounded-by-below postcondition to a family of judgments with split and bounded assertions, which can be further discharged using the \RuleRef{wp} rule, the weakest precondition rule in the relational setting.

\begin{theorem}[Soundness and Completeness of core rules]\label{thm:quantum-infinite-complete}
A judgment is valid iff it can be derived with the rules \RuleRef{duality},
\RuleRef{conseq}, and \RuleRef{wp}.
\end{theorem}


As discussed, the \RuleRef{wp} rule is typically replaced by more practical rules, such as one-sided or two-sided rules adapted from~\cite{qotl}, which exploit the syntactic structure of programs, or by an instantiation of the lifting rule that invokes QHL~\cite{Ying11} for single programs.

\subparagraph*{Running example (quantum case).}
Let $Q \triangleq |00\>_{q_1q_2}\<00| + |11\>_{q_1q_2}\<11| $ and consider the judgment
$$\vdash \rtriple{ I \mid 0 }{ \textsc{QBern}}{\textsc{QUnif}^2}{Q \mid 0}$$
(see \Cref{fig:placeholder} for detailed programs), asserting that two programs yield 
the same partial states on $q_1$ and $q_2$. Note that the two programs are not aligned w.r.t.\, measurement, and therefore the \RuleRef{duality} rule is required for the proof. Moreover, note that the post-condition involves infinity and therefore cannot be proved using the bounded duality rule of~\cite{qotl}. However, it can be proved in our logic. The derivation applies \RuleRef{conseq}, \RuleRef{wp}, and \RuleRef{duality}, reducing the goal to the entailment
\begin{align*}
0 \sqsupseteq\ & \wprecond{\textsc{QBern}}{A_1}\boxplus \wprecond{\textsc{QUnif}^2}{A_2}
\end{align*}
for every $A_1, A_2$ such that $A_1\boxplus A_2\sqsubseteq Q\mid 0$, which can be discharged using elementary reasoning.

\subparagraph*{\textbf{Case study:} Expected distance in an infinite quantum walk.}
Quantum walks~\cite{ABN+01, kempe2003quantum} are the quantum analogues of classical random walks, and constitute a fundamental technique in quantum algorithms.
We analyze discrete-time quantum walks~\cite{ABN+01, kempe2003quantum} using a coin register $\cH_c$ with basis $\{\ket{L},\ket{R}\}$ and an infinite-dimensional position register $\cH_p$. We additionally introduce an infinite-dimensional time register $\cH_t$ for coherent control to remain in a purely quantum setting.
Following~\cite{ABN+01}, let $C$ be the Hadamard operator on $\cH_c$ , and $S$ the shift on $\cH_c\otimes \cH_p$ where $S\ket{R}\ket{n} = \ket{R}\ket{n+1}$ and $S\ket{L}\ket{n} = \ket{L}\ket{n-1}$.
A single step is $W\triangleq S(C\otimes I_p)$, and the controlled walk is $CW \triangleq \sum_{j=0}^{+\infty} |j\>\<j| \otimes W^j$.
The quantum random walk program $S_{QW}$ is as follows, which initializes the coin to $\ket{R}$ and position to $\ket{0}$ before applying $CW$:
\begin{equation*}
    S_{QW}\triangleq {c}\coloneqq{\ket{R}}; {p}\coloneqq{\ket{0}}; {t,c,p}\coloneqq CW[t,c,p].
\end{equation*}

A well-known property of quantum walks~\cite{ABN+01} is their linear growth: after $t$ steps, the expected distance from the origin scales as $\Theta(t)$, in contrast to the $\Theta(\sqrt{t})$ scaling exhibited by classical random walks. This asymptotic separation is a key mechanism underlying the speedups obtained by numerous quantum algorithms based on quantum walks.
We analyze this behavior using relational quantum program logic as follows. 

We model the step count distribution $\mu$ via $|\varphi_{\mu}\> = \sum \sqrt{\mu_j}\ket{j}_t$, and let $\mathsf{P}_{\mu} \triangleq |\varphi_{\mu}\>\<\varphi_{\mu}|$. To bound the expected distance, we decompose the position operator $P$ into positive ($P_1$) and negative ($P_2$) parts, where $P_{1}=\sum_{j=0}^{\infty} j|j\>_p\<j|$ and $P_{2}=\sum_{j=-\infty}^{0} - j|j\>_p\<j|$.
The following judgment asserts that the expected distance is upper bounded by a constant $c_0$ that depends on $\ket{\varphi_{\mu}}$:
$$ \rtriple{\mathsf{P}_{\mu}[t^{\ave{1}}] \otimes \mathsf{P}_{\mu}[t^{\ave{2}}] \mid c_0I}{S_{QW}}{S_{QW}}{ P_{1}\boxplus P_{2}}.$$ 
For example, if $\ket{\phi_{\mu}}=\sum_{j\ge 1}\sqrt{2^{-j}}\ket{j}$, we can take $c_0=2$.

To prove this, we first apply \RuleRef{duality}, reducing the postcondition to the truncated $\trunc{P_{\pm}}{n}$ for all $n \in \mathbb{N}$.
Then, by \RuleRef{wp} and \RuleRef{conseq}, it suffices to show
$ \mathsf{P}_{\mu}[t^{\ave{1}}] \otimes \mathsf{P}_{\mu}[t^{\ave{2}}] \mid c_0I \sqsupseteq A_1\boxplus A_2 $,
where $A_i \triangleq \wprecond{S_{QW}}{\trunc{P_{i}}{n}}$ for $i=1,2$. These simplify to:
\begin{equation*}
     A_i = \sum_{j=0}^{+\infty}\sum_{k=-\infty}^{+\infty} \sum_{l=L,R} |j\>_{t}\<j|\otimes \big(|lk\>_{c,p}\<R0|B_j|R0\>_{c,p}\<lk|\big),
\end{equation*}
with $B_{ij}= (W^{\dagger})^j (I_{c}\otimes \trunc{P_{i}}{n}) W^j$.

Finally, simplifying the extended L\"owner order requires proving
$   cI \sqsupseteq (\mathsf{P}_{\mu}A_1\mathsf{P}_{\mu}) \otimes I + I \otimes (\mathsf{P}_{\mu}A_2\mathsf{P}_{\mu})$.
Equivalently, the sum of the largest eigenvalues of $\mathsf{P}_{\mu}A_1 \mathsf{P}_{\mu}$ and $\mathsf{P}_{\mu}A_2 \mathsf{P}_{\mu} $ is at most $c_0$.
Notably, these eigenvalues differ due to the intrinsic asymmetry of quantum walks~\cite{ABN+01}.
Similarly, lower bounds on the expected distance can be derived by encoding the diagonal operator $t-\abs{p}$ in the postcondition; we omit the details as they are analogous.
\section{Classical-Quantum Programs}\label{sec:cqprog}

In this section, we transfer our completeness results for relational logics of infinite-dimensional quantum programs to the setting of classical--quantum hybrid programs. To this end, we rely on the well-known fact that countable classical state spaces can be embedded into infinite-dimensional quantum systems, allowing classical-quantum programs to be considered as purely infinite-dimensional quantum programs.

\subparagraph*{Preliminaries.}
Let $X$ denote a countable set, 
and $\cH$ be a separable Hilbert space.
All classical-quantum objects such as $\qbbmaps{X,\cH}$, $\qbndmaps{X,\cH}, \cqstates{X}{\cH}$, as well as operations including (partial) trace and sum $\boxplus$, are obtained by pointwise lifting quantum ones over~$X$, with two basic remarks: 1) boundedness should be uniform for all $x \in X$; and 2) (partial) traces are aggregated over~$X$.
For example, the expectation $\E_\cqstate[\cqpredicate]$ is defined by $\sum_{x\in X}\Tr(\cqpredicate(x)\cqstate(x))$ for $\cqpredicate\in\bbmaps{X,\cH}$ and $\cqstate\in\states{X,\cH}$.

\subparagraph*{Duality Theorem.}
Our duality theorem for infinite-dimensional quantum systems can be generalized to classical-quantum setting, resulting in the following duality theorem.
\begin{theorem}[Kantorovich-Rubinstein Duality Theorem for Classical-Quantum Systems]
\label{thm:cq-kr-duality}
Let $\cqpredicate\in\cqbbmaps{X_1\times X_2}{\cH_1\otimes\cH_2}$, and let $\dualset{\cqpredicate}\subseteq
  \cqbndmaps{X_1}{\cH_1}\times\cqbndmaps{X_2}{\cH_2}$ such that
$(\cqpredicate_1,\cqpredicate_2)\in \dualset{\cqpredicate}$ iff
      $\cqpredicate_1 \boxplus \cqpredicate_2\sqsubseteq \cqpredicate$. Then, for any $\cqstate_1\in \cqstates{X_1}{\cH_1}$ and $\cqstate_2\in \cqstates{X_2}{\cH_2}$,
$$
 \inf_{\cqstate
  \in \couplings{\cqstate_1}{\cqstate_2}} \E_{\cqstate}[\cqpredicate] =\sup_{(\cqpredicate_1,\cqpredicate_2)\in\dualset{\cqpredicate}}
  \{\E_{\cqstate_1}[\cqpredicate_1]+\E_{\cqstate_2}[\cqpredicate_2]\}.$$
\end{theorem}

To prove the theorem, we reduce the problem to the (infinite-dimensional) purely quantum setting. The reduction relies on an embedding-retraction pair, where the embedding map $\iota$ sends a classical–quantum linear operator $\cqstate$ to a purely quantum operator defined by $\iota (\cqstate)\triangleq \sum_{x\in X} \ket{x}\bra{x} \otimes \cqstate(x)$, and its dual, the retraction map $\mathcal{R}$, which sends a purely quantum operator $A$ to a classical–quantum operator via $\mathcal{R}(A) (x) = \bra{x}A\ket{x}$.
The following propositions capture the essential properties of these maps and form the backbone of the reduction argument; together, they enable the transfer of coupling and feasibility statements between the classical–quantum and purely quantum settings:
\begin{alphaenumerate}
    \item $\inf_{\cqstate\in \couplings{\cqstate_1}{\cqstate_2}} \E_{\cqstate}[\cqpredicate] = \inf_{\rho\in \couplings{\iota(\cqstate_1)}{\iota(\cqstate_2)}}\Tr(\iota(\cqpredicate)\rho)$.
    \item $\tr(A\iota(\cqstate)) = \E_{\cqstate}[\mathcal{R}(A)]$ if $A$ is bounded.
    \item If $(A_1, A_2)\in \dualset{\iota(\cqpredicate)}$, then $(\mathcal{R}(A_1), \mathcal{R}(A_2)) \in \dualset{\cqpredicate}$.
\end{alphaenumerate}

We now sketch the proof of~\cref{thm:cq-kr-duality}.
As usual, we refer to the infimum as the primal value and the supremum as the dual value. By (a), the classical–quantum primal value coincides with the purely quantum primal value with cost $\iota(\cqpredicate)$.
 By~\cref{thm:quantum-Kantorovich-duality-unbounded}, this equals the corresponding quantum dual value
$$\sup_{(A_1,A_2)\in\dualset{\iota(\cqpredicate)}}
  \{\tr(A_1\iota(\cqstate_1))+\tr(A_2\iota(\cqstate_2))\}.$$
  Using (b), this expression becomes  $$\sup_{(A_1,A_2)\in\dualset{\iota(\cqpredicate)}}  \{\E_{\cqstate_1}[\mathcal{R}(A_1)]+\E_{\cqstate_2}[\mathcal{R}(A_2)]\}.$$
  Finally, by (c), every feasible quantum pair $(A_1, A_2)$  induces a feasible pair $(\mathcal{R}(A_1), \mathcal{R}(A_2)) $
 for the classical–quantum dual problem. Hence, the above value is upper-bounded by the classical–quantum dual value. The claim follows by weak duality. 

\subparagraph*{Syntax and semantics.}
Programs are written in the \textsf{cqWhile} language, with syntax defined in \Cref{fig:syntax-cqwhile}. 
\textsf{cqWhile} is a combination of \textsf{pWhile} and \textsf{qWhile}, enriched with explicit measurement statements that allow information to be extracted from the quantum system and stored in classical variables.
As a result, the classical control flow in \textsf{qWhile} can be replaced by classical guards in \textsf{cqWhile}.

\begin{figure}
\begin{equation*}
    \begin{aligned}
      S ::=\ &  \skp 
        \mid  \assign{x}{e}  
        \mid \sample{x}{\mu}  \\
        & \mid q \coloneqq \ket{e} 
        \mid \oq \coloneqq U [\oq]
        \mid x\gets \imeas\; M[\oq] \\
        & \mid  S_1; S_2  
        \mid \ifte{e}{S_1}{S_2} 
        \mid \whiled{e}{S}
\end{aligned}
\end{equation*}
\vspace{-1em}
\caption{\textsf{cqWhile} languages. Here $x$ ranges over a finite set $\cvars$ of classical variables, $q$ ranges over a finite set $\qvars$ of quantum variables, $e$ ranges over classical expressions, $\mu$ ranges over distribution expressions, $U$ ranges over unitaries, $M$ range over quantum measurements. Expressions take values over a set $\mathcal{V}$.  
}\label{fig:syntax-cqwhile}
\end{figure}

Programs operate over classical-quantum states 
$\cqstates{\cstates}{\cH}$, where $\cstates = \cvars\rightarrow \mathcal{V}$ for classical states, and $\cH = \otimes_{x\in\qvars}\cH_{x}(\mathcal{V}_x)$ the Hilbert space for quantum systems.
Each program $S$ has an interpretation $\sem{S} : \cqstates{\cstates}{\cH}\rightarrow \cqpstates{\cstates}{\cH}$, a point-wise lifting of completely positive linear maps~\cite{feng2020,barbosa2021easypqc}. We say that $S$ is almost surely terminating, or AST, if $\sem{S}$ is trace-preserving, i.e., $\sem{S} : \cqstates{\cstates}{\cH}\rightarrow \cqstates{\cstates}{\cH}$.
 
\subparagraph*{Assertions.}

Relational assertions are elements of $\cqbbmaps{\cstates\times \cstates}{\cH\otimes \cH}$, and are ordered by $\sqsubseteq$, the pointwise lifting of extended L\"owner order.
Note that our choice of assertions is less standard than the usual requirement of \emph{positive} ones. However, the two notions are equivalent for AST programs.

There are two important assertions, namely the weakest precondition and Boolean guarded assertions, which are used for establishing the \RuleRef{wp} rule and one/two-sided rules respectively. Specifically,
the weakest precondition $\wprecond{S}{\cqpredicate} \in \cqbndmaps{\cstates}{\cH}$ for a single AST program $S$ and $\cqpredicate\in \cqbndmaps{\cstates}{\cH}$ is the unique assertion which satisfies $\E_{\cqstate}[\wprecond{S}{\cqpredicate}] = \E_{\sem{S}(\cqstate)}[\cqpredicate]$ for all $\cqstate\in\cqstates{\cstates}{\mathcal{H}}$~\cite{feng2020}.

\begin{example}[Boolean Guarded Assertions]
For every $\phi\subseteq\cstates\times\cstates$ and $\cqpredicate\in\bbmaps{\cstates\times\cstates, \cH\otimes \cH}$, let 
$\phi \mid \cqpredicate \in \bbmaps{\cstates\times\cstates, \cH\otimes \cH}$
be the Boolean guarded assertion defined by the clause:
$$(\phi \mid \cqpredicate)(s,s') 
= \left\{\begin{array}{ll} \cqpredicate(s,s') & \mbox{if }
(s,s')\in \phi \\ A_\infty & \mbox{otherwise,} \end{array}\right.
$$
where $A_\infty\triangleq\{(0,u):u\in \cH\otimes\cH\}$ is the fully infinite linear relation, i.e., its quadratic form or extended trace are always infinite for non-zero states or density operators, respectively. 
\end{example}

\subparagraph*{Program logic.}

Judgments are of the form $\rtriple{P}{S_1}{S_2}{Q}$, where $S_1,S_2$ are AST programs and 
$P,Q\in\cqbbmaps{\cstates\times \cstates}{\mathcal{H}\otimes\mathcal{H}}$.
We say that $\rtriple{P}{S_1}{S_2}{Q}$ is valid, written $\models\rtriple{P}{S_1}{S_2}{Q}$, 
if for every $\cqstate \in \cqstates{\cstates\times \cstates}{\mathcal{H} \otimes \mathcal{H}}$, 
there exists a coupling $\cqstate'$ in 
  $\couplings{\sem{S_1}(\tr_2(\cqstate))}{\sem{S_2}(\tr_1(\cqstate))}$ such that
  $\E_{\cqstate}[P] \geq \E_{\cqstate'}[Q]$.
  
The rules in \Cref{fig:prob:rules} are still sufficient to establish a sound and complete proof system, with careful adaptation of definitions in the classical-quantum case.

\begin{theorem}[Soundness and Completeness of core rules]\label{thm:classical-quantum-infinite-complete}
A judgment is valid iff it can be derived with the rules \RuleRef{duality},
\RuleRef{conseq}, and \RuleRef{wp}.
\end{theorem}

\subparagraph*{Practical rules.}
As in the other settings, we can derive one-sided and two-sided rules, as well as lifting rules that leverage existing hybrid quantum Hoare logics~\cite{feng2020}, to replace the \RuleRef{wp} rule in practice.
\Cref{fig:rules:twosided} presents selected rules used in the following case study.
The \RuleRef{while} rule for loops follows the standard pattern: it requires the two programs to proceed in lockstep and that the invariant is preserved by the lockstep execution of the loop bodies.
The \RuleRef{sample-supp} rule compares two sampling commands and provides a marginal but useful improvement over the sampling rule of~\cite{erhl} (originally stated in a probabilistic setting, which is immaterial here), by allowing the guarded assertion to be strengthened using the support of the witness coupling~$\mu$.
Here, $\E_{(v, w)\sim \mu}[\psi[v/x_1, w/x_2]] \triangleq \sum_{(v,w)} \mu(v,w) \psi[v/x_1, w/x_2] $.
The assertions involved are all unbounded as they merge both qualitative and quantitative parts.

\begin{figure}
\small
  $
  \begin{prooftree}
      \hypo{\vdash \rtriple{b_1\wedge b_2 \mid \psi}{c_1}{c_2}{b_1 \leftrightarrow b_2 \mid \psi}}
      \infer1{\vdash \rtriple{b_1 \leftrightarrow b_2 \mid \psi}{\whiled{b_1}{c_1}}{\whiled{b_2}{c_2}}{\neg b_1\wedge \neg b_2 \mid \psi}}
    \end{prooftree}
  $ \\[0.5em]
 $   
 \begin{prooftree}
       \hypo{\mu \in \couplings{\mu_1}{\mu_2}\qquad \supp (\mu) \subseteq \xi}
      \infer1{\vdash \rtriple{
       \forall x_1 x_2.~\xi \rightarrow b \mid \E_{(v, w)\sim \mu}[\psi[v/x_1, w/x_2]]}{\sample{x_1}{\mu_1}}{\sample{x_2}{\mu_2}}{
       b \mid \psi}}
    \end{prooftree}
   $
  \caption{Selected two-sided rules: the top rule is \RuleLabel{while}, and the bottom rule is \RuleLabel{sample-supp}.}
  \label{fig:rules:twosided}
\end{figure}

\subparagraph*{Running example (classical-quantum case).}
We analyze our final running example that exercises the \RuleRef{duality} rule,
with the left program \textsc{Bern} and the right program \textsc{CQUnif}$^2$
which initializes and measures two qubits $q_2$ and $q'_2$, which yields two bits $x_2$ and $x_2'$. 
Our goal is to show that the two programs compute the same distribution, i.e.,
$$\rtriple{\top \mid 0}{\textsc{Bern}}{\textsc{CQUnif}^2}{x_1=(x_2\wedge x_2')\mid 0}.$$
Again, the two programs are not \emph{aligned} w.r.t.\, sampling and measurement instructions, and therefore the \RuleRef{duality} is required.
We prove it by first applying the \RuleRef{duality} rule to obtain the (infinitary) judgment
$$\rtriple{ 0 }{S_1}{S_2}{Q_1\boxplus Q_2 }$$
for all $(Q_1,Q_2) \in \dualset{x_1 = (x_2\wedge x_2')\mid 0}$.
Next, applying the \RuleRef{wp} rule leads to the (simplified) proof obligation
$$
{\scriptstyle
\begin{aligned}
    & \left(3Q_1[0/x_1] + Q_1[1/x_1]\right)\boxplus  \left(\sum_{ij} \<ij|Q_2[i/x_2'][j/x_2]|ij\>\right) \sqsubseteq 0,
\end{aligned}}
$$
which is discharged with elementary reasoning.

\subparagraph*{\textbf{Case study:} Algorithmic stability of quantum neural networks.}

Algorithmic stability is a central concept in learning theory, formalizing the requirement that a learning algorithm's output changes only slightly when a single training example is modified~\cite{BE02}. Stability has also been studied in quantum machine learning~\cite{CHC+22,GEB24}.
We revisit the stability of quantum neural networks (QNNs) as introduced in \cite{YXX25}. Abstracting away implementation details, we model a QNN by a program $S_{T,X,\phi_0}$ parameterized by a training set $X$, depth $T$, and initial state $\ket{\phi_0}$, which iteratively applies data-dependent unitaries $U_x$ sampled from $X$. This abstraction captures the essential mechanism by which training data influence quantum evolution while remaining general.
\begin{equation*}
       S_{T,X,\phi_0} \triangleq {q}\coloneqq{\ket{\phi_0}};  \assign{t}{0}; 
        \while~ t< T ~\wdo 
              \;   \sample{x}{X}; q \coloneqq U_x[q] ;  \assign{t}{t+1}; 
         \wod
\end{equation*}
Assuming each $U_x$ is a small perturbation of the identity ($\norm{U_x-I}\le c$ for all $x$), we express stability as a relational judgment between executions on neighboring datasets $X_1\sim_1 X_2$:
$$         \rtriple
 {X_1\sim_1 X_2 \mid 4cT } { S_{T, X_1, \phi_0}} {S_{T, X_2, \phi_0}} { \lambda(m_1, m_2).
   (P\otimes I -I\otimes P)}.$$
   The postcondition encodes the variational characterization of trace distance~\cite[Section VIII.B]{qotl} (a quantum analogue of total variation distance), while the bound $4cT$ reflects the cumulative effect of $T$ layers. 
The main step of the proof is to apply the \RuleRef{while} rule (Figure~\ref{fig:rules:twosided}) with the invariant
$$
\psi \triangleq (t^{\ave{1}} =t^{\ave{2}}) \mid \lambda(m_1, m_2). P\otimes I- I\otimes P + (T-t^{\ave{1}})4c I
$$
and loop guards $b_1 = (t^{\ave{1}} < T)$,  $b_2 = (t^{\ave{2}} < T)$.
After the \RuleRef{while} rule, it suffices to show
$\rtriple{b_1\wedge b_2\mid \psi}{B}{B}{b_1\wedge b_2\mid \psi}$, where $B$ denotes the loop body. This follows from \RuleRef{sample-supp} and \RuleRef{conseq}, using a coupling that is the identity except on the unique differing samples $x_1\in X_1$ and $x_2\in X_2$.

\section{Related Work}\label{sec:relwork}

\begin{figure}[t]
\small
\begin{tabular}{|l|l|l|l|l|l|}
 \hline
 Logic & Language & Dimension & Assertion Language & Completeness  \\
  \hline
\textsf{pRHL}~\cite{BartheGZ09}  & \textsf{pWhile} && boolean-valued &
$\no$  \\
$\mathbb{E}$\textsf{pRHL}~\cite{BartheEGHS18} & \textsf{pWhile} && $+\infty$-valued &
$\no$  \\
\textsf{eRHL}~\cite{erhl} & \textsf{pWhile} && real-valued &
$\yes^{\textrm{s}}$  \\
\textsf{eRHL} $+$ ~\cite{qotl}
& \textsf{pWhile} && $+\infty$-valued &
$\yes^{\textrm{b}}$ \\
\textbf{\Cref{sec:probprog}}
& \textbf{\textsf{pWhile}} & &
\textbf{$+\infty$-valued} &
\textbf{$\yes$}  \\
\hline 
\textsf{rqPD}~\cite{GHY19} & \textsf{qWhile} & finite & observable &
$\no$  \\
\textsf{eqRHL}~\cite{LU19} & \textsf{qWhile} & countably infinite  & observable &
$\no$   \\
\textsf{qOTL}~\cite{qotl} & \textsf{qWhile} & finite & $+\infty$-valued observable &
$\yes^{\mathrm{b}}$ \\
\textbf{\Cref{sec:infinite programs}} & 
\textbf{\textsf{qWhile}} & \textbf{countably infinite} &
\textbf{linear relation} &
\textbf{$\yes$} \\
\hline
\textsf{qRHL}~\cite{Unr19a} & \textsf{cqWhile} & countably infinite & projection &
$\no$  \\
\textsf{EasyPQC}~\cite{barbosa2021easypqc}  & \textsf{cqWhile} & finite & boolean-valued &
$\no$   \\
\textbf{\Cref{sec:cqprog}} & \textbf{\textsf{cqWhile}} & \textbf{countably infinite} & \textbf{linear relation} &
\textbf{$\yes$}  \\
\hline
\end{tabular}
\caption{Comparison with prior logics. Ours results are shown in \textbf{bold}.
Language denotes the language setting, either probabilistic (\textsf{pWhile}), purely quantum (\textsf{qWhile}), or classical-quantum (\textsf{cqWhile}). Dimension indicates the dimension of the evolving quantum systems. $\yes{}$, $\yes^s$, $\yes^b$ indicates that completeness holds unconditionally, for split post-conditions, and for bounded post-conditions.  
} 
\label{fig:completeness}
\end{figure}

\subparagraph*{Relational program logics.}
There is a large body of work that develops coupling-based relational program logics for probabilistic and quantum programs, e.g.\,~\cite{BartheGZ09,Unr19a,GHY19,LU19,barbosa2021easypqc,DBLP:journals/pacmpl/0001BHKKM21,DBLP:journals/pacmpl/GregersenAHTB24,DBLP:journals/pacmpl/HaselwarterLAGTB25,erhl,qotl}. \Cref{fig:completeness} provides an overview of selected logics. All these logics are sound w.r.t.\, program semantics. However, few relational logics achieve completeness, and only for restricted classes of assumptions. In the probabilistic setting, \cite{erhl} achieves completeness for a class of split post-conditions. In the quantum setting, \cite{qotl} uses the duality theorem to achieve completeness for finite-dimensional quantum programs and bounded assertions---and similarly for the probabilistic case. Our results subsume all completeness results, but showing completeness for a more general class of programs, without any restrictions on assertions. Most of the aforementioned works focus on the discrete setting; however, a few works, e.g.~\cite{DBLP:journals/entcs/Sato16} consider the continuous setting. In addition, some works e.g.~\cite{DBLP:journals/pacmpl/BaoDF25}, develop relational program logics where coupling-based rules are derived rather than primitives, but these works do not consider completeness issues.

\subparagraph*{Quantum optimal transport.}

There is growing interest in extending optimal transport~\cite{Villani08,cuturi2019computational} to the quantum settings. An early attempt \cite{ZS98} defines quantum transport cost via a probabilistic Monge distance approach. However, most subsequent work~\cite{GMP16,GP17,GP18,CGP20,CEFZ23,FECZ22a} adopts a coupling-based formulation. Under this formulation,
Kantorovich-type duality results were established in \cite{CGP22,GP22a,ZYYY22}. 
A complementary perspective is provided by dynamic formulations:
channel-based approaches include~\cite{DT21,BPTV25,HCW25,DP25};
\cite{CM17,CM20,Wir21,Wir22,WZ21} adopt a gradient flow approach.
Another active direction concerns Lipschitz-type transport distances~\cite{DMTL21,DT23,DKP24}.
Quantum optimal transport also has applications in quantum machine learning \cite{CYL+19,JM23} and in quantum Markov chains \cite{BLMT25}.
  For a broader overview, we refer to the recent survey \cite{Bea25}.

\subparagraph*{Linear relations.}

Our work rests on the theory of linear relations—initiated by Arens~\cite{arens1961operational} and consolidated by Cross~\cite{cross1998multivalued}—which provides the lattice-theoretic basis for our quantum assertions. From an analytic perspective, it generalizes the classic method of defining self-adjoint operators via quadratic forms, a technique originating with Friedrichs' extensions~\cite{friedrichs1934spektraltheorie}. This line of work was formalized by Kato's representation theorem~\cite{kato2013perturbation}, extended by Simon to lower semi-continuous forms—which is essential for establishing convergence~\cite{Simon_1978}—and ultimately generalized to linear relations by Hassi et al.~\cite{Behrndt_Hassi_de_Snoo_Wietsma_2010, behrndt2020boundary} for a rigorous treatment of singular observables.

Our approach also shares deep structural connections with non-commutative integration theory, particularly traces in von Neumann algebras~\cite{ruskai1972inequalities} and $\tau$-measurable operators~\cite{fack1986generalized, haagerup1979lp}, which provides the broader functional analytic context for unbounded operators.
 \section{Discussion and Conclusion}

We have introduced sound and complete program logics for infinite-dimensional quantum programs and classical-quantum programs. Our logics are based on a new interpretation of assertions as linear relations, together with new duality theorems for discrete infinite-dimensional quantum states.

\subparagraph*{Beyond almost-sure termination.}
A first direction for future work is to relax the almost-sure termination (AST) assumption. In our present development, AST plays three roles. It ensures that program outputs can always be coupled, it matches the hypotheses of the duality theorems, and it allows bounded-by-below assertions to be related to non-negative ones. Removing, or at least weakening, this assumption would therefore require changes at both the semantic and proof-theoretic levels. One possible route is to replace full couplings with more general notions such as partial couplings, which can compare programs with different probabilities of termination. Another is to develop or invoke a duality theorem that does not require termination-preserving semantics. This points naturally to unbalanced optimal transport, where the two marginals need not have the same total mass. Understanding whether such dualities can support complete relational logics for non-AST quantum programs is an important theoretical direction.

\subparagraph*{From foundational logics to practical reasoning.}
The logics developed in this paper are intentionally foundational: they isolate the core principles needed for soundness and completeness. For practical verification, however, these core systems need to be enriched with additional proof rules. In particular, one-sided and two-sided rules, rules for common program constructs, and domain-specific reasoning principles will be essential for making the logics usable in larger examples. We expect that, once equipped with such derived rules, the quantum and classical-quantum logics developed here can reach a level of practicality comparable to their probabilistic counterparts.

At the same time, we do not expect these logics to serve as the basis of a fully automated verification tool. As in \textsf{pRHL}, the intended mode of use is likely to be interactive: users guide the proof by choosing couplings, invariants, and decompositions, while automation assists with routine proof obligations. A key difference from the probabilistic setting is the nature of entailment between assertions. In our setting, assertions are linear relations, and entailment involves reasoning about extended Löwner order and potentially unbounded quantum observables. Developing partial automation for such entailment problems is therefore a central step toward making the logics practically usable.

\subparagraph*{Applications to post-quantum and quantum cryptography.}
A major application area for future work is to use the logics to formalize post-quantum and quantum cryptography, respectively to prove security of new post-quantum cryptographic standards and of quantum key distribution (QKD). A main challenge for the post-quantum case is to develop appropriate proof rules to reason about quantum random oracles (QROM)~\cite{boneh2011random}, and cost logics to reason about the complexity of classical-quantum adversaries. In contrast, a main challenge for the quantum case is to develop expressive proof rules for the adversary---in the quantum case, the adversary may be entangled with the state of the cryptographic system whereas in the post-quantum case the cryptographic system is classical.


\nocite{AB06,MMS96,reed1980methods,rudin1987real}
\bibliography{ref}

\begin{thebibliography}{10}

\bibitem{DBLP:journals/pacmpl/0001BHKKM21}
Alejandro Aguirre, Gilles Barthe, Justin Hsu, Benjamin~Lucien Kaminski, Joost{-}Pieter Katoen, and Christoph Matheja.
\newblock A pre-expectation calculus for probabilistic sensitivity.
\newblock {\em Proc. {ACM} Program. Lang.}, 5({POPL}):1--28, 2021.
\newblock \href {https://doi.org/10.1145/3434333} {\path{doi:10.1145/3434333}}.

\bibitem{AB06}
Charalambos~D. Aliprantis and Kim~C. Border.
\newblock {\em Infinite {{Dimensional Analysis}}}.
\newblock Springer-Verlag, Berlin/Heidelberg, 2006.
\newblock \href {https://doi.org/10.1007/3-540-29587-9} {\path{doi:10.1007/3-540-29587-9}}.

\bibitem{ABN+01}
Andris Ambainis, Eric Bach, Ashwin Nayak, Ashvin Vishwanath, and John Watrous.
\newblock One-dimensional quantum walks.
\newblock In {\em Proceedings of the Thirty-Third Annual {{ACM}} Symposium on {{Theory}} of Computing}, pages 37--49, Hersonissos Greece, July 2001. ACM.
\newblock \href {https://doi.org/10.1145/380752.380757} {\path{doi:10.1145/380752.380757}}.

\bibitem{arens1961operational}
Richard Arens.
\newblock Operational calculus of linear relations.
\newblock {\em Pacific J. Math.}, 11(4):9--23, 1961.

\bibitem{erhl}
Martin Avanzini, Gilles Barthe, Davide Davoli, and Benjamin Gr{\'{e}}goire.
\newblock A quantitative probabilistic relational {H}oare logic.
\newblock {\em Proc. {ACM} Program. Lang.}, 9({POPL}):1167--1195, 2025.
\newblock \href {https://doi.org/10.1145/3704876} {\path{doi:10.1145/3704876}}.

\bibitem{BLMT25}
Ainesh Bakshi, Allen Liu, Ankur Moitra, and Ewin Tang.
\newblock A {{Dobrushin}} condition for quantum {{Markov}} chains: {{Rapid}} mixing and conditional mutual information at high temperature, October 2025.
\newblock Accepted by 58th Annual ACM Symposium on Theory of Computing (STOC 2026).
\newblock \href {https://arxiv.org/abs/2510.08542} {\path{arXiv:2510.08542}}.

\bibitem{DBLP:journals/pacmpl/BaoDF25}
Jialu Bao, Emanuele D'Osualdo, and Azadeh Farzan.
\newblock Bluebell: An alliance of relational lifting and independence for probabilistic reasoning.
\newblock {\em Proc. {ACM} Program. Lang.}, 9({POPL}):1719--1749, 2025.
\newblock \href {https://doi.org/10.1145/3704894} {\path{doi:10.1145/3704894}}.

\bibitem{barbosa2021easypqc}
Manuel Barbosa, Gilles Barthe, Xiong Fan, Benjamin Gr\'{e}goire, Shih-Han Hung, Jonathan Katz, Pierre-Yves Strub, Xiaodi Wu, and Li~Zhou.
\newblock Easypqc: Verifying post-quantum cryptography.
\newblock In {\em Proceedings of the 2021 ACM SIGSAC Conference on Computer and Communications Security}, CCS '21, page 2564–2586, New York, NY, USA, 2021. Association for Computing Machinery.
\newblock \href {https://doi.org/10.1145/3460120.3484567} {\path{doi:10.1145/3460120.3484567}}.

\bibitem{BartheEGHS18}
Gilles Barthe, Thomas Espitau, Benjamin Gr{\'{e}}goire, Justin Hsu, and Pierre{-}Yves Strub.
\newblock Proving expected sensitivity of probabilistic programs.
\newblock {\em Proc. {ACM} Program. Lang.}, 2({POPL}):57:1--57:29, 2018.
\newblock \href {https://doi.org/10.1145/3158145} {\path{doi:10.1145/3158145}}.

\bibitem{qotl}
Gilles Barthe, Minbo Gao, Theo Wang, and Li~Zhou.
\newblock {Complete Quantum Relational Hoare Logics from Optimal Transport Duality}.
\newblock In {\em 2025 40th Annual ACM/IEEE Symposium on Logic in Computer Science (LICS)}, pages 884--925, Piscataway, NJ, USA, 2025. IEEE.
\newblock \href {https://doi.org/10.1109/LICS65433.2025.00072} {\path{doi:10.1109/LICS65433.2025.00072}}.

\bibitem{BartheGZ09}
Gilles Barthe, Benjamin Gr{\'{e}}goire, and Santiago~Zanella B{\'{e}}guelin.
\newblock Formal certification of code-based cryptographic proofs.
\newblock In Zhong Shao and Benjamin~C. Pierce, editors, {\em Proceedings of the 36th {ACM} {SIGPLAN-SIGACT} Symposium on Principles of Programming Languages, {POPL} 2009, Savannah, GA, USA, January 21-23, 2009}, pages 90--101. {ACM}, 2009.
\newblock \href {https://doi.org/10.1145/1480881.1480894} {\path{doi:10.1145/1480881.1480894}}.

\bibitem{GHY19}
Gilles Barthe, Justin Hsu, Mingsheng Ying, Nengkun Yu, and Li~Zhou.
\newblock Relational proofs for quantum programs.
\newblock {\em Proc. ACM Program. Lang.}, 4(POPL), December 2019.
\newblock \href {https://doi.org/10.1145/3371089} {\path{doi:10.1145/3371089}}.

\bibitem{Bea25}
Emily Beatty.
\newblock Wasserstein distances on quantum structures: {{An}} overview.
\newblock {\em Reviews in Mathematical Physics}, page 2630003, May 2026.
\newblock \href {https://doi.org/10.1142/S0129055X26300037} {\path{doi:10.1142/S0129055X26300037}}.

\bibitem{behrndt2020boundary}
Jussi Behrndt, Seppo Hassi, and Henk De~Snoo.
\newblock {\em Boundary value problems, Weyl functions, and differential operators}.
\newblock Springer Nature, 2020.

\bibitem{Behrndt_Hassi_de_Snoo_Wietsma_2010}
Jussi Behrndt, Seppo Hassi, Henk de~Snoo, and Rudi Wietsma.
\newblock Monotone convergence theorems for semi-bounded operators and forms with applications.
\newblock {\em Proceedings of the Royal Society of Edinburgh: Section A Mathematics}, 140(5):927–951, 2010.
\newblock \href {https://doi.org/10.1017/S030821050900078X} {\path{doi:10.1017/S030821050900078X}}.

\bibitem{boneh2011random}
Dan Boneh, {\"O}zg{\"u}r Dagdelen, Marc Fischlin, Anja Lehmann, Christian Schaffner, and Mark Zhandry.
\newblock Random oracles in a quantum world.
\newblock In {\em Advances in Cryptology--ASIACRYPT 2011: 17th International Conference on the Theory and Application of Cryptology and Information Security, Seoul, South Korea, December 4-8, 2011. Proceedings 17}, pages 41--69. Springer, 2011.

\bibitem{BE02}
Olivier Bousquet and Andr{\'e} Elisseeff.
\newblock Stability and generalization.
\newblock {\em Journal of Machine Learning Research}, 2:499--526, 2002.

\bibitem{BPTV25}
Gergely Bunth, J{\'o}zsef Pitrik, Tam{\'a}s Titkos, and D{\'a}niel Virosztek.
\newblock Wasserstein distances and divergences of order $p$ by quantum channels, January 2025.
\newblock \href {https://arxiv.org/abs/2501.08066} {\path{arXiv:2501.08066}}.

\bibitem{CGP20}
E.~Caglioti, F.~Golse, and T.~Paul.
\newblock Quantum {{Optimal Transport}} is {{Cheaper}}.
\newblock {\em Journal of Statistical Physics}, 181(1):149--162, October 2020.
\newblock \href {https://doi.org/10.1007/s10955-020-02571-7} {\path{doi:10.1007/s10955-020-02571-7}}.

\bibitem{CGP22}
Emanuele Caglioti, Fran{\c c}ois Golse, and Thierry Paul.
\newblock Towards optimal transport for quantum densities.
\newblock {\em Annali della Scuola Normale Superiore di Pisa, Classe di Scienze}, 2022.

\bibitem{CM17}
Eric~A. Carlen and Jan Maas.
\newblock Gradient flow and entropy inequalities for quantum {{Markov}} semigroups with detailed balance.
\newblock {\em Journal of Functional Analysis}, 273(5):1810--1869, September 2017.
\newblock \href {https://doi.org/10.1016/j.jfa.2017.05.003} {\path{doi:10.1016/j.jfa.2017.05.003}}.

\bibitem{CM20}
Eric~A. Carlen and Jan Maas.
\newblock Non-commutative {{Calculus}}, {{Optimal Transport}} and {{Functional Inequalities}} in {{Dissipative Quantum Systems}}.
\newblock {\em Journal of Statistical Physics}, 178(2):319--378, January 2020.
\newblock \href {https://doi.org/10.1007/s10955-019-02434-w} {\path{doi:10.1007/s10955-019-02434-w}}.

\bibitem{CHC+22}
Matthias~C. Caro, Hsin-Yuan Huang, M.~Cerezo, Kunal Sharma, Andrew Sornborger, Lukasz Cincio, and Patrick~J. Coles.
\newblock Generalization in quantum machine learning from few training data.
\newblock {\em Nature Communications}, 13(1):4919, August 2022.
\newblock \href {https://doi.org/10.1038/s41467-022-32550-3} {\path{doi:10.1038/s41467-022-32550-3}}.

\bibitem{CYL+19}
Shouvanik Chakrabarti, Huang Yiming, Tongyang Li, Soheil Feizi, and Xiaodi Wu.
\newblock {Quantum Wasserstein Generative Adversarial Networks}.
\newblock In H.~Wallach, H.~Larochelle, A.~Beygelzimer, F.~{dAlch{\'e}-Buc}, E.~Fox, and R.~Garnett, editors, {\em Advances in Neural Information Processing Systems}, volume~32. Curran Associates, Inc., 2019.

\bibitem{CEFZ23}
Sam Cole, Micha{\l} Eckstein, Shmuel Friedland, and Karol {\.Z}yczkowski.
\newblock On {{Quantum Optimal Transport}}.
\newblock {\em Mathematical Physics, Analysis and Geometry}, 26(2):14, June 2023.
\newblock \href {https://doi.org/10.1007/s11040-023-09456-7} {\path{doi:10.1007/s11040-023-09456-7}}.

\bibitem{cross1998multivalued}
Ronald Cross.
\newblock {\em Multivalued linear operators}, volume 213.
\newblock CRC Press, 1998.

\bibitem{cuturi2019computational}
Marco Cuturi and Gabriel Peyr{\'e}.
\newblock Computational optimal transport.
\newblock {\em Found. Trends Mach. Learn}, 11(5-6):355--607, 2019.

\bibitem{DKP24}
Giacomo De~Palma, Tristan Klein, and Davide Pastorello.
\newblock Classical shadows meet quantum optimal mass transport.
\newblock {\em Journal of Mathematical Physics}, 65(9):092201, September 2024.
\newblock \href {https://doi.org/10.1063/5.0178897} {\path{doi:10.1063/5.0178897}}.

\bibitem{DMTL21}
Giacomo De~Palma, Milad Marvian, Dario Trevisan, and Seth Lloyd.
\newblock The {{Quantum Wasserstein Distance}} of {{Order}} 1.
\newblock {\em IEEE Transactions on Information Theory}, 67(10):6627--6643, October 2021.
\newblock \href {https://doi.org/10.1109/TIT.2021.3076442} {\path{doi:10.1109/TIT.2021.3076442}}.

\bibitem{DP25}
Giacomo De~Palma and Davide Pastorello.
\newblock Quantum {{Concentration Inequalities}} and {{Equivalence}} of the {{Thermodynamical Ensembles}}: {{An Optimal Mass Transport Approach}}.
\newblock {\em Journal of Statistical Physics}, 192(6):87, June 2025.
\newblock \href {https://doi.org/10.1007/s10955-025-03464-3} {\path{doi:10.1007/s10955-025-03464-3}}.

\bibitem{DT21}
Giacomo De~Palma and Dario Trevisan.
\newblock Quantum {{Optimal Transport}} with {{Quantum Channels}}.
\newblock {\em Annales Henri Poincar\'e}, 22(10):3199--3234, October 2021.
\newblock \href {https://doi.org/10.1007/s00023-021-01042-3} {\path{doi:10.1007/s00023-021-01042-3}}.

\bibitem{DT23}
Giacomo De~Palma and Dario Trevisan.
\newblock The {{Wasserstein Distance}} of {{Order}} 1 for {{Quantum Spin Systems}} on {{Infinite Lattices}}.
\newblock {\em Annales Henri Poincar\'e}, 24(12):4237--4282, December 2023.
\newblock \href {https://doi.org/10.1007/s00023-023-01340-y} {\path{doi:10.1007/s00023-023-01340-y}}.

\bibitem{DP06}
Ellie D'Hondt and Prakash Panangaden.
\newblock Quantum weakest preconditions.
\newblock {\em Mathematical Structures in Computer Science}, 16(3):429–451, 2006.
\newblock \href {https://doi.org/10.1017/S0960129506005251} {\path{doi:10.1017/S0960129506005251}}.

\bibitem{fack1986generalized}
Thierry Fack and Hideki Kosaki.
\newblock Generalized s-numbers of $\tau$-measurable operators.
\newblock {\em Pacific Journal of Mathematics}, 123(2):269--300, 1986.

\bibitem{feng2020}
Yuan Feng and Mingsheng Ying.
\newblock {Quantum Hoare Logic with Classical Variables}.
\newblock {\em ACM Transactions on Quantum Computing}, 2(4), December 2021.
\newblock \href {https://doi.org/10.1145/3456877} {\path{doi:10.1145/3456877}}.

\bibitem{FECZ22a}
Shmuel Friedland, Micha{\l} Eckstein, Sam Cole, and Karol {\.Z}yczkowski.
\newblock Quantum {{Monge-Kantorovich Problem}} and {{Transport Distance}} between {{Density Matrices}}.
\newblock {\em Physical Review Letters}, 129(11):110402, September 2022.
\newblock \href {https://doi.org/10.1103/PhysRevLett.129.110402} {\path{doi:10.1103/PhysRevLett.129.110402}}.

\bibitem{qstrasseninf}
Shmuel Friedland, Jingtong Ge, and Lihong Zhi.
\newblock Quantum {Strassen's} theorem.
\newblock {\em Infinite Dimensional Analysis, Quantum Probability and Related Topics}, 23(03):2050020, 2020.
\newblock \href {https://doi.org/10.1142/S0219025720500204} {\path{doi:10.1142/S0219025720500204}}.

\bibitem{friedrichs1934spektraltheorie}
Kurt Friedrichs.
\newblock Spektraltheorie halbbeschr{\"a}nkter operatoren und anwendung auf die spektralzerlegung von differentialoperatoren.
\newblock {\em Mathematische Annalen}, 109(1):465--487, 1934.

\bibitem{GEB24}
Elies {Gil-Fuster}, Jens Eisert, and Carlos {Bravo-Prieto}.
\newblock Understanding quantum machine learning also requires rethinking generalization.
\newblock {\em Nature Communications}, 15(1):2277, March 2024.
\newblock \href {https://doi.org/10.1038/s41467-024-45882-z} {\path{doi:10.1038/s41467-024-45882-z}}.

\bibitem{GMP16}
Fran{\c c}ois Golse, Cl{\'e}ment Mouhot, and Thierry Paul.
\newblock On the {{Mean Field}} and {{Classical Limits}} of {{Quantum Mechanics}}.
\newblock {\em Communications in Mathematical Physics}, 343(1):165--205, April 2016.
\newblock \href {https://doi.org/10.1007/s00220-015-2485-7} {\path{doi:10.1007/s00220-015-2485-7}}.

\bibitem{GP17}
Fran{\c c}ois Golse and Thierry Paul.
\newblock The {{Schr\"odinger Equation}} in the {{Mean-Field}} and {{Semiclassical Regime}}.
\newblock {\em Archive for Rational Mechanics and Analysis}, 223(1):57--94, January 2017.
\newblock \href {https://doi.org/10.1007/s00205-016-1031-x} {\path{doi:10.1007/s00205-016-1031-x}}.

\bibitem{GP18}
Fran{\c c}ois Golse and Thierry Paul.
\newblock Wave packets and the quadratic {{Monge}}--{{Kantorovich}} distance in quantum mechanics.
\newblock {\em Comptes Rendus. Math\'ematique}, 356(2):177--197, January 2018.
\newblock \href {https://doi.org/10.1016/j.crma.2017.12.007} {\path{doi:10.1016/j.crma.2017.12.007}}.

\bibitem{GP22a}
François Golse and Thierry Paul.
\newblock Optimal transport pseudometrics for quantum and classical densities.
\newblock {\em Journal of Functional Analysis}, 282(9):109417, 2022.
\newblock \href {https://doi.org/10.1016/j.jfa.2022.109417} {\path{doi:10.1016/j.jfa.2022.109417}}.

\bibitem{DBLP:journals/pacmpl/GregersenAHTB24}
Simon~Oddershede Gregersen, Alejandro Aguirre, Philipp~G. Haselwarter, Joseph Tassarotti, and Lars Birkedal.
\newblock Asynchronous probabilistic couplings in higher-order separation logic.
\newblock {\em Proc. {ACM} Program. Lang.}, 8({POPL}):753--784, 2024.
\newblock \href {https://doi.org/10.1145/3632868} {\path{doi:10.1145/3632868}}.

\bibitem{haagerup1979lp}
Uffe Haagerup et~al.
\newblock Lp-spaces associated with an arbitrary von {Neumann} algebra.
\newblock In {\em Algebres d’op{\'e}rateurs et leurs applications en physique math{\'e}matique (Proc. Colloq., Marseille, 1977)}, volume 274, pages 175--184, 1979.

\bibitem{DBLP:journals/pacmpl/HaselwarterLAGTB25}
Philipp~G. Haselwarter, Kwing~Hei Li, Alejandro Aguirre, Simon~Oddershede Gregersen, Joseph Tassarotti, and Lars Birkedal.
\newblock Approximate relational reasoning for higher-order probabilistic programs.
\newblock {\em Proc. {ACM} Program. Lang.}, 9({POPL}):1196--1226, 2025.
\newblock \href {https://doi.org/10.1145/3704877} {\path{doi:10.1145/3704877}}.

\bibitem{heim2020quantum}
Bettina Heim, Mathias Soeken, Sarah Marshall, Chris Granade, Martin Roetteler, Alan Geller, Matthias Troyer, and Krysta Svore.
\newblock Quantum programming languages.
\newblock {\em Nature Reviews Physics}, 2(12):709--722, 2020.

\bibitem{holevo2011probabilistic}
Alexander~S Holevo.
\newblock {\em Probabilistic and statistical aspects of quantum theory}, volume~1.
\newblock Springer Science \& Business Media, Berlin, Heidelberg, 2011.

\bibitem{HCW25}
Matt {Hoogsteder-Riera}, John Calsamiglia, and Andreas Winter.
\newblock Approach to optimal quantum transport via states over time, April 2025.
\newblock \href {https://arxiv.org/abs/2504.04856} {\path{arXiv:2504.04856}}.

\bibitem{Horn2005}
Roger~A. Horn and Fuzhen Zhang.
\newblock {\em The Schur Complement and Its Applications}, chapter Basic Properties of the Schur Complement, pages 17--46.
\newblock Springer US, Boston, MA, 2005.
\newblock \href {https://doi.org/10.1007/0-387-24273-2_2} {\path{doi:10.1007/0-387-24273-2_2}}.

\bibitem{JM23}
Wiktor Jurasz and Christian~B. Mendl.
\newblock Quantum {{Wasserstein GANs}} for {{State Preparation}} at {{Unseen Points}} of a {{Phase Diagram}}, September 2023.
\newblock \href {https://arxiv.org/abs/2309.09543} {\path{arXiv:2309.09543}}.

\bibitem{Kan42}
L.V. Kantorovich.
\newblock On the translocation of masses.
\newblock {\em Dokl. Akad. Nauk SSSR}, 37(7--8):227--229, 1942.
\newblock English translation available in \emph{J. Math. Sci.} (2006).

\bibitem{kato2013perturbation}
Tosio Kato.
\newblock {\em Perturbation theory for linear operators}, volume 132.
\newblock Springer Science \& Business Media, 2013.

\bibitem{kempe2003quantum}
Julia Kempe.
\newblock Quantum random walks: an introductory overview.
\newblock {\em Contemporary Physics}, 44(4):307--327, 2003.

\bibitem{LU19}
Yangjia Li and Dominique Unruh.
\newblock {Quantum Relational Hoare Logic with Expectations}.
\newblock In Nikhil Bansal, Emanuela Merelli, and James Worrell, editors, {\em 48th International Colloquium on Automata, Languages, and Programming (ICALP 2021)}, volume 198 of {\em Leibniz International Proceedings in Informatics (LIPIcs)}, pages 136:1--136:20, Dagstuhl, Germany, 2021. Schloss Dagstuhl -- Leibniz-Zentrum f{\"u}r Informatik.
\newblock \href {https://doi.org/10.4230/LIPIcs.ICALP.2021.136} {\path{doi:10.4230/LIPIcs.ICALP.2021.136}}.

\bibitem{jan24optimal}
Jan Maas, Simone Rademacher, Tam{\'a}s Titkos, and D{\'a}niel Virosztek, editors.
\newblock {\em Optimal Transport on Quantum Structures}.
\newblock Number~29 in Bolyai {{Society}} Mathematical Studies. Springer Nature, Cham, Switzerland, 2024.

\bibitem{MMS96}
Carroll Morgan, Annabelle McIver, and Karen Seidel.
\newblock Probabilistic predicate transformers.
\newblock {\em ACM Transactions on Programming Languages and Systems}, 18(3):325--353, May 1996.
\newblock \href {https://doi.org/10.1145/229542.229547} {\path{doi:10.1145/229542.229547}}.

\bibitem{RR98}
S.~T. Rachev and Ludger R{\"u}schendorf.
\newblock {\em Mass Transportation Problems}.
\newblock Probability and Its Applications. Springer, New York, 1998.

\bibitem{reed1972methods}
Michael Reed and Barry Simon.
\newblock {\em Methods of modern mathematical physics, 2. Fourier Analysis, Self-Adjointness}.
\newblock New York, London: Academic Press, 1972.

\bibitem{reed1980methods}
Michael Reed and Barry Simon.
\newblock {\em Methods of modern mathematical physics: Functional analysis}, volume~1.
\newblock Gulf Professional Publishing, 1980.

\bibitem{rudin1987real}
Walter Rudin.
\newblock {\em Real and complex analysis}.
\newblock McGraw-Hill, Inc., 1987.

\bibitem{ruskai1972inequalities}
Mary~Beth Ruskai.
\newblock Inequalities for traces on von neumann algebras.
\newblock {\em Communications in Mathematical Physics}, 26(4):280--289, 1972.

\bibitem{DBLP:journals/entcs/Sato16}
Tetsuya Sato.
\newblock Approximate relational {Hoare} logic for continuous random samplings.
\newblock In Lars Birkedal, editor, {\em The Thirty-second Conference on the Mathematical Foundations of Programming Semantics, {MFPS} 2016, Carnegie Mellon University, Pittsburgh, PA, USA, May 23-26, 2016}, volume 325 of {\em Electronic Notes in Theoretical Computer Science}, pages 277--298. Elsevier, 2016.
\newblock \href {https://doi.org/10.1016/J.ENTCS.2016.09.043} {\path{doi:10.1016/J.ENTCS.2016.09.043}}.

\bibitem{schmudgen2012unbounded}
Konrad Schm{\"u}dgen.
\newblock {\em {Unbounded self-adjoint operators on Hilbert space}}, volume 265.
\newblock Springer Science \& Business Media, 2012.

\bibitem{Simon_1978}
Barry Simon.
\newblock Lower semicontinuhy of positive quadratic forms.
\newblock {\em Proceedings of the Royal Society of Edinburgh: Section A Mathematics}, 79(3–4):267–273, 1978.
\newblock \href {https://doi.org/10.1017/S0308210500019776} {\path{doi:10.1017/S0308210500019776}}.

\bibitem{Unr19a}
Dominique Unruh.
\newblock Quantum relational {H}oare logic.
\newblock {\em Proc. ACM Program. Lang.}, 3(POPL), January 2019.
\newblock \href {https://doi.org/10.1145/3290346} {\path{doi:10.1145/3290346}}.

\bibitem{Vil03}
C{\'e}dric Villani.
\newblock {\em Topics in Optimal Transportation}.
\newblock Number~58 in Graduate Studies in Mathematics. American Mathematical Society, Providence (R.I.), 2003.

\bibitem{Villani08}
C{\'e}dric Villani.
\newblock {\em Optimal transport: {O}ld and new}.
\newblock Springer, 2008.

\bibitem{Wir21}
Melchior Wirth.
\newblock A {{Noncommutative Transport Metric}} and {{Symmetric Quantum Markov Semigroups}} as {{Gradient Flows}} of the {{Entropy}}, August 2021.
\newblock \href {https://arxiv.org/abs/1808.05419} {\path{arXiv:1808.05419}}.

\bibitem{Wir22}
Melchior Wirth.
\newblock A {{Dual Formula}} for the {{Noncommutative Transport Distance}}.
\newblock {\em Journal of Statistical Physics}, 187(2):19, May 2022.
\newblock \href {https://doi.org/10.1007/s10955-022-02911-9} {\path{doi:10.1007/s10955-022-02911-9}}.

\bibitem{WZ21}
Melchior Wirth and Haonan Zhang.
\newblock Complete {{Gradient Estimates}} of {{Quantum Markov Semigroups}}.
\newblock {\em Communications in Mathematical Physics}, 387(2):761--791, October 2021.
\newblock \href {https://doi.org/10.1007/s00220-021-04199-4} {\path{doi:10.1007/s00220-021-04199-4}}.

\bibitem{YXX25}
Jiaqi Yang, Wei Xie, and Xiaohua Xu.
\newblock Stability and {{Generalization}} of {{Quantum Neural Networks}}, February 2025.
\newblock \href {https://arxiv.org/abs/2501.12737} {\path{arXiv:2501.12737}}.

\bibitem{Ying11}
Mingsheng Ying.
\newblock Floyd--{Hoare} logic for quantum programs.
\newblock {\em ACM Transactions on Programming Languages and Systems (TOPLAS)}, 33(6):19:1--19:49, 2011.
\newblock \href {https://doi.org/10.1145/2049706.2049708} {\path{doi:10.1145/2049706.2049708}}.

\bibitem{Ying16}
Mingsheng Ying.
\newblock {\em Foundations of Quantum Programming}.
\newblock Morgan Kaufmann, Cambridge, MA, USA, second edition, 2024.

\bibitem{ZYYY22}
Li~Zhou, Nengkun Yu, Shenggang Ying, and Mingsheng Ying.
\newblock Quantum earth mover's distance, a no-go quantum {{Kantorovich}}--{{Rubinstein}} theorem, and quantum marginal problem.
\newblock {\em Journal of Mathematical Physics}, 63(10):102201, October 2022.
\newblock \href {https://doi.org/10.1063/5.0068344} {\path{doi:10.1063/5.0068344}}.

\bibitem{ZS98}
Karol Zyczkowski and Wojeciech Slomczynski.
\newblock The {{Monge}} distance between quantum states.
\newblock {\em Journal of Physics A: Mathematical and General}, 31(45):9095--9104, November 1998.
\newblock \href {https://doi.org/10.1088/0305-4470/31/45/009} {\path{doi:10.1088/0305-4470/31/45/009}}.

\end{thebibliography}

\appendix

\section{Mathematical Preliminaries on Linear Relations and Unbounded Observables}
\label{sec:math_prelim}

In this section, we provide a rigorous mathematical foundation for handling unbounded quantum observables. While bounded observables are described by bounded self-adjoint operators, unbounded observables (such as energy or execution time) require the more general framework of \emph{Linear Relations} (LRs) to properly handle domain issues and singularities (infinite values).

\subsection{Linear Relations: Beyond Unbounded Operators}
\label{subsec:lr_definitions}

A linear relation $T$ on a Hilbert space $\mathcal{H}$ is defined as a closed linear subspace of the direct sum $\mathcal{H} \oplus \mathcal{H}$. We identify an operator $A$ with its graph $\mathcal{G}(A) = \{(x, Ax) \mid x \in \Dom(A)\}$. LRs generalize operators by allowing ``multivalued'' behavior.

\begin{definition}[Basic Notions]
    Let $T \subset \mathcal{H} \oplus \mathcal{H}$ be a linear relation.
    \begin{itemize}
        \item \textbf{Domain and Range}: $\Dom(T) = \{x \mid \exists y, (x,y) \in T\}$, $\Ran(T) = \{y \mid \exists x, (x,y) \in T\}$.
        \item \textbf{Kernel and multivalued Part}: $\ker(T) = \{x \mid (x,0) \in T\}$, $\mul(T) = \{y \mid (0,y) \in T\}$.
        \item \textbf{Inverse}: $T^{-1} = \{(y,x) \mid (x,y) \in T\}$. Note that $\ker(T) = \mul(T^{-1})$.
    \end{itemize}
    $T$ is an operator if and only if $\mul(T) = \{0\}$.
\end{definition}

\begin{definition}[Closed and Densely Defined Relations]
    Let $T \subset \mathcal{H} \oplus \mathcal{H}$ be a linear relation.
    \begin{itemize}
        \item \textbf{Closed Linear Relation}: $T$ is called \emph{closed} if it is a closed subspace of $\mathcal{H} \oplus \mathcal{H}$ with respect to the product topology. (Note: A closed operator has a closed graph).
        \item \textbf{Densely Defined}: $T$ is \emph{densely defined} if its domain $\Dom(T) = \{x \mid \exists y, (x,y) \in T\}$ is dense in $\mathcal{H}$.
    \end{itemize}
\end{definition}

\begin{definition}[Adjoint and Self-Adjointness]
    The adjoint $T^*$ is defined by
    \[
    T^* \triangleq \{ (\u, \v) \in \mathcal{H} \oplus \mathcal{H} \mid \forall\, (\u', \v') \in T, \langle \v, \u' \rangle = \langle \u, \v' \rangle \}.
\]
    $T$ is \emph{symmetric} if $T \subset T^*$, and \emph{self-adjoint} if $T = T^*$.
\end{definition}

A crucial structural property of self-adjoint linear relations is the orthogonal decomposition of the Hilbert space. Since $\mul(H^*) = (\Dom(H))^\perp$ generally holds, self-adjointness ($H=H^*$) implies the following decomposition.

\begin{proposition}[Canonical Decomposition, \cite{cross1998multivalued,behrndt2020boundary}]
    Let $H$ be a self-adjoint linear relation on $\mathcal{H}$. Then the multivalued part is just the orthogonal complement of the domain:
    \[ \mul(H) = (\Dom(H))^\perp. \]
    Consequently, the Hilbert space admits the orthogonal decomposition:
    \[ \mathcal{H} = \overline{\Dom(H)} \oplus \mul(H). \]
    With respect to this decomposition, $H$ can be uniquely split into a densely defined self-adjoint operator $H_{\operatorname{op}}$ acting on the subspace $\overline{\Dom(H)}$ and a purely multivalued component on $\mul(H)$. Thus $\mul(H)$ is always closed.
\end{proposition}

In the context of quantitative verification (e.g., execution time, energy), observables are typically non-negative or at least have a finite lower bound. We formally define this property for linear relations.

\begin{definition}[Bounded Below Linear Relations]
    A self-adjoint linear relation $H$ is said to be \emph{bounded below} if there exists a real constant $\gamma \in \mathbb{R}$ such that
    \[ \langle x, y \rangle \ge \gamma \|x\|^2, \quad \forall (x, y) \in H. \]
    In particular, if $\gamma = 0$, $H$ is called \emph{positive}.
\end{definition}

For any self-adjoint linear relation $H$ bounded below by $\gamma$, we can shift it to a positive relation $H^{\prime}\triangleq\{(x,y-\gamma x)\mid (x,y)\in H\}$. Note that $\Dom(H)=\Dom(H^{\prime})$ and $\mul(H)=\mul(H^{\prime})$. We denote $H^{\prime}$ as $\shift{H}{\gamma}$.

For bounded below relations, the spectral decomposition can be formulated without involving negative infinity, simplifying the topological structure to the one-point compactification of the semi-infinite line. 
Since we don't do not allow $-\infty$ to occur, in this appendix we use $\infty$ to denote $+\infty$. 
To formulate the spectral decomposition for relations that may include "infinite values" (singularities), we must extend the standard measure-theoretic framework to the extended real line.

\begin{definition}[Spectral Measures on the Extended Real Line]\label{def:spectral-measure}
    Let $\overline{\mathbb{R}} = \mathbb{R} \cup \{+\infty\}$ be the extended real line equipped with the standard topology (where open neighborhoods of $+\infty$ are of the form $(a, +\infty]$). Let $\mathbb{B}(\overline{\mathbb{R}})$ denote the Borel $\sigma$-algebra on $\overline{\mathbb{R}}$.
    A \emph{spectral measure} $E$ on a Hilbert space $\mathcal{H}$ is a map $E$ from $\mathbb{B}(\overline{\mathbb{R}}) $ to orthogonal projections, satisfying:
    \begin{enumerate}
        \item \textbf{Normalization}: $E(\overline{\mathbb{R}}) = I$ (the identity operator).
        \item \textbf{Orthogonality}: $E(\Omega_1 \cap \Omega_2) = E(\Omega_1) E(\Omega_2)$ for any $\Omega_1, \Omega_2 \in \mathbb{B}(\overline{\mathbb{R}})$.
        \item \textbf{$\sigma$-additivity}: For any countable family of disjoint sets $\{\Omega_i\}$, $E(\bigcup_i \Omega_i) = \sum_i E(\Omega_i)$ in the strong operator topology.
    \end{enumerate}
    Crucially, a spectral measure on $\overline{\mathbb{R}}$ can have a non-zero projection at infinity, $E(\{+\infty\})$, which captures the singular or multivalued component of the associated linear relation.
\end{definition}

Given a spectral measure $E$, for any vector $\u \in \mathcal{H}$, the map $\Omega \mapsto \langle \u, E(\Omega)\u \rangle = \|E(\Omega)\u\|^2$ defines a finite, non-negative scalar measure on $\overline{\mathbb{R}}$. We denote this induced measure by $\mu_\u$.
Consequently, for any measurable function $f: \overline{\mathbb{R}} \to \mathbb{C}$, the operator $f(H)$ is rigorously defined via integration with respect to these scalar measures:
\[ \langle \u, f(H) \v \rangle := \int_{\overline{\mathbb{R}}} f(\lambda) \, d\langle \u, E(\lambda)\v \rangle, \]
where the domain consists of vectors for which the integral $|f(\lambda)|^2$ is finite. In our shorthand notation, $d\|E_H(\lambda)\u\|^2$ refers to integration with respect to the scalar measure $\mu_\u$.

\begin{theorem}[Spectral Theorem for Bounded Below Relations, \cite{cross1998multivalued,schmudgen2012unbounded,behrndt2020boundary}]\label{thm:spectral-theorem-for-LR}
    Let $H$ be a self-adjoint linear relation bounded below by $\gamma$. There exists a unique spectral measure $E_H(\cdot)$ on the extended real interval $[\gamma, \infty]$ (as a subset of $\mathbb{R} \cup \{+\infty\}$) such that:
    \begin{enumerate}
        \item The relation is represented by the spectral integral:
        \[ H = \int_{[\gamma, \infty]} \lambda \, dE_H(\lambda). \]
        Specifically, for any pair $(x,y) \in H$, we have $\langle x,y \rangle = \int_{[\gamma, \infty]} \lambda \, d\|E_H(\lambda)x\|^2$.
        \item The multivalued part corresponds precisely to the eigenspace at positive infinity: $\mul(H) = \Ran(E_H(\{+\infty\}))$.
    \end{enumerate}
\end{theorem}

\begin{corollary}
    For any positive linear relation $H=\int_{[0,\infty]} \lambda \, dE_H(\lambda)$, there exists a unique positive linear relation $\sqrt{H}=\int_{[0,\infty]} \sqrt{\lambda} \, dE_H(\lambda)$, called the square root of $H$. Note that $\Dom(\sqrt{H})$ can be strictly larger then $\Dom(H)$, and $\Dom(\sqrt{H})$ can be characterized by \cite[Proposition 5.16]{schmudgen2012unbounded}.
\end{corollary}

\begin{definition}
    For any self-adjoint linear relation $H$ bounded below by $\gamma$, and a positive number $c\in[0,+\infty)$, we denote
    \[
    cH\triangleq\{(x,cy)\in\cH\oplus\cH\mid(x,y)\in H\}.
    \]
    Note that $cH$ is also a self-adjoint linear relation $H$ bounded below by $c\gamma$. If $H = \int_{[\gamma, \infty]} \lambda \, dE_H(\lambda)$, it's obvious that $cH = \int_{[c\gamma, \infty]} c\lambda \, dE_H(\lambda)$.
\end{definition}

\subsection{Quadratic Forms and Generalized Sums}
\label{subsec:forms_and_sums}

The arithmetic sum of two unbounded linear relations (or operators) $A+B$ is often ill-defined due to the intersection of their domains being too small or empty. To rigorously define sums of observables (e.g., total energy or time), we rely on the correspondence between linear relations and quadratic forms, which relaxes the domain constraints of operator algebra.

\begin{definition}[Sesquilinear and Quadratic Forms]
    A \emph{sesquilinear form} is a map $\mathfrak{s}: \Dom_{\mathfrak{s}} \times \Dom_{\mathfrak{s}} \to \mathbb{C}$, where $\Dom_{\mathfrak{s}}$ is a linear subspace of $\mathcal{H}$, such that $\mathfrak{s}(\cdot, \cdot)$ is linear in the first argument and conjugate-linear in the second.
    
    The associated \emph{quadratic form} $\mathfrak{t}: \Dom_{\mathfrak{t}} \to \mathbb{C}$ (with $\Dom_{\mathfrak{t}} = \Dom_{\mathfrak{s}}$) is defined by $\mathfrak{t}(\u) := \mathfrak{s}(\u, \u)$.
    Conversely, the sesquilinear form can be recovered from the quadratic form via the \emph{polarization identity}:
    \[ 
        \mathfrak{s}(\u, \v) = \frac{1}{4} \sum_{k=0}^3 i^{-k} \mathfrak{t}(\u + i^k \v). 
    \]
    In the future, whenever it does not cause confusion, we will use \( \mathfrak{t}(\u, \v) \) to denote the sesquilinear form induced by the quadratic form \( \mathfrak{t} \).
\end{definition}

\begin{definition}[Semiboundedness and Closedness]
    A quadratic form $\mathfrak{t}$ is called:
    \begin{itemize}
        \item \textbf{Hermitian} (or \emph{symmetric}):   if it takes real values, i.e., $\mathfrak{t}(\u) \in \mathbb{R}$ for all $\u \in \Dom_{\mathfrak{t}}$.
        
        \item \textbf{Semibounded (Bounded Below)}: if there exists $\gamma \in \mathbb{R}$ such that $\mathfrak{t}(\u) \ge \gamma \|\u\|^2$ for all $\u \in \Dom_{\mathfrak{t}}$.
        \item \textbf{Closed}: if $\mathfrak{t}$ is semibounded and its domain $\Dom_{\mathfrak{t}}$ is complete with respect to the form norm $\|\u\|_{\mathfrak{t}} := (\mathfrak{t}(\u) + (1-\gamma)\|\u\|^2)^{1/2}$. Equivalently, $\mathfrak{t}$ is closed if it is lower semi-continuous in the Hilbert space topology.
    \end{itemize}
\end{definition}

For verification purposes, we work with observables defined on the entire Hilbert space but taking infinite values. 
Note that for a semibounded closed Hermitian quadratic form $\mathfrak{t}$, if $\u\notin\Dom_{\mathfrak{t}}$, we always denote $t(\u)=+\infty$.
This motivates the following construction via spectral measures.

\begin{definition}[Extended Quadratic Form from Linear Relations]
    Let $H$ be a self-adjoint linear relation bounded below by $\gamma$. The \emph{extended quadratic form} $\mathfrak{t}_H: \mathcal{H} \to \mathbb{R} \cup \{+\infty\}$ is defined by the spectral integral:
    \[ 
        \mathfrak{t}_H(\u) := \int_{[\gamma, \infty]} \lambda \, d\|E_H(\lambda)\u\|^2. 
    \]
    The \emph{quadratic domain} of this form is the subspace where the integral is finite: $\Dom_{\mathfrak{t}}(H)= \{ \u \in \mathcal{H} \mid \mathfrak{t}_H(\u) < \infty \}$.
    Note that $\mathfrak{t}_H(\u) = \infty$ if $\u$ has any overlap with the multivalued part $\mul(H)$ or the integral on $\bR$ diverges, and $\Dom_{\mathfrak{t}}(H)$ can be strictly larger than $\Dom(H)$. In fact, it is not hard to prove that $\Dom_{\mathfrak{t}}(H)=\Dom_{\mathfrak{t}}(\shift{H}{\gamma})=\Dom(\sqrt{\shift{H}{\gamma}})$.
\end{definition}

The following theorem connects the abstract forms to our spectral definitions, justifying the form sum operation.

\begin{theorem}[First Representation Theorem, \cite{kato2013perturbation,behrndt2020boundary}]
    \label{thm:first_representation}
    There is a one-to-one correspondence between (either densely or non-densely defined) closed semibounded quadratic forms and bounded below self-adjoint linear relations.
    Specifically, for any closed semibounded Hermitian form $\mathfrak{t}$, there exists a unique bounded below self-adjoint relation $H$ such that $\mathfrak{t} = \mathfrak{t}_H$ (in the sense of the extended definition above).
\end{theorem}
\begin{proof}[Proof Sketch]
    The correspondence is established via the Riesz representation theorem in a rigged Hilbert space setting.
    
    $(\Rightarrow)$ Given a bounded below self-adjoint relation $H$, the spectral theorem provides a unique spectral measure $E_H$. The form is explicitly constructed via the integral $\mathfrak{t}(\u) = \int \lambda \, d\|E_H(\lambda)\u\|^2$, which is known to be closed and semibounded.
    
    $(\Leftarrow)$ Conversely, let $\mathfrak{t}$ be a semibounded closed form. We define the relation $H$ directly as the set of pairs satisfying the representation condition:
    \[ 
        H = \left\{ (x, y) \in \Dom_{\mathfrak{t}} \times \mathcal{H} \mid \forall u \in \Dom_{\mathfrak{t}}, \mathfrak{t}(x, u) = \langle y, u \rangle \right\}. 
    \]
    Standard results in perturbation theory (see \cite[Theorem VI.2.1]{kato2013perturbation}) confirm that this set forms a self-adjoint linear relation bounded below, with its multivalued part $\mul(H)$ precisely equal to the orthogonal complement of the form domain $\Dom(\mathfrak{t})^\perp$.
\end{proof}

Therefore, for a bounded below linear relation $H$, it's easy to see that $\mathfrak{t}_{cH}(x)=c\cdot\mathfrak{t}_{H}(x)$ holds for every $x\in\cH$. 

The correspondence between relations and forms allows us to rigorously extend the standard order on bounded operators to the unbounded and multivalued setting.

\begin{definition}[Extended L\"owner Order]
    \label{def:extended_lowner_order}
    Let $A$ and $B$ be self-adjoint linear relations which are both bounded below by $m$. We define the partial order $A \sqsubseteq B$ via their associated extended quadratic forms:
    \[ A \sqsubseteq B \iff \mathfrak{t}_A(\u) \le \mathfrak{t}_B(\u), \quad \forall \u \in \mathcal{H}. \]
    This inequality holds pointwise on the extended real line $[m, \infty]$. Note that formally $A \sqsubseteq B$ implies the inclusion of form domains: $\Dom(\mathfrak{t}_B) \subseteq \Dom(\mathfrak{t}_A)$.
\end{definition}


This correspondence allows us to define the sum of relations via the sum of their forms, bypassing domain issues in the operator space.

\begin{definition}[Form Sum of Linear Relations]
    Let $A$ and $B$ be two self-adjoint linear relations bounded below. Let $\mathfrak{t}_A$ and $\mathfrak{t}_B$ be their associated closed semibounded forms (via Theorem \ref{thm:first_representation}).
    We define the \emph{form sum} $A \dot{+} B$ as the unique self-adjoint linear relation associated with the sum of forms $\mathfrak{t} := \mathfrak{t}_A + \mathfrak{t}_B$, where:
    \[ \Dom_{\mathfrak{t}} = \Dom_{\mathfrak{t}}(A) \cap \Dom_{\mathfrak{t}}(B), \quad \mathfrak{t}(\u) = \mathfrak{t}_A(\u) + \mathfrak{t}_B(\u). \]
\end{definition}

\begin{proposition}[Well-definedness]
    The form sum $A \dot{+} B$ is a well-defined self-adjoint linear relation bounded below.
\end{proposition}

\begin{proof}
    Since $A$ and $B$ are bounded below, their associated forms $\mathfrak{t}_A$ and $\mathfrak{t}_B$ are bounded below and closed (i.e., lower semi-continuous in the Hilbert space topology).
    
    Consider the sum form $\mathfrak{t} = \mathfrak{t}_A + \mathfrak{t}_B$.
    \begin{itemize}
        \item \textbf{Semiboundedness}: If $\mathfrak{t}_A(\u) \ge \gamma_A \|\u\|^2$ and $\mathfrak{t}_B(\u) \ge \gamma_B \|\u\|^2$, then $\mathfrak{t}(\u) \ge (\gamma_A + \gamma_B) \|\u\|^2$. Thus $\mathfrak{t}$ is bounded below.
        \item \textbf{Closedness}: Since $\mathfrak{t}_A$ and $\mathfrak{t}_B$ are closed, they are lower semi-continuous functions. The sum of two lower semi-continuous functions is also lower semi-continuous. Therefore, the form $\mathfrak{t}$ is closed.
    \end{itemize}
    
    Since $\mathfrak{t}$ is a closed semibounded form, by the First Representation Theorem (Theorem \ref{thm:first_representation}), there exists a unique semibounded self-adjoint linear relation $S$ such that $\mathfrak{t}_S = \mathfrak{t}$. We define $A \dot{+} B := S$. This guarantees the existence and uniqueness of the form sum.
\end{proof}

\subsection{Strong Resolvent Convergence (SRC)}
\label{subsec:src}

To discuss the convergence of observables (e.g., in loop iterations), we require a topology that behaves well for unbounded objects. The standard operator norm topology is inapplicable because the difference between two unbounded operators is generally undefined or unbounded. Instead, we look at their “bounded inverses”.

\begin{definition}[Resolvent Operator]
    Let $H$ be a self-adjoint linear relation and $z \in \mathbb{C} \setminus \mathbb{R}$ be a complex number with non-zero imaginary part. The \emph{resolvent operator} $R_H(z)$ is defined explicitly as the set of pairs:
    \[ R_H(z)=(H-z)^{-1} := \left\{ (y, x) \in \mathcal{H} \oplus \mathcal{H} \;\middle|\; (x, y + z x) \in H \right\}. \]
\end{definition}

\begin{remark}[Why the Resolvent is Bounded]
    While $H$ itself may be unbounded or multivalued (containing ``infinite'' eigenvalues), the resolvent $R_H(z)$ transforms the problem into a bounded setting.
    Intuitively, since $H$ is self-adjoint, its ``spectrum'' (generalized eigenvalues) lies entirely on the real line $\mathbb{R}$. Because $z$ is chosen from the complex plane ($\Im(z) \neq 0$), it is strictly separated from the spectrum of $H$. In linear algebra terms, the matrix $(H - zI)$ is invertible because $z$ is not an eigenvalue.
    Formally, for any self-adjoint relation, $R_H(z)$ is always a bounded linear operator defined on the entire Hilbert space $\mathcal{H}$, satisfying the bound:
    \[ \|R_H(z)\| \le \frac{1}{|\Im(z)|}. \]
    This allows us to reduce the convergence of pathological relations to the convergence of well-behaved bounded operators.
\end{remark}

\begin{definition}[Strong Resolvent Convergence]
    A sequence of self-adjoint relations $H_n$ is said to converge to $H$ in the strong resolvent sense ($H_n \xrightarrow{SRC} H$) if for some (and hence all) $z \in \mathbb{C} \setminus \mathbb{R}$:
    \[ R_{H_n}(z)\u \to R_H(z)\u, \quad \forall \u \in \mathcal{H} \]
    in the norm topology of $\mathcal{H}$.
\end{definition}

The following theorem connects the monotonic order of observables with SRC, justifying our use of limits.

\begin{theorem}[Monotonicity and Convergence]
    \label{thm:monotonicity_convergence}
    Let $\{H_n\}$ be a sequence of self-adjoint relations bounded below such that $H_1 \sqsubseteq H_2 \sqsubseteq \dots$.
    Then there exists a unique self-adjoint relation $H_\infty$ such that $H_n \xrightarrow{SRC} H_\infty$.
    Furthermore, $H_{\infty}$ is the least upper bound (supremum) of the sequence with respect to $\sqsubseteq$, i.e., $H_\infty = \sup_{n} H_n$.
\end{theorem}

\begin{proof}[Proof Sketch]
    The existence of the limit is established by moving to the dual picture of quadratic forms.
    Let $\mathfrak{t}_n$ be the closed form associated with $H_n$. The monotonicity condition implies that the sequence of forms is pointwise non-decreasing.
    We define the limit form $\mathfrak{t}_{\infty}(\u) := \lim_{n \to \infty} \mathfrak{t}_{n}(\u)$ (valued in $\mathbb{R} \cup \{+\infty\}$) with the quadratic domain $\Dom_{\mathfrak{t}_{\infty}} = \{ \u \in \mathcal{H} \mid \lim \mathfrak{t}_{n}(\u) < \infty \}$.
    It is a standard result in perturbation theory (see \cite{Simon_1978} or \cite[Theorem VIII.3.11]{kato2013perturbation}) that the limit of an increasing sequence of closed forms is itself a closed form.
    By the First Representation Theorem, this limit form $\mathfrak{t}_\infty$ uniquely determines a self-adjoint relation $H_\infty$. The strong resolvent convergence also follows from the convergence of the associated forms \cite[Theorem 3.1]{Behrndt_Hassi_de_Snoo_Wietsma_2010}. Finally, We show $H_{\infty} = \sup_n H_n$.
    \begin{itemize}
        \item \textit{Upper Bound:} By definition, $\mathfrak{t}_n(\u) \le \mathfrak{t}_\infty(\u)$ for all $n$. This is equivalent to $H_n \sqsubseteq H_\infty$.
        \item \textit{Least Upper Bound:} Suppose $K$ is any upper bound, i.e., $H_n \sqsubseteq K$ for all $n$. This means $\mathfrak{t}_n(\u) \le \mathfrak{t}_K(\u)$ for all $n$. Taking the supremum over $n$, we get:
        \[ \mathfrak{t}_\infty(\u) = \sup_n \mathfrak{t}_n(\u) \le \mathfrak{t}_K(\u). \]
        This implies $H_\infty \sqsubseteq K$.
    \end{itemize}
    Thus, $H_\infty$ is the unique supremum.
\end{proof}

\begin{remark}[Intuition]
    Intuitively, Strong Resolvent Convergence ensures that the spectral properties of the sequence behave continuously. For an increasing sequence of observables (representing, e.g., accumulated costs), the ``energy levels'' shift upwards. SRC guarantees that this shift leads to a well-defined limit observable $H_\infty$, whose spectral projections are the strong limits of the spectral projections of $H_n$. This excludes pathological behaviors where the spectrum might ``evaporate'' or oscillate wildly.
\end{remark}
\begin{corollary}\label{corollary:convex-convergence}
    Let $\{H_n\}$ be a sequence of self-adjoint relations uniformly bounded below by $\gamma$, and $\{a_n\}$ be a sequence of positive numbers with $\sum_n a_n=1$.
    Then there exists a self-adjoint relation $H$ bounded below by $\gamma$ such that the form sum $\sum_{k=1}^{n}a_k H_k$ convergent to $H$ in the sense of SRC. We denote $H=\sum_{n=1}^{\infty}a_n H_n$.
\end{corollary}
\begin{proof}
    This can be directly obtained by \Cref{thm:first_representation} and \Cref{thm:monotonicity_convergence}.
\end{proof}

Finally, to connect our theoretical framework with practical verification techniques, we define the truncation of a linear relation. Consistent with standard program semantics (e.g., timeout or saturation), we cap the value at $n$ rather than projecting it to zero.

\begin{definition}[Truncation of Linear Relations]
    \label{def:truncation}
    Let $H$ be a self-adjoint linear relation bounded below by $\gamma$. For any real number $n > \gamma$, the \emph{truncation} $H^{(n)}$ is defined as the bounded self-adjoint operator:
    \[ H^{(n)} := \int_{[\gamma, n]} \lambda \, dE_H(\lambda) + n \cdot E_H((n, \infty]). \]
    Intuitively, this operator behaves as $H$ where the value is small, and saturates to the constant $n$ on the subspace corresponding to large or infinite values (including the multivalued part).
\end{definition}

\begin{proposition}[Convergence of Truncations]
    \label{prop:truncation_src}
    The sequence of truncations $\{H^{(n)}\}_{n \in \mathbb{N}}$ converges to $H$ in the strong resolvent sense:
    \[ H^{(n)} \xrightarrow{SRC} H \quad \text{as } n \to \infty. \]
\end{proposition}

\begin{proof}
    Fix $z \in \mathbb{C} \setminus \mathbb{R}$ and $\u \in \mathcal{H}$. We compare the resolvents.
    The resolvent of the truncation $H^{(n)}$ acts on the two spectral subspaces as:
    \[ R_{H^{(n)}}(z)\u = \int_{[\gamma, n]} \frac{1}{\lambda - z} \, dE_H(\lambda)\u + \frac{1}{n - z} E_H((n, \infty])\u. \]
    The resolvent of the original relation $H$ is:
    \[ R_{H}(z)\u = \int_{[\gamma, \infty]} \frac{1}{\lambda - z} \, dE_H(\lambda)\u = \int_{[\gamma, n]} \frac{1}{\lambda - z} \, dE_H(\lambda)\u + \int_{(n, \infty]} \frac{1}{\lambda - z} \, dE_H(\lambda)\u. \]
    Subtracting the two, the terms on $[\gamma, n]$ cancel out. The squared norm of the difference is determined purely by the tail integral:
    \[ 
        \| (R_H(z) - R_{H^{(n)}}(z))\u \|^2 = \int_{(n, \infty]} \left| \frac{1}{\lambda - z} - \frac{1}{n - z} \right|^2 \, d\|E_H(\lambda)\u\|^2. 
    \]
    We define the integrand function $f_n(\lambda) = \left| \frac{1}{\lambda - z} - \frac{1}{n - z} \right| \cdot \mathbb{I}_{(n, \infty]}(\lambda)$, where we adopt the convention $1/(\infty - z) = 0$.
    \begin{itemize}
        \item \textbf{Pointwise Convergence}: For any fixed finite $\lambda$, eventually $n > \lambda$, so $\lambda \notin (n, \infty]$ and $f_n(\lambda) = 0$. For $\lambda = +\infty$, $f_n(+\infty) = |0 - \frac{1}{n-z}| \to 0$ as $n \to \infty$. Thus $f_n(\lambda) \to 0$ pointwise everywhere on $[\gamma, \infty]$.
        \item \textbf{Domination}: The term $|\frac{1}{\lambda - z}|$ is bounded by $1/|\Im(z)|$, and $|\frac{1}{n - z}|$ is similarly bounded. Thus, $|f_n(\lambda)| \le \frac{2}{|\Im(z)|}$, which is integrable with respect to the finite measure $d\|E_H(\cdot)\u\|^2$.
    \end{itemize}
    By the Lebesgue Dominated Convergence Theorem for abstract measures (see, e.g., \cite[Theorem 1.34]{rudin1987real}), the integral vanishes as $n \to \infty$. Therefore, $R_{H^{(n)}}(z)\u \to R_H(z)\u$.
\end{proof}
\begin{proposition}[Monotonicity of Truncations]
    \label{prop:truncation_monotonicity}
    The sequence of truncations is monotonically increasing and bounded above by $H$ in the extended Löwner order. Specifically, for any $\gamma < n < m$:
    \[ H^{(n)} \sqsubseteq H^{(m)} \sqsubseteq H. \]
\end{proposition}

\begin{proof}
    This follows from the functional calculus. The truncation $H^{(n)}$ corresponds to the function $f_n(\lambda) = \min(\lambda, n)$ (defined on $[\gamma, \infty]$). Since $n < m$, we have $f_n(\lambda) \le f_m(\lambda) \le \lambda$ for all $\lambda \in [\gamma, \infty]$. The operator inequality follows immediately from the order-preserving property of the spectral calculus.
\end{proof}

\subsection{Extended Trace and Expectation Values}
\label{subsec:extended_trace}

Standard quantum mechanics defines the expectation value via $\Tr(A\rho)$. For unbounded observables or linear relations, we require a rigorous definition that handles infinite energy and singularities consistently.

\subsubsection{Definitions and Equivalence}

We provide two equivalent definitions for the expectation value. The first is based on spectral measure theory (coordinate-free), and the second is based on basis expansion (computational).

\begin{definition}[Definition via Spectral Integral]
    \label{def:trace_integral}
    Let $H$ be a self-adjoint linear relation bounded below by $\gamma$, and $\rho$ be a partial density operator on $\cH$. Let $E_H$ be the spectral measure of $H$ on the extended real line $\bR\cup\{+\infty\}$. The \emph{extended trace} is defined as the Lebesgue-Stieltjes integral with respect to the induced measure $\mu_\rho^H(\Omega) = \tr(E_H(\Omega)\rho)$:
    \[
        \Tr(H\rho) := \int_{[\gamma, \infty]} \lambda \, d\tr(E_H(\lambda)\rho).
    \]
    This value is well-defined in $\mathbb{R} \cup \{+\infty\}$. Specifically, if $\rho$ has overlap with the multivalued part $\mul(H)$, the term $\infty \cdot \tr(E_H(\{+\infty\})\rho)$ ensures the expectation diverges correctly.
\end{definition}

Alternatively, one can define the trace by inspecting the action of the quadratic form on the eigenstates of $\rho$. This corresponds to the intuitive "domain check" procedure in program verification.

\begin{definition}[Definition via Basis Expansion / Domain Check]
    \label{def:trace_basis}
    Let $\rho = \sum_{k} p_k |\u_k\rangle\langle\u_k|$ be the spectral decomposition of $\rho$, where $\{|\u_k\rangle\}$ is an orthonormal basis of eigenvectors and $p_k \ge 0$. Let $\mathfrak{t}_H$ be the closed quadratic form associated with $H$. The extended trace is:
    \[
        \Tr(H\rho) := \sum_{k} p_k \cdot \mathfrak{t}_H(\u_k),
    \]
    where we adopt the convention that $\mathfrak{t}_H(\u) = +\infty$ if $\u \notin \Dom_{\mathfrak{t}}(H)$ (which includes the case where $\u$ has a component in $\mul(H)$).
\end{definition}

\begin{proposition}[Equivalence and Basis Independence]
    The definitions \ref{def:trace_integral} and \ref{def:trace_basis} are equivalent. 
\end{proposition}

\begin{proof}
    By the spectral theorem, the form value is $\mathfrak{t}_H(\u_k) = \int_{[\gamma, \infty]} \lambda \, d\|E_H(\lambda)\u_k\|^2$. Substituting this into the sum:
    \[
        \text{Sum} = \sum_k p_k \int_{[\gamma, \infty]} \lambda \, d\langle \u_k, E_H(\lambda)\u_k \rangle.
    \]
    To justify swapping the sum and integral, we consider the product measure on $\mathbb{N} \times [\gamma, \infty]$ (counting measure $\times$ spectral measure). Since the measures are $\sigma$-finite, we apply Rudin's Fubini Theorem \cite[Theorem 8.8]{rudin1987real}. 
    
    We assume $\gamma < 0$ without loss of generality (if $\gamma \ge 0$, the negative integration interval is empty, and the proof reduces entirely to the positive part). We split the domain into two parts:
    
    1. \textbf{Negative Part $[\gamma, 0)$}: The integrand $\lambda$ is bounded, and the measure is finite. Thus, the integral converges absolutely:
    \[ \sum_k p_k \int_{[\gamma, 0)} |\lambda| \, d\mu_{\u_k} < \infty. \]
    By Fubini's theorem for integrable functions (specifically \cite[Theorem 8.8(c)]{rudin1987real}), we can legally interchange the sum and integral.
    
    2. \textbf{Positive Part $[0, \infty]$}: The integrand $\lambda$ is non-negative. By Fubini-Tonelli's theorem (specifically \cite[Theorem 8.8(a)]{rudin1987real}, the non-negative case), the interchange is valid regardless of whether the value is finite or infinite.
    
    Combining both parts via linearity:
    \[
        \text{Sum} = \int_{[\gamma, \infty]} \lambda \, d\left( \sum_k p_k \langle \u_k, E_H(\lambda)\u_k \rangle \right) = \int_{[\gamma, \infty]} \lambda \, d\tr(E_H(\lambda)\rho),
    \]
    which recovers the spectral integral definition.
\end{proof}

\subsubsection{Linearity Properties}

The extended trace behaves linearly, respecting the lower bounds.

\begin{proposition}[Linearity]
    \label{prop:trace_linearity}
    The extended trace satisfies:
    \begin{enumerate}
        \item[(1)] \textbf{Linearity in State:} For partial density operators $\rho_1, \rho_2 $ and scalars $0\leq c_1, c_2 \leq 1$:
        \[ \Tr(H (c_1 \rho_1 + c_2 \rho_2)) = c_1 \Tr(H \rho_1) + c_2 \Tr(H \rho_2). \]
        \item[(2)] \textbf{Additivity in Observable:} For self-adjoint relations $A, B$ bounded below (where the form sum $A \dot{+} B$ is defined):
        \[ \Tr((A \dot{+} B)\rho) = \Tr(A\rho) + \Tr(B\rho). \]
        \item[(3)] \textbf{Convexity in Observable:} Let $\{H_n\}$ be a sequence of self-adjoint relations uniformly bounded below by $\gamma$, and $\{a_n\}$ be a sequence of positive numbers with $\sum_n a_n=1$. Then
        \[ \Tr((\sum_{n=1}^{\infty}a_n H_n)\rho) = \sum_{n=1}^{\infty}a_n\Tr(H_n\rho). \]
    \end{enumerate}
\end{proposition}
\begin{proof}
    For (1), notice that the induced measure is linear in $\rho$: $\mu_{c_1\rho_1 + c_2\rho_2}(\Omega) = c_1\mu_{\rho_1}(\Omega) + c_2\mu_{\rho_2}(\Omega)$. The result follows from the linearity of the Lebesgue-Stieltjes integral with respect to the measure.
    
    For (2), we employ the equivalent basis definition (\ref{def:trace_basis}). Let $\rho = \sum p_k |\u_k\rangle\langle\u_k|$. By the definition of the form sum, $\mathfrak{t}_{A \dot{+} B}(\u) = \mathfrak{t}_A(\u) + \mathfrak{t}_B(\u)$ for any vector $\u$ (with the convention that sums involving $+\infty$ are $+\infty$). Thus:
    \[
        \sum_k p_k \mathfrak{t}_{A \dot{+} B}(\u_k) = \sum_k p_k (\mathfrak{t}_A(\u_k) + \mathfrak{t}_B(\u_k)) = \sum_k p_k \mathfrak{t}_A(\u_k) + \sum_k p_k \mathfrak{t}_B(\u_k).
    \]
    The rearrangement is valid because the terms are bounded below.

    For (3), without loss of generality, we assume $\gamma \ge 0$ (otherwise, replace $H_n$ with $H_n - \gamma I \ge 0$).
    From \Cref{corollary:convex-convergence}, $H=\sum_{n=1}^{\infty}a_nH_n$ is a well-defined self-adjoint relation.
    Let $\rho = \sum_{k} \lambda_k |\u_k\rangle\langle\u_k|$ be the spectral decomposition of the state $\rho$, where $\lambda_k \ge 0$ and $\sum \lambda_k = 1$.
    By the definition of the extended trace in terms of quadratic forms:
    \begin{align*}
        \Tr\left( \left(\sum_{n=1}^{\infty} a_n H_n\right) \rho \right) 
        &= \sum_{k} \lambda_k \mathfrak{t}_{H}(\u_k) \tag{\Cref{def:trace_basis}} \\
        &= \sum_{k} \lambda_k \left( \sum_{n} a_n \mathfrak{t}_{H_n}(\u_k) \right) \tag{Form sum, \Cref{corollary:convex-convergence}} \\
        &= \sum_{k} \sum_{n} \lambda_k a_n \mathfrak{t}_{H_n}(\u_k).
    \end{align*}
    Since all terms $\mathfrak{t}_{H_n}(\u_k)$ are non-negative (due to the lower bound assumption) and the coefficients $\lambda_k, a_n$ are positive, we can swap the order of summation (Tonelli's Theorem for series):
    \begin{align*}
        \cdots &= \sum_{n} a_n \left( \sum_{k} \lambda_k \mathfrak{t}_{H_n}(\u_k) \right) \tag{Swap Sums, see \cite{rudin1987real}} \\
        &= \sum_{n} a_n \Tr(H_n \rho). \tag{\Cref{def:trace_basis}}
    \end{align*}
    This completes the proof.
\end{proof}
\subsubsection{Convergence Theorems}

We now establish the convergence properties. Firstly we establish that the extended trace is the limit of its truncations.

\begin{lemma}[Approximation by Truncations]
    \label{lemma:trace_approx}
    Let $H$ be bounded below by $\gamma$ and $H^{(M)}$ be its truncation at $M > \gamma$. Then for any partial density operator $\rho$:
    \[ \Tr(H^{(M)}\rho) \le \Tr(H\rho) \quad \text{and} \quad \lim_{M \to \infty} \Tr(H^{(M)}\rho) = \Tr(H\rho). \]
\end{lemma}

\begin{proof}
    The truncation is defined via the functional calculus as $H^{(M)} := f_M(H)$, where $f_M(\lambda) = \min(\lambda, M)$ is a Borel function on the extended real line $\overline{\mathbb{R}}$.
    
    Recall the spectral theorem for self-adjoint relations (\Cref{thm:spectral-theorem-for-LR}), which establishes a one-to-one correspondence between $H$ and a spectral measure $E_H$ on $\overline{\mathbb{R}}$. By the definition of the functional calculus, the operator $H^{(M)}$ is given by the integral $\int \lambda\, dE_{H^{(M)}}(\lambda)$.
    Consequently, the expectation value transforms as:
    \[ \Tr(H^{(M)}\rho) = \int_{[\gamma, \infty]} \lambda \, d\tr(E_{H^{(M)}}(\lambda)\rho) = \int_{[\gamma, \infty]} f_M(\lambda) \, d\tr(E_H(\lambda)\rho). \]
    Thus, calculating $\Tr(H^{(M)}\rho)$ is equivalent to integrating the truncated function $f_M(\lambda)$ against the original measure $\mu_\rho^H(\cdot)=\tr(E_H(\cdot)\rho)$.
    
    Now we analyze the limit. For sufficiently large $M$, $f_M(\lambda) \le \lambda$ holds everywhere on the spectrum (treating $\infty$ naturally), implying $H^{(M)} \sqsubseteq H$ and $\Tr(H^{(M)}\rho) \le \Tr(H\rho)$.
    
    Furthermore, $f_M(\lambda) \nearrow \lambda$ pointwise as $M \to \infty$. To apply the Monotone Convergence Theorem (which requires non-negative functions), we consider the shifted sequence $g_M(\lambda) = f_M(\lambda) - \gamma$. Since $H$ is bounded below by $\gamma$, we have $f_M(\lambda) \ge \gamma$, so $g_M(\lambda) \ge 0$. Applying the standard MCT (\cite[Theorem 1.26]{rudin1987real}) to $g_M$:
    \[ \lim_{M \to \infty} \int_{[\gamma, \infty]} (f_M(\lambda) - \gamma) \, d\mu_\rho^H = \int_{[\gamma, \infty]} (\lambda - \gamma) \, d\mu_\rho^H. \]
    Since the measure is finite ($\mu_\rho^H(\overline{\mathbb{R}}) \le 1$), the constant term $\gamma$ is integrable and cancels from both sides, yielding $\lim \Tr(H^{(M)}\rho) = \Tr(H\rho)$.
\end{proof}

With this lemma, we prove the fundamental convergence theorems in the natural order.

\begin{theorem}[Quantum Fatou's Lemma]
    \label{thm:quantum_fatou}
    Let $\{H_n\}$ be a sequence of self-adjoint relations uniformly bounded below by $\gamma$, such that $H_n \xrightarrow{SRC} H$. For any fixed partial density operator $\rho $:
    \[ \Tr(H\rho) \le \liminf_{n \to \infty} \Tr(H_n \rho). \]
\end{theorem}

\begin{proof}
    Fix a truncation level $M > \gamma$. Let $H_n^{(M)}$ and $H^{(M)}$ be the truncations of $H_n$ and $H$ respectively, defined by the function $f_M(\lambda) = \min(\lambda, M)$.
    
    We invoke the continuity of the functional calculus with respect to strong resolvent convergence. Recall that for self-adjoint relations, SRC is defined via the strong convergence of resolvents $R_z(H_n) \to R_z(H)$, which are bounded single-valued operators.
    According to \cite[Theorem VIII.20(b)]{reed1980methods}, for any sequence of self-adjoint operators (or relations via their resolvents) converging in SRC, $f(H_n) \to f(H)$ strongly for any \emph{bounded continuous function} $f$.
    
    Since our operators are uniformly bounded below by $\gamma$, the truncation function $f_M$ restricted to the spectrum $[\gamma, \infty)$ is bounded and continuous. Thus:
    \[ H_n^{(M)} = f_M(H_n) \xrightarrow{strong} f_M(H) = H^{(M)}. \]
    
    For bounded operators, strong convergence implies the convergence of expectation values (trace) against a fixed trace-class operator $\rho$. Thus:
    \[ \lim_{n \to \infty} \Tr(H_n^{(M)}\rho) = \Tr(H^{(M)}\rho). \]
    
    By Lemma \ref{lemma:trace_approx}, we have the approximation inequality $\Tr(H_n \rho) \ge \Tr(H_n^{(M)} \rho)$. Taking the liminf:
    \[ \liminf_{n \to \infty} \Tr(H_n \rho) \ge \lim_{n \to \infty} \Tr(H_n^{(M)} \rho) = \Tr(H^{(M)} \rho). \]
    Since this holds for any $M$, letting $M \to \infty$ and applying Lemma \ref{lemma:trace_approx} again ($\Tr(H^{(M)}\rho) \to \Tr(H\rho)$) proves the theorem.
\end{proof}

\begin{theorem}[Quantum Monotone Convergence Theorem]\label{thm:quantum-monotone-convergence}
    Let $\{H_n\}$ be an increasing sequence of self-adjoint relations bounded below ($H_1 \sqsubseteq H_2 \sqsubseteq \dots$) converging to $H$ in SRC. Then for any $\rho$:
    \[ \lim_{n \to \infty} \Tr(H_n \rho) = \Tr(H \rho). \]
\end{theorem}

\begin{proof}
    By Quantum Fatou's Lemma, $\Tr(H\rho) \le \liminf \Tr(H_n\rho)$.
    Conversely, since $H_n \sqsubseteq H$ for all $n$, we have $\Tr(H_n\rho) \le \Tr(H\rho)$ by monotonicity. Thus $\limsup \Tr(H_n\rho) \le \Tr(H\rho)$.
    Combining these gives the limit.
\end{proof}

\begin{theorem}[Lower Semi-continuity with respect to State]
    \label{thm:lsc_rho}
    Let $H$ be a self-adjoint relation bounded below. If $\{\rho_n\}$ converges to $\rho$ in trace norm, then:
    \[ \Tr(H\rho) \le \liminf_{n \to \infty} \Tr(H\rho_n). \]
\end{theorem}

\begin{proof}
    Fix $M$. Since $H^{(M)}$ is bounded, the map $\sigma \mapsto \Tr(H^{(M)}\sigma)$ is continuous. Thus $\lim_n \Tr(H^{(M)}\rho_n) = \Tr(H^{(M)}\rho)$.
    Using the approximation inequality $\Tr(H\rho_n) \ge \Tr(H^{(M)}\rho_n)$:
    \[ \liminf_{n \to \infty} \Tr(H\rho_n) \ge \lim_{n \to \infty} \Tr(H^{(M)}\rho_n) = \Tr(H^{(M)}\rho). \]
    Letting $M \to \infty$, the RHS converges to $\Tr(H\rho)$ by Lemma \ref{lemma:trace_approx}.
\end{proof}

\begin{example}[Strict Inequality / Energy Escape]
    The inequality in Theorem \ref{thm:lsc_rho} can be strict.
    The expectation value can "escape" to infinity even if the state converges strongly.
    
    Consider a Hilbert space with orthonormal basis $\{|k\rangle\}_{k=0}^\infty$. Let $H$ be an unbounded operator with eigenvalues $\lambda_k = k^2$ (e.g., a discrete harmonic oscillator), so $H|k\rangle = k^2|k\rangle$.
    Let $\rho = |0\rangle\langle 0|$ be the ground state. Construct a sequence of mixed states mixing a large portion of the ground state with a tiny portion of a high-energy state:
    \[ \rho_n = \left(1 - \frac{1}{n}\right) |0\rangle\langle 0| + \frac{1}{n} |n\rangle\langle n|. \]
    
    \begin{enumerate}
        \item \textbf{State Convergence}: In trace norm, $\rho_n$ converges to $\rho$:
        \[ \|\rho_n - \rho\|_1 = \| -\frac{1}{n}|0\rangle\langle 0| + \frac{1}{n}|n\rangle\langle n| \|_1 = \frac{1}{n} + \frac{1}{n} = \frac{2}{n} \to 0. \]
        \item \textbf{Expectation Divergence}: The expectation value of $H$ in state $\rho_n$ is:
        \[ \Tr(H\rho_n) = \left(1 - \frac{1}{n}\right) \cdot 0 + \frac{1}{n} \cdot n^2 = n. \]
        As $n \to \infty$, $\Tr(H\rho_n) \to \infty$.
    \end{enumerate}
    Thus, we have a strict inequality:
    \[ \Tr(H\rho) = 0 < \infty = \liminf_{n \to \infty} \Tr(H\rho_n). \]
    This phenomenon confirms that the expectation value functional is only lower semicontinuous, not continuous, with respect to the trace norm topology.
\end{example}

\begin{theorem}[Generalized Quantum Fatou's Lemma]
    \label{thm:gen_fatou}
    Let $\{H_n\}$ be a sequence of self-adjoint relations uniformly bounded below by $\gamma$ with $H_n \xrightarrow{SRC} H$. Let $\{\rho_n\}$ be a sequence of states with $\|\rho_n - \rho\|_1 \to 0$. Then:
    \[ \Tr(H\rho) \le \liminf_{n \to \infty} \Tr(H_n \rho_n). \]
\end{theorem}

\begin{proof}
    Fix a truncation level $M > \gamma$. Let $H_n^{(M)}$ and $H^{(M)}$ be the truncations defined by $f_M(\lambda) = \min(\lambda, M)$.
    
    First, we establish a lower bound using monotonicity. Since $f_M(\lambda) \le \lambda$ on the spectrum $[\gamma, \infty)$, we have the operator inequality $H_n \sqsupseteq H_n^{(M)}$ for all $n$. Consequently, for the expectation values:
    \[ \Tr(H_n \rho_n) \ge \Tr(H_n^{(M)} \rho_n). \]
    
    Next, we analyze the convergence of the right-hand side. Note that $H_n^{(M)}$ and $H^{(M)}$ are bounded operators.
    \begin{enumerate}
        \item By the continuity of functional calculus (as argued in Theorem \ref{thm:quantum_fatou}), $H_n \xrightarrow{SRC} H$ implies $H_n^{(M)} \xrightarrow{strong} H^{(M)}$ strongly. Moreover, the sequence $\{H_n^{(M)}\}$ is uniformly bounded in operator norm by $\max(|\gamma|, M)$.
        \item The states converge in trace norm: $\|\rho_n - \rho\|_1 \to 0$.
    \end{enumerate}
    It is a standard property of the trace functional that if bounded operators $A_n \to A$ strongly (with uniform norm bound) and states $\sigma_n \to \sigma$ in trace norm, then $\Tr(A_n \sigma_n) \to \Tr(A \sigma)$. Thus:
    \[ \lim_{n \to \infty} \Tr(H_n^{(M)} \rho_n) = \Tr(H^{(M)} \rho). \]
    
    Taking the limit inferior on the inequality $\Tr(H_n \rho_n) \ge \Tr(H_n^{(M)} \rho_n)$:
    \[ \liminf_{n \to \infty} \Tr(H_n \rho_n) \ge \lim_{n \to \infty} \Tr(H_n^{(M)} \rho_n) = \Tr(H^{(M)} \rho). \]
    
    Finally, letting $M \to \infty$, the term $\Tr(H^{(M)}\rho)$ converges to $\Tr(H\rho)$ by Lemma \ref{lemma:trace_approx}.
\end{proof}

\subsection{Proof of the Main Convergence Theorem in Section \ref{Unbounded_Section}}
\label{subsec:proof_main_convergence}

We conclude this mathematical appendix by rigorously justifying the \textbf{Main Convergence Theorem} (Theorem \ref{thm:main-convergence-for-trunc}) presented in the main text. That theorem summarizes the convergence properties specifically for the sequence of truncations $A_n$.

\begin{proof}[Proof of Theorem \ref{thm:main-convergence-for-trunc}]
    Let $A$ be bounded below and $A_n$ be its truncation at $n$.
    
    \textbf{1. Monotone Convergence for Truncations:}
    This is exactly Lemma \ref{lemma:trace_approx} (Approximation by Truncations). 

     \textbf{2. Generalized Quantum Fatou Lemma:}
    This is a specific instance of Theorem \ref{thm:gen_fatou} (Generalized Quantum Fatou's Lemma).
    By Proposition \ref{prop:truncation_src} (in the main text), $A_n \xrightarrow{SRC} A$. Thus we get the conclusion by applying  Theorem \ref{thm:gen_fatou} on $A_n$.
    
    \textbf{3. Lower Semi-continuity (LSC):}
    This corresponds to Theorem \ref{thm:lsc_rho} (LSC with respect to State). 
\end{proof}
\section{Quantum Optimal Transport}

\begin{lemma}[Compactness of couplings, see~{\cite[Theorem 1.4]{qstrasseninf}}]%
    \label{lem:coupling-compact}
    For any density operators $\rho_1$ and $\rho_2$, the set of coupling, i.e., $ \couplings{\rho_1}{\rho_2}$, is compact with respect to the trace-norm topology.
\end{lemma}




\subsection{Proof of the Duality Theorem with Bounded Costs}

\begin{theorem}[Kantorovich Duality for Infinite Dimensional Quantum Systems]%
    \label{thm:quantum-Kantorovich-duality-appendix}
     Let $\mathcal{H}_1$ and $\mathcal{H}_2$ be two Hilbert spaces, 
    $\rho_1\in \mathcal{D}^1(\mathcal{H}_1)$, $\rho_2\in \mathcal{D}^1(\mathcal{H}_2)$ be two density operators, and $C$ be a bounded positive operator.
    Then, 
    \begin{equation*}
        \inf_{\rho:\coupling{\rho_1}{\rho_2}} \Tr (C\rho)
        = \sup_{\rbra{n, C_1, C_2}\in \DualSet{C} } \Tr(C_1\rho_1) + \Tr(C_2\rho_2) -n,
    \end{equation*}
    where $\DualSet{C} = \set{\rbra{n,C_1, C_2}}{C_1\otimes I + I \otimes C_2 \sqsubseteq C+ nI}$.
\end{theorem}

\begin{proof}
    In the following,
    we use $\opt{\cdot}$ to denote the optimal value of the primal optimization problem, i.e.,
    \begin{equation*}
        \opt{C} \coloneqq \inf_{\rho:\coupling{\rho_1}{\rho_2}} \Tr (C\rho),
    \end{equation*}
    and $\dopt{\cdot}$ to denote the optimal value of the dual optimization problem, i.e.,
    \begin{equation*}
        \dopt{C} \coloneqq \sup_{\rbra{n, C_1, C_2}\in \DualSet{C} } \Tr(C_1\rho_1) + \Tr(C_2\rho_2) -n.
    \end{equation*}
    Since for $i = 1,2$, $\rho_i$ is a density operator on $\mathcal{H}_i$, 
    we can write its spectral decomposition as
    \begin{equation*}
        \rho_i = \sum_{j=1}^{\infty} \lambda^{(i)}_j \ket{e_j^{(i)}}\bra{e_j^{(i)}},
    \end{equation*}
    with $\sum_{j=1}^{\infty} \lambda^{(i)}_{j} = 1$, $ \left\{ \ket{e_j^{(i)}} \right\}_{j=1}^{\infty}$ being a orthonormal basis of $\mathcal{H}_i$, and $\lambda_j^{(i)}\ge 0 $ for all integer $j$.

     Now, let $0< \delta < 1/3$ be a fixed constant. Since $\sum_{j=1}^{\infty} \lambda^{(i)}_{j} = 1$ for $i = 1,2$, there is an integer $N$ such that $\sum_{j=1}^{\infty} \lambda^{(i)}_{j} \ge 1 -\delta$ for $i = 1,2$.
     For $i = 1,2$, let $K^{(i)}$ denote $\spanspace\left\{ \ket{e_j^{(i)}} \right\}_{j=1}^{N}$, and $P_{K^{(i)}}$ denote the projector onto $K^{(i)}$.
     In addition, for $i = 1, 2$, let 
     \begin{equation*}
         \rho^{\prime}_j \coloneqq \frac{\sum_{j=1}^{N} \lambda^{(i)}_{j} \ket{e_j^{(i)}}\bra{e_j^{(i)}}}{\sum_{j=1}^{N} \lambda^{(i)}_{j}}.
     \end{equation*}
     which is a density operator, and $C'$ be the restriction of $C$ on the space $K^{(1)}\otimes K^{(2)}$.
     Now, let 
     $\opt{C^{\prime}}$ denote the optimal value of the primal optimization problem for the modified problem $C'$ and $\rho_1', \rho_2'$, i.e.,
    \begin{equation*}
        \opt{C'} \coloneqq \inf_{\rho:\coupling{\rho_1'}{\rho_2'}} \Tr (C'\rho),
    \end{equation*}
    We claim that 
    \begin{equation*}
        \opt{C} \le \opt{C'} + 2\delta \norm{C},
    \end{equation*}
    whose proof is deferred to~\Cref{prop:pertubation-primal-quantum-duality}.
    Similarly, 
    let 
     $\dopt{C^{\prime}}$ denote the optimal value of the dual optimization problem for the modified problem $C'$ and $\rho_1', \rho_2'$, i.e.,
    \begin{equation*}
        \dopt{C'} \coloneqq \sup_{\rbra{n, C_1, C_2}\in \DualSet{C'} } \Tr(C_1\rho_1') + \Tr(C_2\rho_2') -n.
    \end{equation*}
    We claim that
    \begin{equation*}
        \dopt{C} \ge \dopt{C'} - 3\sqrt{\delta}- 4\delta \norm{C}-6\sqrt{\delta}\norm{C}^2,
    \end{equation*}
    whose proof is deferred to~\Cref{prop:pertubation-dual-quantum-duality}.
    Then, by~\Cref{thm:quantum-Kantorovich-duality-finite-dim}, we have $\dopt{C'} = \opt{C'}$, giving
    \begin{equation*}
         \opt{C} \le \opt{C'} + 2\delta \norm{C} \le \dopt{C} + 3\sqrt{\delta}+ 6\delta \norm{C}+ 6\sqrt{\delta}\norm{C}^2.
    \end{equation*}
    Combined with the trivial weak duality $\opt{C}\ge \dopt{C}$, the result then follows by taking the limit $\delta\to 0$.
\end{proof}

\begin{proposition}%
\label{prop:pertubation-primal-quantum-duality}
    Let $\delta, C, C', \rho_1, \rho_1', \rho_2, \rho_2', \opt{C}$ and $\opt{C'}$ be the same as defined in the proof of~\Cref{thm:quantum-Kantorovich-duality-appendix}.
    Then, 
    \begin{equation*}
        \opt{C} \le \opt{C'} + 2\delta \norm{C}.
    \end{equation*}
\end{proposition}

\begin{proof}
    Let $\rho'$ be any coupling of $\rho_1'$ and $\rho_2'$
    on $K^{(1)}\otimes K^{(2)}$.
    Recall that 
    for $i = 1,2$, the spectral decomposition of $\rho_i$ is
    \begin{equation*}
        \rho_i = \sum_{j=1}^{\infty} \lambda^{(i)}_j \ket{e_j^{(i)}}\bra{e_j^{(i)}},
    \end{equation*}
    with $\sum_{j=1}^{\infty} \lambda^{(i)}_{j} = 1$, $ \left\{ \ket{e_j^{(i)}} \right\}_{j=1}^{\infty}$ being a orthonormal basis of $\mathcal{H}_i$, and $\lambda_j^{(i)}\ge 0 $ for all integer $j$.
    We then define a density operator $\rho$ on $\mathcal{H}_1\otimes \mathcal{H}_2$ as
    \begin{equation*}
        \rho\coloneqq \rbra*{\sum_{j=1}^N \lambda_j^{(1)}}\rbra*{\sum_{j=1}^N \lambda_j^{(2)}} \rho^{\prime} + \sum_{\max\{j_1, j_2\}\ge N} \lambda_{j_1}^{(1)}\lambda_{j_2}^{(2)} \ket{e_{j_1}^{(1)}}\bra{e_{j_1}^{(1)}}\otimes \ket{e_{j_2}^{(2)}}\bra{e_{j_2}^{(2)}}.
    \end{equation*}
    We first show $\tr_2(\rho) = \rho_1$. In fact, we have
    \begin{equation*}
        \begin{aligned}
            \tr_2(\rho) &= \sum_{j_2} \bra{e_{j_2}^{(2)}} \rho \ket{e_{j_2}^{(2)}} \\
            &= \rbra*{\sum_{j=1}^N \lambda_j^{(1)}}\rbra*{\sum_{j=1}^N \lambda_j^{(2)}} \rho_1^{\prime} +  \sum_{\max\{j_1, j_2\}\ge N} \lambda_{j_1}^{(1)}\lambda_{j_2}^{(2)} \ket{e_{j_1}^{(1)}}\bra{e_{j_1}^{(1)}} \\
             &= \rbra*{\sum_{j=1}^N \lambda_j^{(2)}} \rbra*{\sum_{j_1=1}^N \lambda_{j}^{(1)} \ket{e_{j}^{(1)}}\bra{e_{j}^{(1)}}} 
             +   \rbra*{\sum_{j=1}^{N} \lambda_j^{(2)}} \sum_{j_1 = N+1}^{\infty}\lambda_{j_1}^{(1)}\ket{e_{j_1}^{(1)}}\bra{e_{j_1}^{(1)}}  
             +  \rbra*{\sum_{j=N+1}^{\infty} \lambda_j^{(2)}} \sum_{j_1 = 1}^{\infty}\lambda_{j_1}^{(1)}\ket{e_{j_1}^{(1)}}\bra{e_{j_1}^{(1)}}\\
             &= \rbra*{\sum_{j=1}^{\infty} \lambda_j^{(2)}} \rbra*{\sum_{j_1=1}^{\infty} \lambda_{j}^{(1)} \ket{e_{j}^{(1)}}\bra{e_{j}^{(1)}}} = \rho_1.
        \end{aligned}
    \end{equation*}
    By symmetry, $\tr_1(\rho) = \rho_2$. 
    Therefore, $\rho$ is a valid coupling of $\rho_1$ and $\rho_2$.
    Moreover, we have
    \begin{equation*}
    \begin{aligned}
        \tr(C\rho) &= \rbra*{\sum_{j=1}^N \lambda_j^{(1)}}\rbra*{\sum_{j=1}^N \lambda_j^{(2)}} \tr(C\rho^{\prime}) + \sum_{\max\{j_1, j_2\}\ge N} \lambda_{j_1}^{(1)}\lambda_{j_2}^{(2)} \bra{e_{j_1}^{(1)}} \bra{e_{j_2}^{(2)}} C  \ket{e_{j_1}^{(1)}}\ket{e_{j_2}^{(2)}} \\
        & \le \tr(C\rho^{\prime} ) + \norm {C } \sum_{\max\{j_1, j_2\}\ge N} \lambda_{j_1}^{(1)}\lambda_{j_2}^{(2)} \\
        & \le \tr(C\rho^{\prime} ) + \norm {C } \rbra*{1-  \sum_{j_1, j_2\le N} \lambda_{j_1}^{(1)}\lambda_{j_2}^{(2)}} \\
        &\le \tr(C\rho^{\prime} ) + 2\delta \norm {C },
    \end{aligned}
    \end{equation*}
    where we use $\sum_{j_1, j_2\le N} \lambda_{j_1}^{(1)}\lambda_{j_2}^{(2)} \ge (1-\delta)^2\ge 1-2\delta$.
    Taking infimum on the right hand side, we know
    \begin{equation*}
        \inf_{\rho:\coupling{\rho_1}{\rho_2}} \tr(C\rho ) \le \inf_{\rho':\coupling{\rho_1'}{\rho_2'}} \tr(C\rho' ) + 2\delta,
    \end{equation*}
    which is what we want.
\end{proof}

\begin{proposition}%
\label{prop:pertubation-dual-quantum-duality}
       Let $\delta, C, C', \rho_1, \rho_1', \rho_2, \rho_2', \dopt{C}$ and $\dopt{C'}$ be the same as defined in the proof of~\Cref{thm:quantum-Kantorovich-duality-appendix}.
    Then, 
    \begin{equation*}
        \dopt{C} \ge \dopt{C'} - 3\sqrt{\delta}- 4\delta \norm{C}-6\sqrt{\delta}\norm{C}^2,
    \end{equation*}
\end{proposition}

\begin{proof}
  We first note that
      \begin{equation*}
        \dopt{C'} = \sup_{ C_1, C_2 } \Tr(C_1\rho_1') + \Tr(C_2\rho_2')
    \end{equation*}
    where $C_1$ and $C_2$ are bounded self-adjoint operators (not necessarily positive) satisfying 
    \begin{equation*}
        C_1\otimes I + I\otimes C_2 \sqsubseteq C'.
    \end{equation*}
    By~\Cref{lem:perturbation-analysis}, for a fixed $\varepsilon\in \interval{0}{1/2}$, there exist bounded self-adjoint operators $B_1'$ and $B_2'$ on $K^{(1)}$ and $K^{(2)}$ satisfying
    \begin{itemize}
        \item $B_1'\otimes I + I\otimes B_2' \sqsubseteq C'$;
        \item $\Tr(B_1'\rho_1') + \Tr(B_2' \rho_2')\ge \dopt{C'}- 2\varepsilon$;
        \item $\max\{\norm{B_1'}, \norm{B_2'}\}\le \norm{C'} + \frac{1}{\varepsilon}\norm{C'}^2$.
    \end{itemize}
    Then, we define $B_1$ and $B_2$ on $\mathcal{H}_1$ and $\mathcal{H}_2$ respectively as
    \begin{equation*}
        B_1 \triangleq B_1' - \frac{\varepsilon}{2} P_{K^{(1)}} -  \rbra*{\frac{\norm{C}^2}{\varepsilon}+\norm{B_2'}} P_{K^{(1),\perp}}, \quad
        B_2\triangleq B_2' - \frac{\varepsilon}{2} P_{K^{(2)}} -  \rbra*{\frac{\norm{C}^2}{\varepsilon}+\norm{B_1'}} P_{K^{(2),\perp}}.
    \end{equation*}
     Therefore, $B_1, B_2$ are bounded self-adjoint operators. 
     
     We claim that $B_1\otimes I + I\otimes B_2 \sqsubseteq C$.
     To show this, we prove that for any $\ket{v}\in \mathcal{H}_1\otimes \mathcal{H}_2$, it holds that 
     $\bra{v}\rbra{C-B_1\otimes I- I\otimes B_2}\ket{v}\ge 0 $.
     In fact, for any $\ket{v}\in \mathcal{H}_1\otimes \mathcal{H}_2 $
     with $\ket{v}\ne 0$, we can uniquely write it as
     $\ket{v} = \ket{v_0} + \ket{v_{\perp}}$, with $\ket{v_0}\in K^{(1)}\otimes K^{(2)}$, and $\ket{v_{\perp}} \in \rbra{K^{(1),\perp}\otimes K^{(2)}}\oplus \rbra{K^{(1)}\otimes K^{(2),\perp}}\oplus \rbra{K^{(1),\perp}\otimes K^{(2),\perp}}$.
     We have
     $\bra{v_0}B_1\otimes I\ket{v_0} = \bra{v_0}B_1'\otimes I\ket{v_0} - \bra{v_0} \frac{\varepsilon}{2} P_{K^{(1)}}\otimes I \ket{v_0}$,
     $\bra{v_0}I\otimes B_2\ket{v_0} = \bra{v_0}I\otimes B_2'\ket{v_0} - \bra{v_0} I\otimes\frac{\varepsilon}{2} P_{K^{(2)}} \ket{v_0}$,
     $\bra{v_{\perp}}B_1\otimes I\ket{v_0} = 0$,  $\bra{v_{\perp}}I\otimes B_2\ket{v_0} = 0$.
     Therefore, we obtain
     \begin{equation*}
         \begin{aligned}
             \bra{v_0}C-B_1\otimes I- I\otimes B_2 \ket{v_0} &=  \bra{v_0}C'-B_1'\otimes I- I\otimes B_2' \ket{v_0} + \bra{v_0} \frac{\varepsilon}{2} P_{K^{(1)}}\otimes I \ket{v_0}  +\bra{v_0} I\otimes \frac{\varepsilon}{2} P_{K^{(2)}} \ket{v_0} \\
             &\ge  \varepsilon \norm{\ket{v_0}}^2,
         \end{aligned}
     \end{equation*}
     and
    \begin{equation*}
        \bra{v_{\perp}} C-B_1\otimes I- I\otimes B_2 \ket{v_0} 
        =  \bra{v_{\perp}} C\ket{v_0} -  \bra{v_{\perp}} B_1\otimes I+ I\otimes B_2 \ket{v_0}  = \bra{v_{\perp}} C\ket{v_0} \le \norm{C} \norm{\ket{v_0}} \norm{\ket{v_{\perp}}}.
    \end{equation*}
    For computing $ \bra{v_{\perp}} C-B_1\otimes I- I\otimes B_2 \ket{v_{\perp}}$, we need to decompose
    $\ket{v_{\perp}} = \ket{v_{0\perp}} + \ket{v_{\perp0}}  + \ket{v_{\perp\perp}}$, where $\ket{v_{0\perp}}\in K^{(1),\perp}\otimes K^{(2)}$, 
    $\ket{v_{\perp0}}\in K^{(1)}\otimes K^{(2),\perp}$, and $\ket{v_{\perp\perp}}\in K^{(1),\perp}\otimes K^{(2),\perp}$. Then we have
    \begin{align*}
    &\quad \langle v_\perp | C - B_1 \otimes I - I \otimes B_2 | v_\perp \rangle \\
    & \ge - \langle v_\perp | B_1 \otimes I + I \otimes B_2 | v_\perp \rangle \\
    & = \left\langle v_{0\perp} \left| -B_1^{\prime} \otimes I + \frac{\varepsilon}{2} P_{\mathcal{K}^{(1)}} \otimes I + I \otimes \left( \frac{\|C\|^2}{\varepsilon} + \|B_1^{\prime}\| \right) P_{\mathcal{K}^{(2),\perp}} \right| v_{0\perp} \right\rangle \\
    & \quad + \left\langle v_{\perp0} \left| \left( \frac{\|C\|^2}{\varepsilon} + \|B_2^{\prime}\| \right) P_{\mathcal{K}^{(1),\perp}} \otimes I - I \otimes B_2^{\prime} + \frac{\varepsilon}{2} I \otimes P_{\mathcal{K}^{(2)}} \right| v_{\perp0} \right\rangle \\
    & \quad + \left\langle v_{\perp\perp} \left| \left( \frac{\|C\|^2}{\varepsilon} + \|B_2^{\prime}\| \right) P_{\mathcal{K}^{(1),\perp}} \otimes I + I \otimes \left( \frac{\|C\|^2}{\varepsilon} + \|B_1^{\prime}\| \right) P_{\mathcal{K}^{(2),\perp}} \right| v_{\perp\perp} \right\rangle \\
    & \ge \left( \frac{\varepsilon}{2} + \frac{\|C\|^2}{\varepsilon} \right) \|v_{0\perp}\|^2 + \left( \frac{\varepsilon}{2} + \frac{\|C\|^2}{\varepsilon} \right) \|v_{\perp0}\|^2 + \left( \frac{2}{\varepsilon}\|C\|^2 + \|B_1^{\prime}\| + \|B_2^{\prime}\| \right) \|v_{\perp\perp}\|^2 \\
    & \ge \frac{\|C\|^2}{\varepsilon} \left( \|v_{0\perp}\|^2 + \|v_{\perp0}\|^2 + \|v_{\perp\perp}\|^2 \right) \\
    & = \frac{\|C\|^2}{\varepsilon} \|v_\perp\|^2.
    \end{align*}
    Therefore, we get
    \begin{align*}
    &\quad\langle v | C - B_1 \otimes I - I \otimes B_2 | v \rangle \\
    &\ge \langle v_0 | C - B_1 \otimes I - I \otimes B_2 | v_0 \rangle + \langle v_{\perp} | C - B_1 \otimes I - I \otimes B_2 | v_{\perp} \rangle 
    - 2 \left| \langle v_{\perp} | C - B_1 \otimes I - I \otimes B_2 | v_0 \rangle \right| \\
    &\ge \varepsilon \| v_0 \|^2 + \frac{1}{\varepsilon} \|C\|^2 \| v_{\perp} \|^2 - 2 \|C\| \| v_{\perp} \| \| v_0 \| \\
    &= \left( \sqrt{\varepsilon} \| v_0 \| - \frac{\|C\|}{\sqrt{\varepsilon}} \| v_{\perp} \| \right)^2 \\
    &\ge 0.
    \end{align*}
    Then, we have
    \begin{align*}
    \tr(B_1\rho_1) 
    &= \sum_{j=1}^{\infty} \lambda_j^{(1)} \langle e_j^{(1)} | B_1 | e_j^{(1)} \rangle \\
    &= \sum_{j=1}^{\infty} \lambda_j^{(1)} \langle e_j^{(1)} | B_1^{\prime} - \frac{\varepsilon}{2} P_{\mathcal{K}^{(1)}} - \left( \frac{\|C\|^2}{\varepsilon} + \|B_2^{\prime}\| \right) P_{\mathcal{K}^{(1),\perp}} | e_j^{(1)} \rangle \\
    &= \left( \sum_{j=1}^{N} \lambda_j^{(1)} \right) \tr(B_1^{\prime}\rho_1^{\prime}) - \frac{\varepsilon}{2} \sum_{j=1}^{N} \lambda_j^{(1)} - 
    \left( \frac{\|C\|^2}{\varepsilon} + \|B_2^{\prime}\| \right) \sum_{j=N+1}^{\infty} \lambda_j^{(1)} \\
    &=\tr(B_1^{\prime}\rho_1^{\prime}) - \frac{\varepsilon}{2} \sum_{j=1}^{N} \lambda_j^{(1)} - \left( \sum_{j=N+1}^{\infty} \lambda_j^{(1)} \right) \left( \frac{\|C\|^2}{\varepsilon} + \|B_2^{\prime}\| + \tr(B_1^{\prime}\rho_1^{\prime}) \right) \\
    &\ge \tr(B_1^{\prime}\rho_1^{\prime}) - \frac{\varepsilon}{2} \quad - \quad \delta \cdot \left( \frac{\|C\|^2}{\varepsilon} + \|B_1^{\prime}\| + \|B_2^{\prime}\| \right).
    \end{align*}
    Similarly, 
    \[
    \tr(B_2\rho_2) \ge \tr(B_2^{\prime}\rho_2^{\prime}) - \frac{\varepsilon}{2} - \delta \left( \frac{\|C\|^2}{\varepsilon} + \|B_1^{\prime}\| + \|B_2^{\prime}\| \right).
    \]
    Therefore,
    \begin{align*}
    \dopt{C} &\ge \tr(B_1\rho_1) + \tr(B_2\rho_2)\\
    &\ge \tr(B_1^{\prime}\rho_1^{\prime}) + \tr(B_2^{\prime}\rho_2^{\prime}) - \varepsilon - 2\delta \left( \frac{\|C\|^2}{\varepsilon} + \|B_1^{\prime}\| + \|B_2^{\prime}\| \right) \\
    &\ge \dopt{C'} - 3\varepsilon - 2\delta \left( 2\|C\| + \frac{3}{\varepsilon}\|C\|^2 \right).
    \end{align*}
    And by taking $\varepsilon := \sqrt{\delta}$, we have the desired inequality.
\end{proof}

\subsection{Unbounded Duality Theorem}
\begin{theorem}[Kantorovich Duality for Infinite Dimensional Quantum Systems with Unbounded Cost]%
    \label{thm:quantum-Kantorovich-duality-unbounded-inf-appendix}
     Let $\mathcal{H}_1$ and $\mathcal{H}_2$ be two Hilbert spaces, 
    $\rho_1\in \mathcal{D}^1(\mathcal{H}_1)$, $\rho_2\in \mathcal{D}^1(\mathcal{H}_2)$ be two density operators, and $C$ be a bounded-by-below linear relation.
    Then, 
    \begin{equation*}
        \inf_{\rho\in\couplings{\rho_1}{\rho_2}} \Tr (C\rho)
        = \sup_{\rbra{n, A_1, A_2}\in \DualSet{C} } \Tr(A_1\rho_1) + \Tr(A_2\rho_2) -n,
    \end{equation*}
    where $\DualSet{C} = \set{\rbra{n,A_1, A_2}}{A_1\otimes I + I \otimes A_2 \sqsubseteq C+ nI}$.
\end{theorem}
\begin{proof}
    Denote $C^{(j)}\triangleq \trunc{C}{j}$ for $j\in \mathbb{N}$. By~\Cref{prop:truncation_monotonicity}, we know for any $j$, 
    $C^{(j)} \sqsubseteq C$. This means $\DualSet{C^{(j)}}\subseteq \DualSet{C}$, and $\bigcup_{j\in \mathbb{N}} \DualSet{C^{(j)}} \subseteq \DualSet{C}$.
    Then, noting that $C^{(j)}$ is bounded, 
    we directly have 
\begin{equation}
     \inf_{\rho\in\couplings{\rho_1}{\rho_2}} \Tr (C^{(j)}\rho)
        = \sup_{\rbra{n, A_1, A_2}\in \DualSet{C^{(j)}} } \Tr(A_1\rho_1) + \Tr(A_2\rho_2) -n,
\end{equation}    
    By the lower semi-continuity~(\Cref{thm:lsc_rho}) and compactness of couplings~(\Cref{lem:coupling-compact}), we know there exists a coupling $\rho^{(j)}\in\couplings{\rho_1}{\rho_2}$ such that 
    $\Tr(C^{(j)}\rho^{(j)})=\sup_{\rbra{n, A_1, A_2}\in \DualSet{C^{(j)}} } \Tr(A_1\rho_1) + \Tr(A_2\rho_2) -n$. 
    Let $\rho^{*}$ be a limit point of $\{\rho^{(j)}\}$. By compactness, 
    $\rho^{*}$ is also a coupling of $\rho_1$ and $\rho_2$.
    We then know
    \begin{equation*}
    \begin{aligned}
                &\sup_{\rbra{n, A_1, A_2}\in \DualSet{C} } \Tr(A_1\rho_1) + \Tr(A_2\rho_2) -n \\
                \ge\;& \sup_j \sup_{\rbra{n, A_1, A_2}\in \DualSet{C^{(j)}} } \Tr(A_1\rho_1) + \Tr(A_2\rho_2) -n \\
                =\;& \sup_j \Tr(C^{(j)}\rho^{(j)}) \\
                \ge\;& \liminf_j \Tr(C^{(j)}\rho^{(j)}) \\
                \ge\;&  \Tr(C\rho^{*}) \\
                \ge\;& \inf_{\rho\in\couplings{\rho_1}{\rho_2}} \Tr(C\rho) \\
    \end{aligned}
    \end{equation*}
    Combined with the direct observation 
    \begin{equation*}
        \inf_{\rho\in\couplings{\rho_1}{\rho_2}} \Tr(C\rho) \ge \sup_{\rbra{n, A_1, A_2}\in \DualSet{C} } \Tr(A_1\rho_1) + \Tr(A_2\rho_2) -n,
    \end{equation*}
    we get the desired claim.
\end{proof}

\section{Proofs in Probabilistic Programs}
\label{sec:appendix-proofs-in-probabilistic}

For completeness, we recall the weakest precondition for probabilistic programs proposed in~\cite{MMS96} as follows.
Here,  the truncated iterates $\whiledtit{b}{c}{n}$ of a loop $\whiled{b}{c}$ are defined inductively as follows:
$$\begin{array}{rcl}
\whiledtit{b}{c}{0} & = & \ifte{b}{\abort}{\skp}, \\
\whiledtit{b}{c}{n+1} & = & \ifte{b}{c;\whiledtit{b}{c}{n}}{\skp}. \\
\end{array}$$

\begin{table}[h]
\centering
\begin{tabular}{p{3.5cm}p{9cm}}
\toprule
\textbf{Command} & \textbf{Weakest Precondition} \\
\midrule
$\skp$ & $\wprecond{\skp}{\cqpredicate} \triangleq \cqpredicate$ \\[0.8em]

$\abort$ & $\wprecond{\abort}{\cqpredicate} \triangleq 0$ \\[0.8em]

$\assign{x}{e}$ & $\wprecond{\assign{x}{e}}{\cqpredicate} \triangleq \cqpredicate[e/x]$ \\[0.8em]

$\sample{x}{\mu}$ & $\wprecond{\sample{x}{\mu}}{\cqpredicate} \triangleq \E_{v\sim\mu}[\cqpredicate[v/x]]$ \\[0.8em]

$c_1;c_2$ & $\wprecond{c_1;c_2}{\cqpredicate} \triangleq \wprecond{c_1}{(\wprecond{c_2}{\cqpredicate})}$ \\[0.8em]

$\ifte{b}{c_1}{c_2}$ & 
$\wprecond{\ifte{b}{c_1}{c_2}}{\cqpredicate} \triangleq ( \wprecond{c_1}{\cqpredicate})|_b + 
(\wprecond{c_2}{\cqpredicate})|_{\neg b}$ \\[0.8em]

$\whiled{b}{c}$ & 
$\wprecond{\whiled{b}{c}}{\cqpredicate} \triangleq
\lim_{n\rightarrow\infty} \wprecond{\whiledtit{b}{c}{n} }{\cqpredicate}$ \\

\bottomrule
\end{tabular}
\caption{Structural representations of non-relational weakest preconditions in probabilistic programs.
}
\label{tab:weakest-precondition-prob}
\end{table}

Recall that a function $f: X\to [-\infty, +\infty]$ is lower-semicontinuous if and only if for all
$x_{\alpha}\to x$, $\liminf_{\alpha} f(x_{\alpha}) \ge f(x)$ (see Lemma 2.42 in~\cite{AB06}).
We first recall the following theorem, which states that we can always find a minimizer in a compact set for a lower-semicontinuous function.
\begin{theorem}[Theorem 2.43 in~\cite{AB06}]\label{thm:compact-lsc-attain-minimum}
    A real-valued lower-semicontinuous function on a compact
space attains a minimum value, and the nonempty set of minimizers is compact.
\end{theorem}

\begin{theorem}[Kantorovich-Rubinstein Duality Theorem,~{\cite[Theorem 5.10]{Villani08}}]
\label{thm:kr-duality-appendix}
Let $Q\in\bbmaps{X_1\times X_2}$, and let $\dualset{Q}\subseteq
  \bndmaps{X_1}\times\bndmaps{X_2}$ such that
$(Q_1,Q_2)\in \dualset{Q}$ iff
      $Q_1 \boxplus Q_2\sqsubseteq Q$. Then
$$
 \min_{\mu
  \in \couplings{\mu_1}{\mu_2}} \E_{\mu}[Q] =\sup_{(Q_1,Q_2)\in\dualset{Q}}
  \{\E_{\mu_1}[Q_1]+\E_{\mu_2}[Q_2]\}.$$
\end{theorem}

\begin{proposition}\label{prop:split-expectation-prob}
    Let $\mu_1\in \distr(X_1)$, $\mu_2\in \distr(X_2)$, and $\mu\in \couplings{\mu_1}{\mu_2}$. Then, for $Q_1\in \bbmaps{X_1}$ and $Q_2\in \bbmaps{X_2}$, we have
    \begin{equation*}
        \E_{\mu}[Q_1\boxplus Q_2] = \E_{\mu_1}[Q_1] + \E_{\mu_2}[Q_2].
    \end{equation*}
\end{proposition}
\begin{proof}
     By definition, we know $(Q_1\boxplus Q_2)(x_1,x_2) = Q_1(x_1)+Q_2(x_2)$. Therefore, we have
     \begin{equation*}
         \begin{aligned}
             \E_{\mu}[Q_1\boxplus Q_2] 
             &=\sum_{x_1\in X_1,x_2\in X_2} \mu(x_1,x_2) ( Q_1(x_1) + Q_2(x_2) ) \\
            &=\sum_{x_1\in X_1,x_2\in X_2} \mu(x_1,x_2)  Q_1(x_1) + \sum_{x_1\in X_1,x_2\in X_2} \mu(x_1,x_2)  Q_2(x_2)  \\
             &=\sum_{x_1\in X_1}   Q_1(x_1) \sum_{x_2\in X_2} \mu(x_1,x_2) + \sum_{x_1\in X_1,x_2\in X_2}  Q_2(x_2) \sum_{x_1\in X_1} \mu(x_1,x_2) \\
             &=\sum_{x_1\in X_1}   Q_1(x_1) \mu_1(X_1) + \sum_{x_1\in X_1,x_2\in X_2}  Q_2(x_2) \mu_2(X_2) \\
             &= \E_{\mu_1}[Q_1]+\E_{\mu_2}[Q_2],
         \end{aligned} 
     \end{equation*}
     which gives the desired result.
\end{proof}

Note that the set of couplings is compact, and the expectation is lower-semicontinuous. Therefore, there is always
\begin{theorem}[Soundness and Completeness of core rules]\label{thm:prob-complete-appendix}
A judgment is valid iff it can be derived with the rules \RuleRef{duality},
\RuleRef{conseq}, and \RuleRef{wp}.
\end{theorem}

\begin{proof}
    For soundness, we prove as follows:
    
    \RuleRef{conseq}: For any $s_1,s_2\in\cstates$, by definition of validity and the assumption, there
exists $\mu\in \couplings{\sem{S_1}~s_1}{\sem{S_2}~s_2}$ such that $P(s_1,s_2)\ge \E_{\mu}[Q]$. Since $P'\sqsupseteq P$, we know $P(s_1,s_2)\le P'(s_1,s_2)$. Also, from $Q\sqsupseteq Q'$ we know $ \E_{\mu}[Q]\ge  \E_{\mu}[Q']$. This gives $P'(s_1,s_2)\ge \E_{\mu}[Q']$.

   \RuleRef{duality}: We fix $s_1,s_2\in\cstates$. By definition of validity and the assumption, for any $(Q_1,Q_2)\in \dualset{Q}$, there
exists $\mu\in \couplings{\sem{S_1}~s_1}{\sem{S_2}~s_2}$ such that $$P(s_1,s_2)\ge \E_{\mu}[Q_1\boxplus Q_2] =  \E_{\sem{S_1}~s_1}[Q_1]+\E_{\sem{S_2}~s_2}[Q_2]$$ by~\Cref{prop:split-expectation-prob}. This means, $$P(s_1,s_2)\ge \sup_{(Q_1,Q_2)\in \dualset{Q}} \E_{\sem{S_1}~s_1}[Q_1]+\E_{\sem{S_2}~s_2}[Q_2].$$ 
However, by~\Cref{thm:kr-duality-appendix}, we know 
$$\sup_{(Q_1,Q_2)\in \dualset{Q}} \E_{\sem{S_1}~s_1}[Q_1]+\E_{\sem{S_2}~s_2}[Q_2] = \min_{\mu\in \couplings{\sem{S_1}~s_1}{\sem{S_2}~s_2}} \E_{\mu}[Q].$$ Therefore, taking $\mu$ being the minimizer of $\min_{\mu\in \couplings{\sem{S_1}~s_1}{\sem{S_2}~s_2}} \E_{\mu}[Q]$, we know $P(s_1,s_2)\ge \E_{\mu}[Q]$.

\RuleRef{wp}: We fix $s_1,s_2\in\cstates$, and let $\mu$ be the trivial coupling of $\sem{S_1}~s_1$ and $\sem{S_2}~s_2$ (i.e., the product measure of the two distribution). By~\Cref{prop:split-expectation-prob}, we know 
$$\E_{\mu}[Q_1\boxplus Q_2] = \E_{\sem{S_1}~s_1}[Q_1] + \E_{\sem{S_2}~s_2}[Q_2].$$ Since $\E_{\sem{S_1}~s_1}[Q_1] = \wprecond{S_1}{Q_1} (s_1)$ and $\E_{\sem{S_2}~s_2}[Q_2] = \wprecond{S_2}{Q_2} (s_2)$, we know $$\wprecond{S_1}{Q_1}\boxplus \wprecond{S_2}{Q_2} (s_1,s_2) = \wprecond{S_1}{Q_1} (s_1) + \wprecond{S_2}{Q_2} (s_2) = \E_{\mu}[Q_1\boxplus Q_2].$$

For completeness, we prove as follows:

Assume that $\models
\rtriple{P}{S_1}{S_2}{Q}$. By definition of validity, 
we have for every $s_1,s_2\in\cstates$, there
exists $\mu\in \couplings{\sem{S_1}~s_1}{\sem{S_2}~s_2}$ such that
$\E_{\mu}[Q] \leq P(s_1,s_2)$. 
This means, for every $s_1,s_2\in\cstates$, $$ P(s_1,s_2) \geq \inf_{\mu\in \couplings{\sem{S_1}~s_1}{\sem{S_2}~s_2}}\E_{\mu}[Q].$$
By~\Cref{thm:kr-duality-appendix}, this means
for every $s_1,s_2\in\cstates$, $$ P(s_1,s_2) \geq \sup_{(Q_1,Q_2)\in \dualset{Q}}\E_{\sem{S_1}~s_1}[Q_1] + \E_{\sem{S_2}~s_2}[Q_2] .$$
By the property of weakest
precondition that $\wprecond{S}{Q}=\lambda s. \E_{\sem{S}~s}[Q]$, and by~\Cref{prop:split-expectation-prob}, 
we have for every $(Q_1,Q_2)\in\dualset{Q}$, and
for every $s_1,s_2\in\cstates$,
$$ P(s_1,s_2) \geq \wprecond{S_1}{Q_1} (s_1) + \wprecond{S_2}{Q_2}(s_2) .$$
By definition,
this means
$\wprecond{S_1}{Q_1}\boxplus
\wprecond{S_2}{Q_2}\sqsubseteq P$ for every $(Q_1,Q_2)\in\dualset{Q}$.  
By the \RuleRef{wp} and \RuleRef{conseq} rules, it
follows that $\vdash \rtriple{P}{S_1}{S_2}{Q_1\boxplus Q_2}$. One concludes by finally
applying the \RuleRef{duality} rule.
\end{proof}

\section{Proofs in Finite-Dimensional and Infinite-Dimensional Quantum Programs}

We first recall the semantics of the qWhile language in~\cite{Ying11,Ying16}, which is formulated for infinite-dimensional quantum programs. The corresponding semantics for finite-dimensional programs follows by minor and direct modifications, and can therefore be regarded as a simple specialisation of the infinite-dimensional case.
\begin{definition}[Denotational Semantics of qWhile,~\cite{Ying11,Ying16}]\label{lem-structural} 
For any input state $\rho \in \qstate{\mathcal{H}}$, we have: 
\begin{enumerate}
  \item \label{dsem_skp} $\sem{\skp}(\rho) = \rho$;
  \item \label{dsem_init} 
  $\sem{q := \ket{0}}(\rho) = 
  \begin{cases}
    |0\rangle_q\langle 0|\rho|0\rangle_q\langle 0| + |0\rangle_q\langle 1|\rho|1\rangle_q\langle 0|, \text{ if} \text{ type}(q) = \mathsf{Boolean}; \\
       \sum_{n=-\infty}^{\infty} |0\rangle_q\langle n|\rho|n\rangle_q\langle 0|, \text{ if} \text{ type}(q) = \mathsf{Integer}:
  \end{cases}      $
  \item \label{dsem_uni} $\sem{\oq := U[\oq]}(\rho) = U_{\oq} \rho U_{\oq}^\dagger$;
  \item \label{dsem_comp} $\sem{S_1; S_2}(\rho) = \sem{S_2}(\sem{S_1}(\rho))$; 
  \item \label{dsem_if} $\sem{\ifb(\guard v\cdot M[\oq] = v \to S_v)\ife}(\rho) = \sum_v \sem{S_v}(M_v \rho M_v^\dagger)$; 
    \item \label{dsem_while} for the while loop $\while[M, S]\triangleq \while\ M[\oq]=1\ \wdo\ S\ \wod$: 
    $$\sem{\while[M, S]}(\rho) = \bigsqcup_{k = 0}^{\infty} \sem{\while^{(k)}[M,S]}(\rho),$$ 
where $\while^{(k)}[M,S]$ is the $k$-fold iteration of the loop: 
\begin{equation*}\label{iteration}\begin{cases}
  &\while^{(0)}[M,S]  \triangleq \textbf{abort}, \\
  &\while^{(k+1)}[M,S] \triangleq \begin{aligned} 
  &\ifb\ M[\oq]  =\  0 \to \skp\\  
  &\guard\qquad \qquad 1\to S; \while^{(k)}[M,S]\ \ife
\end{aligned}\end{cases}\end{equation*}
for $k\geq 0$, $\bigsqcup$ stands for the least upper bound in the CPO of partial density operators with the L\"owner order $\sqsubseteq$, 
and $\abort$ is a program that never terminates, i.e., $\sem{\abort}(\rho) = \bf{0}$ for all $\rho$.
\end{enumerate}\end{definition}

For completeness, we present the explicit form of the weakest precondition of qWhile in~\Cref{tab:weakest-precondition-quantum}, which is proposed in~{\cite[Proposition 7.1]{Ying11}}. Note that all predicates here are \emph{bounded}.

\begin{table}[tb]
\centering
\begin{tabular}{p{3.5cm}p{9cm}}
\toprule
\textbf{Command} & \textbf{Weakest Precondition} \\
\midrule
$\skp$ & $\wprecond{\skp}{\cqpredicate} \triangleq \cqpredicate$ \\[0.8em]

$q \coloneqq \ket{v}$ & $\wprecond{q \coloneqq \ket{v}}{\cqpredicate} \triangleq \sum_{i} \ket{i}_{q}\bra{v}\cqpredicate\ket{v}_{q}\bra{i}$ \\[0.8em]

$\oq \coloneqq U [\oq]$ & $\wprecond{\oq \coloneqq U [\oq]}{\cqpredicate} \triangleq U^{\dagger}_{\oq}\cqpredicate U_{\oq}$ \\[0.8em]

$c_1;c_2$ & $\wprecond{c_1;c_2}{\cqpredicate} \triangleq \wprecond{c_1}{(\wprecond{c_2}{\cqpredicate})}$ \\[0.8em]

$\ifb\ (\guard v\cdot M[\oq] = v \to S_v)\ \ife$ & 
$\wprecond{\ifb\ (\guard v\cdot M[\oq] = v \to S_v)\ \ife}{\cqpredicate} \triangleq \sum_v M_v^{\dagger}\wprecond{S_v}{Q} M_v$ \\[0.8em]

$\while\ M[\oq]=1\ \wdo\ S\ \wod $ & 
\begin{tabular}[t]{@{}l}
   $\wprecond{\whiled{b}{c}}{\cqpredicate} \triangleq \bigvee_n Q_n, 
$    \\
 where  $\begin{cases}
    Q_0 = 0; \\
    Q_{n+1} = M_0^{\dagger} QM_0 + M_1^{\dagger} \wprecond{S}{Q_n} M_1.
\end{cases}$
\end{tabular}
\\
\bottomrule
\end{tabular}
\caption{Structural representations of non-relational weakest preconditions in quantum programs.
}
\label{tab:weakest-precondition-quantum}
\end{table}

Applying~\Cref{thm:compact-lsc-attain-minimum}, we know we can always find a minimizer for the optimal transport cost, since $\couplings{\rho_1}{\rho_2}$ is a compact set (as it is closed and bounded in trace norm, see~\cite[Proposition C.1]{qotl}). We write this formally as follows.

\begin{theorem}[Kantorovich-Rubinstein Duality Theory for Finite Dimensional Quantum Systems, see~{\cite[Theorem III.3]{qotl}}]%
\label{thm:quantum-Kantorovich-duality-finite-dim-appendix}
     Let $\cH_1$ and $\cH_2$ be two finite-dimensional Hilbert spaces, 
    $\rho_1\in \qstate{\cH_1}$, $\rho_2\in \qstate{\cH_2}$ be two density operators, and $Q\in \qbndmaps{\cH_1\otimes \cH_2}$. 
    Let $\dualset{Q}\subseteq \qbndmaps{\cH_1}\times \qbndmaps{\cH_2}$ such that $(Q_1, Q_2)\in \dualset{Q}$ iff $ Q_1\boxplus Q_2 \sqsubseteq Q$.
    Then, 
    \begin{equation*}
      \min_{\rho\in\couplings{\rho_1}{\rho_2}} \tr (Q\rho) 
    = \sup_{(Q_1, Q_2)\in \dualset{Q}} \tr(Q_1\rho_1) + \tr(Q_2\rho_2).
    \end{equation*}
\end{theorem}

\begin{theorem}[Kantorovich Duality for Infinite-Dimensional Quantum Systems with Bounded-By-Below Cost]%
    \label{thm:quantum-Kantorovich-duality-unbounded-appendix}
     Let $\mathcal{H}_1$ and $\mathcal{H}_2$ be two Hilbert spaces, 
    $\rho_1\in \qstate{\mathcal{H}_1}$, $\rho_2\in \qstate{\mathcal{H}_2}$ be two density operators, and $C\in \qbbmaps{\cH_1\otimes \cH_2}$ be a bounded-by-below self-adjoint linear relation.
Let $\dualset{Q}\subseteq \qbbmaps{\cH_1}\times \qbbmaps{\cH_2}$ be the set such that $(A_1, A_2)\in \dualset{Q}$ iff $Q_1\boxplus Q_2\sqsubseteq Q$,
    Then, 
    \begin{equation*}
         \min_{\rho\in\couplings{\rho_1}{\rho_2}} \Tr (Q\rho)
         = \sup_{(Q_1,Q_2)\in \dualset{Q}} \Tr(Q_1\rho_1) + \Tr(Q_2\rho_2),
    \end{equation*}
\end{theorem}


We now prove some propositions that will be used in the soundness and completeness proof. Note that these propositions hold for infinite-dimensional case.
\begin{proposition}\label{prop:split-expectation-quantum}
    Let $\rho_1\in \qstate{\cH_1}$, $\rho_2\in \qstate{\cH_2}$, and $\rho\in \couplings{\rho_1}{\rho_2}$. Then, for $Q_1\in \bndmaps{\cH_1}$ and $Q_2\in \bndmaps{\cH_2}$, we have
    \begin{equation*}
        \tr((Q_1\boxplus Q_2)\rho) = \tr(Q_1\rho_1) + \tr(Q_2\rho_2).
    \end{equation*}
\end{proposition}
\begin{proof}
     By definition, we know $(Q_1\boxplus Q_2) = Q_1\otimes I + I\otimes Q_2$. Therefore, we have
     \begin{equation*}
         \begin{aligned}
            \tr((Q_1\boxplus Q_2)\rho) &= \tr((Q_1\otimes I + I\otimes Q_2)\rho) \\
             &= \tr((Q_1\otimes I) \rho) + \tr((I\otimes Q_2) \rho) \\
            &=\tr(Q_1\rho_1) + \tr(Q_2\rho_2),
         \end{aligned} 
     \end{equation*}
       which gives the desired result.
\end{proof}

\begin{proposition}\label{prop:partial-trace-tensor}
    For $\rho\in \qstate{\cH_1\otimes\cH_2}$, $\mathcal{E}_1\in \mathcal{QO}(\mathcal{H}_1)$, and  $\mathcal{E}_2\in \mathcal{QO}(\mathcal{H}_2)$ which is trace-preserving,
    we have 
    \begin{equation*}
    \begin{aligned}
        \tr_2(( \mathcal{E}_1\otimes I)(\rho)) &= \mathcal{E}_1(\tr_2(\rho)).\\
        \tr_2((I\otimes \mathcal{E}_2)(\rho)) &= \tr_2(\rho).
    \end{aligned}
    \end{equation*}
\end{proposition}

\begin{proof}
    For the first equation, note that $\tr_2(\cdot) = I\otimes \tr(\cdot)$. Therefore, it commutes with $\mathcal{E}_1\otimes I$.
    
    For the second equation, note that the map $\tr_2((I\otimes \mathcal{E}_2)(\cdot))$ is linear, and for all $\rho = \rho_1\otimes \rho_2$, we have $$\tr_2((I\otimes \mathcal{E}_2)(\rho_1\otimes \rho_2)) = \tr_2( \rho_1\otimes \mathcal{E}_2(\rho_2))  = \tr_2(\rho_1) = \tr_2(\rho_1).$$
    The claim then holds by the uniqueness of partial trace.
\end{proof}

\begin{proposition}\label{prop:ast-preserve-coupling}
     Let $\rho_1\in \qstate{\cH_1}$, $\rho_2\in \qstate{\cH_2}$, and $\rho\in \couplings{\rho_1}{\rho_2}$. If $\mathcal{E}_1\in \mathcal{QO}(\mathcal{H}_1)$ and $\mathcal{E}_2\in \mathcal{QO}(\mathcal{H}_2)$ are trace-preserving, then
     $(\mathcal{E}_1\otimes \mathcal{E}_2)(\rho)\in \couplings{\mathcal{E}_1(\rho_1)}{\mathcal{E}_2(\rho_2)}$.
\end{proposition}

\begin{proof}
   Note that $\mathcal{E}_1\otimes \mathcal{E}_2$ is completely positive and trace-preserving. $(\mathcal{E}_1\otimes \mathcal{E}_2)(\rho)$ is a density operator.
   We now show $\tr_2((\mathcal{E}_1\otimes \mathcal{E}_2)(\rho)) = \mathcal{E}_1(\rho_1)$. By~\Cref{prop:partial-trace-tensor},  we know $$\tr_2((\mathcal{E}_1\otimes \mathcal{E}_2)(\rho)) 
   = \tr_2((\mathcal{E}_1\otimes I)((I\otimes  \mathcal{E}_2)(\rho))) = \mathcal{E}_1(\tr_2((I\otimes  \mathcal{E}_2)(\rho))) = \mathcal{E}_1(\tr_2(\rho))
   \mathcal{E}_1(\rho_1).$$ Similarly, $\tr_1((\mathcal{E}_1\otimes \mathcal{E}_2)(\rho)) = \mathcal{E}_2(\rho_2)$.
\end{proof}

\begin{theorem}[Soundness and Completeness of Core Rules in Quantum Case]\label{thm:finite-quantum-complete-appendix}
A judgment is valid in quantum programs iff it can be derived with the rules \RuleRef{duality},
\RuleRef{conseq}, and \RuleRef{wp}.
\end{theorem}

\begin{proof}
    We prove for the infinite-dimensional case. The proof of finite-dimensional case is the same except that we use~\Cref{thm:quantum-Kantorovich-duality-finite-dim-appendix} instead of~\Cref{thm:quantum-Kantorovich-duality-unbounded-appendix}.

    For soundness, we prove as follows:
    
    \RuleRef{conseq}: By definition of validity and the assumption,  every $\rho \in \qstate{\cH\otimes \cH}$,  
there exists a coupling $\sigma$ in 
  $\couplings{\sem{S_1}(\tr_2(\rho))}{\sem{S_2}(\tr_1(\rho))}$ such that
  $\Tr(P \rho) \geq \Tr(Q \sigma)$. Since $P'\sqsupseteq P$, we know $\Tr(P \rho)\le \Tr(P' \rho)$. Also, from $Q\sqsupseteq Q'$ we know $\Tr(Q \sigma)\ge \Tr(Q' \sigma)$. This gives $\Tr(P' \rho) \geq \Tr(Q' \sigma)$.

   \RuleRef{duality}: We fix $\rho \in \qstate{\cH\otimes \cH}$. By definition of validity and the assumption, for any $(Q_1,Q_2)\in \dualset{Q}$, there exists a coupling $\sigma$ in 
  $\couplings{\sem{S_1}(\tr_2(\rho))}{\sem{S_2}(\tr_1(\rho))}$  such that $$\Tr(P \rho) \geq \Tr((Q_1\boxplus Q_2) \sigma) =  \tr(Q_1\sem{S_1}(\tr_2(\rho))) + \tr(Q_2\sem{S_2}(\tr_1(\rho)))$$ by~\Cref{prop:split-expectation-quantum}. This means, $$\Tr(P \rho) \ge \sup_{(Q_1,Q_2)\in \dualset{Q}}  \tr(Q_1\sem{S_1}(\tr_2(\rho))) + \tr(Q_2\sem{S_2}(\tr_1(\rho))).$$ 
However, by~\Cref{thm:quantum-Kantorovich-duality-unbounded-appendix}, we know 
$$\sup_{(Q_1,Q_2)\in \dualset{Q}}  \tr(Q_1\sem{S_1}(\tr_2(\rho))) + \tr(Q_2\sem{S_2}(\tr_1(\rho))) = \min_{\sigma\in \couplings{\sem{S_1}(\tr_2(\rho))}{\sem{S_2}(\tr_1(\rho))}} \Tr(Q \sigma).$$ Therefore, taking $\sigma$ being the minimizer of $\min_{\sigma\in \couplings{\sem{S_1}(\tr_2(\rho))}{\sem{S_2}(\tr_1(\rho))}} \Tr(Q \sigma)$, we know $\Tr(P \rho) \ge \Tr(Q \sigma)$.

\RuleRef{wp}: We fix $\rho \in \qstate{\cH\otimes \cH}$, and let $\sigma =  (\sem{S_1}\otimes \sem{S_2})(\rho)$. By~\Cref{prop:ast-preserve-coupling}, we know $\sigma \in \couplings{\sem{S_1}(\tr_2(\rho))}{\sem{S_2}(\tr_1(\rho))}$ .  By~\Cref{prop:split-expectation-quantum}, we know 
$$\tr((Q_1\boxplus Q_2 )\sigma) = \tr(Q_1 \sem{S_1}(\tr_2(\rho)) )+ \tr(Q_1 \sem{S_2}(\tr_1(\rho)) ).$$ Since $ \tr(Q_1 \sem{S_1}(\tr_2(\rho)) ) = \tr(\wprecond{S_1}{Q_1}\tr_1(\rho))$ and $\tr(Q_2 \sem{S_2}(\tr_1(\rho)) ) = \tr(\wprecond{S_2}{Q_2}\tr_2(\rho))$, we know $$\tr((\wprecond{S_1}{Q_1}\boxplus \wprecond{S_2}{Q_2} )\rho ) = \tr(\wprecond{S_1}{Q_1}\tr_1(\rho)) + \tr(\wprecond{S_2}{Q_2}\tr_2(\rho))  = \tr((Q_1\boxplus Q_2 )\sigma).$$

For completeness, we prove as follows:

Assume that $\models
\rtriple{P}{S_1}{S_2}{Q}$. By definition of validity, 
we have: for every $\rho \in \mathcal{D}^1(\mathcal{H} \otimes \mathcal{H})$, 
  there exists a coupling $\sigma\in \couplings{\sem{S_1}(\tr_2(\rho))}{\sem{S_2}(\tr_1(\rho))}$ such that
  $\tr(P \rho) \geq \tr(Q \sigma)$. 
This means, for every $\rho \in \mathcal{D}^1(\mathcal{H} \otimes \mathcal{H})$, $$ \tr(P \rho) \geq \inf_{\sigma\in \couplings{\sem{S_1}(\tr_2(\rho))}{\sem{S_2}(\tr_1(\rho))}} \tr(Q \sigma) .$$
By~\Cref{thm:quantum-Kantorovich-duality-unbounded-appendix}, this means
for every $\rho \in \mathcal{D}^1(\mathcal{H} \otimes \mathcal{H})$, $$ \tr(P \rho)  \geq \sup_{(Q_1,Q_2)\in \dualset{Q}} \tr(Q_1 \sem{S_1}(\tr_2(\rho)))  + \tr(Q_2 \sem{S_2}(\tr_1(\rho))) .$$
By the property of weakest
precondition that $\tr(\wprecond{S}{Q}\rho_1)= \tr(Q\sem{S}(\rho_1))$, and by~\Cref{prop:split-expectation-prob}, 
we have
for every $\rho \in \mathcal{D}^1(\mathcal{H} \otimes \mathcal{H})$, and for every $(Q_1,Q_2)\in\dualset{Q}$,
$$ \tr(P \rho)  \geq  \tr((\wprecond{S_1}{Q_1}\boxplus
\wprecond{S_2}{Q_2}) \rho) .$$
By definition,
this means
$\wprecond{S_1}{Q_1}\boxplus
\wprecond{S_2}{Q_2}\sqsubseteq P$ for every $(Q_1,Q_2)\in\dualset{Q}$.  
By the \RuleRef{wp} and \RuleRef{conseq} rules, it
follows that $\vdash \rtriple{P}{S_1}{S_2}{Q_1\boxplus Q_2}$. One concludes by finally
applying the \RuleRef{duality} rule.
\end{proof}

\section{Proofs in Classical-Quantum Programs}
\label{sec:appendix-proofs-cq-programs}
\subsection{Basic Definitions and Properties}


We first recall some basic notions and definitions in classical-quantum setting.
For a countable set $X$ and a separable Hilbert space $\cH$,
the set of classical-quantum (cq) state $\cqstates{X}{\cH}$ consists of the functions $\cqstate:X\to \qpstate{\cH}$ with trace no more than $1$,
where its trace is defined as
\[
\tr(\cqstate) \triangleq \sum_x\tr(\cqstate(x)).
\]
Correspondingly, the trace norm distance $\|\cdot\|_{1}$ between two cq-states $\cqstate_1, \cqstate_2\in \cqstates{X}{\cH}$
\[
\|\cqstate_1-\cqstate_2\|_{1}\triangleq\sum_{x\in X}\|\cqstate_1(x)-\cqstate_2(x)\|_{1}=\sum_{x\in X}\tr(\sqrt{(\cqstate_1(x)^{\dagger}-\cqstate_2(x)^{\dagger})(\cqstate_1(x)-\cqstate_2(x))}).
\]
The partial trace $\tr_2(\cdot)$ can be defined as, for $\cqstate\in \cqpstates{X_1\times X_2}{\cH_1\otimes\cH_2}$,
\begin{equation*}
    \tr_2(\cqstate) (x_1) = \sum_{x_2\in X_2} \tr_2(\cqstate(x_1,x_2)).
\end{equation*}
$\tr_1(\cdot)$ can be defined similarly. We say $\cqstate\in \couplings{\cqstate_1}{\cqstate_2}$ if $\tr_2(\cqstate) = \cqstate_1$ and $\tr_1(\cqstate) = \cqstate_2$.
For a countable set $X$ and a separable Hilbert space $\cH$,
the set of bounded maps $\cqbndmaps{X}{\cH}$  consists of the functions $\cqpredicate:X\to \qpstate{\cH}$ with bounded operator norm,
which is defined as
\begin{equation*}
    \norm{\cqpredicate} = \sup_{x\in X} \norm{\cqpredicate(x)}.
\end{equation*}
For $\cqpredicate_1\in\cqbndmaps{X_1}{\cH_1}$ and $\cqpredicate_2\in\cqbndmaps{X_2}{\cH_2}$, we define
$\cqpredicate_1\boxplus \cqpredicate_2\in \cqbndmaps{X_1\times X_2}{\cH_1\otimes \cH_2}$ as 
\begin{equation*}
    (\cqpredicate_1\boxplus \cqpredicate_2) (x_1, x_2) \triangleq \cqpredicate_1(x_1)\otimes I + I\otimes \cqpredicate_2(x_2).
\end{equation*}
For $\cqpredicate\in\bbmaps{X,\cH}$ and $\cqstate\in\states{X,\cH}$, the expectation $\E_\cqstate[\cqpredicate]$ is defined by $\sum_{x\in X}\Tr(\cqpredicate(x)\cqstate(x))$. 

\begin{definition}[Extended L\"owner Order for Classical-Quantum Predicates]\label{def:extended-lowner-order-cqpredicate}
    Let $\cqpredicate_1,\cqpredicate_2\in\cqbbmaps{X}{\cH}$. We say $\cqpredicate_1\sqsubseteq\cqpredicate_2$ if $\cqpredicate_1(x)\sqsubseteq\cqpredicate_2(x)$ holds for every $x\in X$.
\end{definition}

\begin{definition}[Embedding Map]\label{def:embedding}
    Let the embedding map $\iota:\cqbbmaps{X}{\cH}\to \qbbmaps{\mathbb{C}^{X}\otimes \cH}$ be defined by $\iota (\cqpredicate)\triangleq \sum_{x\in X} \ket{x}\bra{x} \otimes \cqpredicate(x)$.
\end{definition}

\begin{definition}[Retraction Map]\label{def:restriction}
    For each $x \in X$, let $V_x: \cH \to \bC^X \otimes \cH$ be the isometry embedding defined by $V_x |u\rangle = |x\rangle \otimes |u\rangle$ for all $|u\rangle \in \mathcal{H}$.
    Let $\mathcal{R}: \qbndmaps{\bC^X \otimes \cH} \to \qbndmaps{X,\cH}$ be the retraction map defined by:
    $$ \mathcal{R}(\cqpredicate)(x) \triangleq V_x^\dagger \cqpredicate V_x $$
    for any bounded operator $\cqpredicate$, where $V^{\dagger} $ is a projection: $\mathbb{C}^X \otimes \mathcal{H} \to \mathcal{H}$:
    \[
    V_x^{\dagger} (|y\rangle \otimes |v\rangle) =
    \begin{cases}
    |v\rangle, & \text{if\,} y = x \\
    0, & \text{\if\, } y \neq x
    \end{cases}
    \]
\end{definition}


\begin{proposition}\label{prop:bounded-cq-expectation}
    For $\cqpredicate\in \qbndmaps{\mathbb{C}^{X}\otimes \cH}$ and $\cqstate\in \cqstates{X}{\cH}$,
    $\tr(\cqpredicate\iota(\cqstate)) = \E_{\cqstate}[\mathcal{R}(\cqpredicate)]$. 
\end{proposition}

\begin{proof}
    Recall that $\iota(\cqstate) = \sum_{x \in X} |x\rangle\langle x| \otimes \cqstate(x)$. Note that $|x\rangle\langle x| \otimes \cqstate(x)$ can be rewritten using the isometry $V_x$:
    $$ |x\rangle\langle x| \otimes \cqstate(x) = V_x \cqstate(x) V_x^\dagger $$

    Then,
    $$
    \begin{aligned}
    \tr(\cqpredicate \iota(\cqstate)) &= \tr\left( \cqpredicate \sum_{x\in X} V_x \cqstate(x) V_x^\dagger \right) \\
    &= \sum_{x} \tr( \cqpredicate V_x \cqstate(x) V_x^\dagger ) \quad (\text{Linearity}) \\
    &= \sum_{x\in X} \tr( V_x^\dagger \cqpredicate V_x \cqstate(x) ) \quad (\text{Cyclic property of the usual trace, see \cite[Theorem VI.25]{reed1980methods}}) \\
    &= \sum_{x\in X} \tr( \mathcal{R}(\cqpredicate)(x) \cqstate(x) ) \quad (\text{Definition of } \mathcal{R}) \\
    &= \mathbb{E}_{\cqstate}[\mathcal{R}(\cqpredicate)].
    \end{aligned}
    $$
\end{proof}

\begin{proposition}\label{prop:inf-embeding-cq}
    For $\cqstate_1, \cqstate_2\in \cqstates{X}{\cH}$ and $\cqpredicate\in \cqbbmaps{X\times X}{\cH\otimes \cH}$, we have
    $\inf_{\cqstate\in \couplings{\cqstate_1}{\cqstate_2}} \E_{\cqstate}[\cqpredicate] = \inf_{\rho\in \couplings{\iota(\cqstate_1)}{\iota(\cqstate_2)}}\Tr(\iota(\cqpredicate)\rho)$.
\end{proposition}
\begin{proof}
    Firstly we prove $\mathrm{LHS}\geq \mathrm{RHS}$. For every $\cqstate\in\couplings{\cqstate_1}{\cqstate_2}$, we construct $\iota(\cqstate)$ in $\couplings{\iota(\cqstate_1)}{\iota(\cqstate_2)}$:
    \[
    \iota(\cqstate)=\sum_{x_1\in X_1, x_2\in X_2}\ket{x_1}\bra{x_1}\otimes\ket{x_2}\bra{x_2}\otimes\cqstate(x_1,x_2)
    \]
    It's easy to verify that $\tr_2(\iota(\cqstate))=\iota(\cqstate_1)$ and $\tr_1(\iota(\cqstate))=\iota(\cqstate_2)$, thus $\iota(\cqstate)$ is a coupling in $\couplings{\iota(\cqstate_1)}{\iota(\cqstate_2)}$.
    Since
    \[
    \Tr(\iota(\cqpredicate)\iota(\cqstate))=\E_{\cqstate}[\cqpredicate],
    \]
    and $\couplings{\cqstate_1}{\cqstate_2}\subseteq\couplings{\iota(\cqstate_1)}{\iota(\cqstate_2)}$, 
    we know  $\inf_{\cqstate\in \couplings{\cqstate_1}{\cqstate_2}}\E_{\cqstate}[\cqpredicate]\geq\inf_{\rho\in \couplings{\iota(\cqstate_1)}{\iota(\cqstate_2)}}\tr(\iota(\cqpredicate)\rho)$. 

   Next we show that $\mathrm{LHS}\leq \mathrm{RHS}$. Thi s is sufficient to show that for any quantum coupling $\rho \in\couplings{\iota(\cqstate_1)}{\iota(\cqstate_2)}$, there exists a classical-quantum coupling $\cqstate$ yielding an equal or lower cost.

    Let $\mathbf{x} = (x_1, x_2)$ and $\mathsf{P}_{\mathbf{x}} = |x_1\rangle\langle x_1| \otimes |x_2\rangle\langle x_2| \otimes I_{\mathcal{H} \otimes \mathcal{H}}$. We define the dephasing channel (measurement in the classical basis) $\mathcal{E}$ as:
    $$ \rho' \triangleq \mathcal{E}(\rho) = \sum_{\mathbf{x} \in X \times X} \mathsf{P}_{\mathbf{x}} \rho \mathsf{P}_{\mathbf{x}}. $$
    The state $\rho'$ is block-diagonal with respect to the classical registers, implying $\rho' = \iota(\cqstate)$ for some $\cqstate \in \cqstates{X \times X}{\mathcal{H} \otimes \mathcal{H}}$. Furthermore, since the marginals $\iota(\cqstate_1)$ and $\iota(\cqstate_2)$ are already diagonal in the classical basis, the dephasing operation preserves them. Thus, $\cqstate \in \mathcal{C}(\cqstate_1, \cqstate_2)$.

    We now compare the costs using the truncation technique to handle the unboundedness of $\cqpredicate$. Let $H \triangleq \iota(\cqpredicate)$. Since $\cqpredicate$ is bounded below, $H$ is a bounded-below self-adjoint linear relation. Let $H^{(n)} = \trunc{H}{n}$ be the sequence of bounded truncations (\Cref{def:truncation}).

    Crucially, since $H$ is constructed via the embedding $\iota$, it is block-diagonal with respect to the classical basis. Consequently, its truncations $H^{(n)}$ are also block-diagonal and commute with the projections $\mathsf{P}_{\mathbf{x}}$. This implies that $H^{(n)}$ is a fixed point of the channel $\mathcal{E}$ (or equivalently, $\mathcal{E}^\dagger(H^{(n)}) = H^{(n)}$).

    Since $H^{(n)}$ is bounded, we can apply the cyclic property of the trace:
    $$
    \begin{aligned}
    \tr(H^{(n)} \rho') &= \tr(H^{(n)} \mathcal{E}(\rho)) \\
    &= \tr(\mathcal{E}^\dagger(H^{(n)}) \rho) \quad (\text{Self-    duality of } \mathcal{E}) \\
    &= \tr(H^{(n)} \rho).
    \end{aligned}
    $$
    Finally, we apply \Cref{lemma:trace_approx}, which states that $\Tr(H \sigma) = \lim_{n \to \infty} \tr(H^{(n)} \sigma)$. Taking the limit $n \to \infty$ on both sides:
    $$ \Tr(\iota(\cqpredicate) \rho') = \lim_{n \to \infty} \tr(H^{(n)} \rho') = \lim_{n \to \infty} \tr(H^{(n)} \rho) = \Tr(\iota(\cqpredicate) \rho). $$
    This shows that for any quantum coupling $\rho$, the induced classical-quantum coupling $\cqstate$ (via $\rho'$) achieves the same cost. Therefore, the minimum over classical-quantum couplings cannot be larger than the minimum over quantum couplings. Thus, LHS $\leq$ RHS.

    Combining both directions, the equality holds. 
\end{proof}

\begin{proposition}\label{prop:cq-duality-set-containment}
    For $\cqpredicate \in \cqbndmaps{X}{\cH}$, let $\dualset{\cqpredicate}\subseteq
  \cqbndmaps{X_1}{\cH_1}\times\cqbndmaps{X_2}{\cH_2}$ such that
    $(\cqpredicate_1,\cqpredicate_2)\in \dualset{\cqpredicate}$ iff
      $\cqpredicate_1 \boxplus \cqpredicate_2\sqsubseteq \cqpredicate$ (\Cref{def:extended-lowner-order-cqpredicate}), and  $\dualset{\iota(\cqpredicate)}\subseteq \qbndmaps{\cH_1}\times \qbndmaps{\cH_2}$ such that $(A_1, A_2)\in \dualset{\iota(\cqpredicate)}$ iff $ A_1\boxplus A_2 \sqsubseteq \iota(\cqpredicate)$. If $(A_1, A_2)\in \dualset{\iota(\cqpredicate)}$, then $(\mathcal{R}(A_1), \mathcal{R}(A_2)) \in \dualset{\cqpredicate}$.
\end{proposition}

\begin{proof}
    Let $\Phi$ be the dephasing channel (measurement in the classical basis) on $\cH_1 \otimes \cH_2$. Explicitly, $\Phi(\rho) = \sum_{\vec{x}} \mathsf{P}_{\vec{x}} \rho \mathsf{P}_{\vec{x}}$, where $\mathsf{P}_{\vec{x}}$ are projections onto the classical basis, and $\vec{x}\triangleq(x_1,x_2),x_1\in X_1,x_2\in X_2$. $\Phi$ is completely positive and unital, thus preserves the Löwner order.
    
    A key property linking retraction and embedding is that $\Phi(A \otimes I) = \iota(\mathcal{R}(A)) \otimes I$. This holds because:
    \[
        \Phi(A \otimes I) = \sum_{x_1} (\mathsf{P}_{x_1} \otimes I) (A \otimes I) (\mathsf{P}_{x_1} \otimes I) 
        = \sum_{x_1} V_{x_1} (V_{x_1}^\dagger A V_{x_1}) V_{x_1}^\dagger \otimes I
        = \iota(\mathcal{R}(A)) \otimes I.
    \]
    
    Assume $(A_1, A_2) \in \dualset{\iota(\cqpredicate)}$, i.e., $A_1 \otimes I + I \otimes A_2 \sqsubseteq \iota(\cqpredicate)$. Applying $\Phi$ to both sides:
    \[
        \Phi(A_1 \otimes I) + \Phi(I \otimes A_2) \sqsubseteq \Phi(\iota(\cqpredicate)).
    \]
    Using the property derived above, the LHS becomes $\iota(\mathcal{R}(A_1)) \otimes I + I \otimes \iota(\mathcal{R}(A_2)) = \iota(\mathcal{R}(A_1) \boxplus \mathcal{R}(A_2))$. 
    For the RHS, since $\iota(\cqpredicate)$ is already block-diagonal, it is a fixed point of $\Phi$, so $\Phi(\iota(\cqpredicate)) = \iota(\cqpredicate)$.
    
    Thus, $\iota(\mathcal{R}(A_1) \boxplus \mathcal{R}(A_2)) \sqsubseteq \iota(\cqpredicate)$. Since $\iota$ is an order-embedding (i.e., $\iota(A) \sqsubseteq \iota(B) \iff A(\vec{x}) \sqsubseteq B(\vec{x})$ holds for every $\vec{x}$), we conclude $\mathcal{R}(A_1) \boxplus \mathcal{R}(A_2) \sqsubseteq \cqpredicate$.
\end{proof}

\begin{proposition}\label{prop:split-expectation-cq}
    Let $\cqstate_1\in \cqstates{X_1}{\cH_1}$, $\cqstate_2\in \cqstates{X_2}{\cH_2}$, and $\cqstate\in \couplings{\cqstate_1}{\cqstate_2}$. Then, for $\cqpredicate_1\in \cqbndmaps{X_1}{\cH_1}$ and $\cqpredicate_2\in \cqbndmaps{X_2}{\cH_2}$, we have
    \begin{equation*}
        \E_{\cqstate}[\cqpredicate_1\boxplus \cqpredicate_2] = \E_{\cqstate_1}[\cqpredicate_1] + \E_{\cqstate_2}[\cqpredicate_2].
    \end{equation*}
\end{proposition}
\begin{proof}
     By definition, we know $(\cqpredicate_1\boxplus \cqpredicate_2)(x_1,x_2) = \cqpredicate_1(x_1)\otimes I +I\otimes \cqpredicate_2(x_2)$. Therefore, we have
     \begin{equation*}
         \begin{aligned}
             \E_{\cqstate}[\cqpredicate_1\boxplus \cqpredicate_2]
             &=\sum_{x_1\in X_1,x_2\in X_2} \Tr((\cqpredicate_1(x_1)\otimes I +I\otimes \cqpredicate_2(x_2)) \cqstate(x_1,x_2)) \\
            &=\sum_{x_1\in X_1,x_2\in X_2} \Tr((\cqpredicate_1(x_1)\otimes I) \cqstate(x_1,x_2)) +  \Tr((I\otimes \cqpredicate_2(x_2)) \cqstate(x_1,x_2)) \\
             &=\sum_{x_1\in X_1,x_2\in X_2} \Tr(\cqpredicate_1(x_1)\tr_2(\cqstate(x_1,x_2))) +  \Tr( \cqpredicate_2(x_2)\tr_1(\cqstate(x_1,x_2))) \\
            &=\sum_{x_1\in X_1} \Tr(\cqpredicate_1(x_1)\sum_{x_2\in X_2} \tr_2(\cqstate(x_1,x_2))) + \sum_{x_2\in X_2}  \Tr( \cqpredicate_2(x_2) \sum_{x_2\in X_2}\tr_1(  \cqstate(x_1,x_2))) \\
            &=\sum_{x_1\in X_1} \Tr(\cqpredicate_1(x_1)\cqstate_1(x_1)) + \sum_{x_2\in X_2}  \Tr( \cqpredicate_2(x_2) \cqstate_2(x_2)),
         \end{aligned} 
     \end{equation*}
     which gives the desired result.
\end{proof}

For $x\in X$ and $\rho\in\qstate{\cH}$, we use $\cqsim{x}{\rho}\in\cqstates{X}{\cH}$ to denote a singleton classical-quantum state, which satisfies
\begin{equation*}
    \cqsim{x}{\rho} (y) 
    = \begin{cases}
        \rho, \text{ if } y = x; \\
        0, \text{ otherwise}.
    \end{cases}
\end{equation*}

\begin{proposition}[Adapted from the ePrint version of \cite{barbosa2021easypqc}, Proposition A.9]\label{prop:curri-semantic-function}
    For any cq-program $c_1$, there exist a mapping
    \[
    f: X_1 \times X_1 \to (\qpstate{\cH_1} \to \qpstate{\cH_1})
    \]
    satisfying
    \begin{itemize}
        \item for all $x_i, x_o\in X_1$ and $\rho\in\qstate{\cH_1}$, we have
        \[ \sem{c_1}\cqsim{x_i}{\rho}(x_o) = f(x_i,x_o)(\rho);\]
        \item for all $x_i, x_o\in X_1$,
        $f(x_i, x_o)$ is a trace non-increasing quantum operation on $\cH_{\qvar_1}$;
        \item for all $x_i\in X_1$, $\sum_{x_o\in X_1} f(x_i, x_o)$ is a trace non-increasing quantum operation on $\cH_{1}$; it is trace-preserving if $c_1$ is almost surely terminating ($\AST$);
        \item for any input simple state $\cqsim{(x_{i1},x_{2})}{\rho}$ where $\rho \in\qpstate{\cH_1\otimes\cH_2}$, we have
        \[
        \sem{c}\cqsim{(x_{i1},x_{2})}{\rho}
        = \sum_{x_{o1}} \cqsim{(x_{o1}, x_2)}{(f(x_{i1}, x_{o1})\otimes\mathcal{I}_{\cH_2})(\rho)}.
        \]
    \end{itemize}
\end{proposition}

\begin{proposition}\label{prop:partial-trace-tensor-cq}
    For $\cqstate\in \cqstates{X_1\times X_2}{\cH_1\otimes\cH_2}$, $c_1$, $c_2$ are left and right AST cq-programs, i.e., satisfying $\sem{c_1}:\cqstates{X_1}{\cH_1}\to \cqpstates{X_1}{\cH_1}$, and  $\sem{c_2}: \cqstates{X_2}{\cH_2}\to \cqstates{X_2}{\cH_2}$ which is trace-preserving,
    we have 
    \begin{equation*}
    \begin{aligned}
        \tr_2(( \sem{c_1}\otimes I)(\cqstate)) &= \sem{c_1}(\tr_2(\cqstate)).\\
        \tr_2((I\otimes \sem{c_2})(\cqstate)) &= \tr_2(\cqstate).
    \end{aligned}
    \end{equation*}
\end{proposition}
\begin{proof}
    By linearity, it suffices to prove the proposition for singleton states. Let $\cqstate = \cqsim{x}{\rho}$ for $x = (x_{i1},x_2)\in X_1\times X_2$ and $\rho\in\qstate{\cH_1\otimes \cH_2}$.
    By~\Cref{prop:curri-semantic-function}, we have
    $\sem{c_1} \cqsim{(x_{i1},x_{2})}{\rho}
        = \sum_{x_{o1}} \cqsim{(x_{o1}, x_2)}{(f(x_{i1}, x_{o1})\otimes\mathcal{I}_{\cH_2})(\rho)}$.
    Therefore, by~\Cref{prop:partial-trace-tensor}, we have
    \begin{equation*}
       \begin{aligned}
            \tr_2(\sem{c_1} \cqsim{(x_{i1},x_{2})}{\rho}) &= \sum_{x_{o1}} \cqsim{x_{o1}}{\tr_2((f(x_{i1}, x_{o1})\otimes\mathcal{I}_{\cH_2})(\rho))} \\
             &= \sum_{x_{o1}} \cqsim{x_{o1}}{\tr_2\left((f(x_{i1}, x_{o1})\otimes\mathcal{I}_{\cH_2})(\rho)\right)}\\
             &= \sum_{x_{o1}} \cqsim{x_{o1}}{f(x_{i1}, x_{o1})(\tr_2(\rho))}\\
             &= \sem{c_1}(\tr_2(\cqstate)).
       \end{aligned}
    \end{equation*}

    Similarly, let $\cqstate = \cqsim{x}{\rho}$ for $x = (x_{1},x_{i2})\in X_1\times X_2$ and $\rho\in\qstate{\cH_1\otimes \cH_2}$.
    By~\Cref{prop:curri-semantic-function}, we 
    have $\sem{c_2} \cqsim{(x_{1},x_{i2})}{\rho}
        = \sum_{x_{o2}} \cqsim{(x_{1}, x_{o2})}{(\mathcal{I}_{\cH_1}\otimes f(x_{i2},x_{o2}))(\rho)}$.
    By~\Cref{prop:curri-semantic-function}, $\sum_{x_{o2}}f(x_{i2},x_{o2})$ is a completely positive trace preserving map.
    Therefore, by~\Cref{prop:partial-trace-tensor},
    \begin{equation*}
        \begin{aligned}
           \tr_2(I\otimes \sem{c_2} \cqsim{(x_{1},x_{i2})}{\rho})
        &= \sum_{x_{o2}} \cqsim{x_{1}}{\tr_2((\mathcal{I}_{\cH_1}\otimes f(x_{i2},x_{o2}))(\rho))} \\
        &=  \cqsim{x_{1}}{\tr_2\left((\mathcal{I}_{\cH_1}\otimes \sum_{x_{o2}}f(x_{i2},x_{o2}))(\rho)\right)} \\
        &=  \cqsim{x_{1}}{\tr_2(\rho)} = \tr_2(\cqstate).
        \end{aligned}
    \end{equation*}
    This yields the proof.
\end{proof}

\begin{proposition}\label{prop:ast-preserve-coupling-cq}
     Let $\cqstate_1\in \cqstates{X_1}{\cH_1}$, $\cqstate_2\in \cqstates{X_2}{\cH_2}$, and $\cqstate\in \couplings{\cqstate_1}{\cqstate_2}$. If $c_1$, $c_2$ are left and right AST cq-programs, i.e., satisfying $\sem{c_1}:\cqstates{X_1}{\cH_1}\to \cqstates{X_1}{\cH_1}$ and  $\sem{c_2}: \cqstates{X_2}{\cH_2}\to \cqstates{X_2}{\cH_2}$ which are trace-preserving, then
     $(\sem{c_1}\otimes \sem{c_2})(\cqstate)\in \couplings{\sem{c_1}(\cqstate_1)}{\sem{c_2}(\cqstate_2)}$.
\end{proposition}

\begin{proof}
   Note that $\sem{c_1}\otimes \sem{c_2}(\cqstate)$ is a cq-state.
   We now show $\tr_2((\sem{c_1}\otimes \sem{c_2})(\cqstate)) = \sem{c_1}(\cqstate_1)$. 
   
   By~\Cref{prop:partial-trace-tensor-cq},  we know $$\tr_2((\sem{c_1}\otimes \sem{c_2})(\cqstate)) 
   = \tr_2((\sem{c_1}\otimes I)((I\otimes  \sem{c_2})(\cqstate))) = \sem{c_1}(\tr_2((I\otimes  \sem{c_2})(\cqstate))) = \sem{c_1}(\tr_2(\cqstate)) = 
   \sem{c_1}(\cqstate_1).$$ 
   Similarly, $\tr_1((\sem{c_1}\otimes \sem{c_2})(\cqstate)) = \sem{c_2}(\rho_2)$.
\end{proof}

\begin{proposition}\label{prop:convex-combinations-couplings}
    Suppose for every $i\in \mathbb{N}$, $\cqstate^{(i)} \in \cqstates{X_1\times X_2}{\cH_1\otimes\cH_2}$ is a coupling of 
    $\cqstate_1^{(i)}\in\cqstates{X_1}{\cH_1}$ and $\cqstate_2^{(i)}\in\cqstates{X_2}{\cH_2}$.
    Then, for any sequence of positive reals $\lambda_i$
    satisfying $\sum_i \lambda_i = 1$,
    $\sum_i \lambda_i \cqstate^{(i)}$
    is a coupling of 
    $\sum_i \lambda_i \cqstate_1^{(i)}$ and
      $\sum_i \lambda_i \cqstate_2^{(i)}$.
\end{proposition}

\begin{proof}
    It is clear that  $\sum_i \lambda_i \cqstate^{(i)}$ is a state, by the convexity of states.
    To verify that it is a valid coupling, 
    we compute as
    \[
    \tr_1\left(\sum_i \lambda_i \cqstate^{(i)}\right)
    =  \sum_i \lambda_i  \tr_1(  \cqstate^{(i)})
    = \sum_i \lambda_i \cqstate_{1}^{(i)},
    \]
    and similarly for $\tr_2(\sum_i \lambda_i \cqstate^{(i)})$.
\end{proof}

\begin{proposition}[Embedding and Retraction Preserve Couplings]\label{prop:coupling-preserve-embedding}
    Let $\cqstate_1\in \cqstates{X_1}{\cH_1}$ and $\cqstate_2\in \cqstates{X_2}{\cH_2}$. Then if $\rho\in \couplings{\iota(\cqstate_1)}{\iota(\cqstate_2)}$,
    then $\mathcal{R}(\rho) \in \couplings{\cqstate_1}{\cqstate_2}$.
\end{proposition}

\begin{proof}
    We first prove $\tr_2(\mathcal{R}(\rho) ) = \cqstate_1$, and $\tr_1(\mathcal{R}(\rho) ) = \cqstate_2$ is exactly the same.
    Since $\rho\in \couplings{\iota(\cqstate_1)}{\iota(\cqstate_2)}$, we know
    $$\tr_2(\rho) = \iota(\cqstate_1) = \sum_x \ket{x}\bra{x}\otimes \cqstate_1(x).$$
    Therefore,
    $
    \tr_2(\mathcal{R}(\rho))(x) = \tr_2(V_x^{\dagger} \rho V_x) = V_x^{\dagger} \tr_2( \rho ) V_x =  \cqstate_1(x),
    $
    where the second equality is by~\Cref{prop:partial-trace-tensor}.
\end{proof}

\begin{lemma}[Compactness of cq-couplings]\label{lem:coupling-compact-cq}
    For any states $\cqstate_1$ and $\cqstate_2$, the set of their couplings, i.e., $\couplings{\cqstate_1}{\cqstate_2}$, is compact with respect to the trace-norm topology.
\end{lemma}
\begin{proof}
    It is a direct corollary of~\Cref{lem:coupling-compact}, which states that the set of coupling of two density operators in a separable Hilbert space (i.e., Hilbert space of finite or countably-infinite dimensional) is compact.
    The claim follows by noting that the retraction map
    $\mathcal{R}$ is a bounded and continuous linear map and using~\Cref{prop:coupling-preserve-embedding}, since the continuous image of a compact space is compact.
\end{proof}

\begin{lemma}\label{lemma:cq_expectation_equals_lifting}
    Let $\cqpredicate=\{\cqpredicate(x)\}_{x\in X}\in\qbbmaps{X,\cH}$ be a classical-quantum assertion, and $\cqstate\in\cqpstates{X}{\cH}$. Let $\iota(\cqpredicate)$ and $\iota(\cqstate)$ be defined as \Cref{def:embedding}. Then
    \[
    \Tr(\iota(\cqpredicate)\iota(\cqstate))=\E_{\cqstate}[\cqpredicate].
    \]
\end{lemma}
\begin{proof}
    We denote $H$ as $\iota(\cqpredicate)$ and $\rho$ as $\iota(\cqstate)$. For every $x\in X$, we denote the spectral measure induced by $\cqpredicate(x)$ as $E_x$. Then we can construct the candidate spectral measure $E_H(\cdot)$ for $H$ as:
    \[
    E_H(\Omega)\triangleq\sum_{x\in X}\ket{x}\bra{x}\otimes E_x(\Omega),~\forall\,\Omega\in\mathbb{B}(\bR\cup\{+\infty\}).
    \]
    The sum is convergent under the sense of strong operator topology, and it's easy to verify this construction is indeed the spectral measure corresponding to $H$, i.e., it satisfies the \Cref{def:spectral-measure} and \Cref{thm:spectral-theorem-for-LR}. Thus, since $\rho$ can be written as
    \[
    \rho=\sum_{y\in X}\ket{y}\bra{y}\otimes\cqstate(y),
    \]
    we can compute the induced measure of $H$ on $\rho$ as follows:
    \[
    \begin{aligned}
    \mu_{\rho}^{H}(\Omega)=\tr(E_H(\Omega)\rho)&=\tr\left( \left(\sum_{x} |x\rangle\langle x| \otimes E_x(\Omega)\right) \left(\sum_{y} |y\rangle\langle y| \otimes \cqstate(y)\right) \right)\\
    &=\tr\left( \sum_{x,y} \delta_{xy} |x\rangle\langle x| \otimes (E_x(\Omega)\cqstate(y)) \right) \\
    &= \sum_{x} \tr\left( |x\rangle\langle x| \otimes (E_x(\Omega)\cqstate(x)) \right) \\
    &= \sum_{x} \tr\left( E_x(\Omega)\cqstate(x) \right)\\
    &=\sum_{x} \mu_{\cqstate(x)}^{\cqpredicate(x)}(\Omega), ~\text{for every Borel set $\Omega$ on extended reals.}
    \end{aligned}
    \]
    Thus, assuming that $\cqpredicate=\{\cqpredicate(x)\}_{x\in X}$ are uniformly bounded below by $\gamma$, by the definition of extended trace (\Cref{def:trace_integral}) and the linearity of the Lebesgue-Stieltjes integral for countable sums of measures \cite{rudin1987real}, we know
    \begin{align*}
    \Tr(\iota(\cqpredicate)\iota(\cqstate))\text{\,(i.e.\,$\Tr(H\rho)$)}&=\int_{[\gamma,+\infty]}\lambda\,d\mu_{\rho}^{H}\\
    &=\sum_{x}\int_{[\gamma,+\infty]}\lambda\,d\mu_{\cqstate(x)}^{\cqpredicate(x)}\\
    &=\sum_{x}\Tr(\cqpredicate(x)\cqstate(x))\\
    &=\E_{\cqstate}[\cqpredicate].
    \end{align*}
    This completes the proof.
\end{proof}

\begin{proposition}\label{prop:LSC-with-cqstate}
    Let $\cqpredicate \in \qbbmaps{X, \mathcal{H}}$, and $\{\cqstate_n\}_{n \in \mathbb{N}}$ be a sequence in $\cqpstates{X}{\cH}$ converging to $\cqstate \in \cqpstates{X}{\cH}$ in trace norm, i.e., $\lim_{n \to \infty} \|\cqstate_n - \cqstate\|_{1} = 0$. Then:
    $$ 
    \mathbb{E}_{\cqstate}[\cqpredicate] \leq \liminf_{n \to \infty} \mathbb{E}_{\cqstate_n}[\cqpredicate]. 
    $$
\end{proposition}
\begin{proof}
    Firstly it's not hard to verify that the embedding map $\iota$ preserves the trace norm, i.e., $\|\iota(\cqstate)\|_{1}=\|\cqstate\|_{1}$. Therefore, $\cqstate_n\to\cqstate$ in trace norm $\iff\iota(\cqstate_n)\to\iota(\cqstate)$ in trace norm. 

    We denote $\iota(\cqpredicate),\iota(\cqstate_n)$ and $\iota(\cqstate)$ as $H,\rho_n,\rho$ respectively. By \Cref{lemma:cq_expectation_equals_lifting}, we have $\mathbb{E}_{\cqstate}[\cqpredicate]=\Tr(H\rho)$ and $\mathbb{E}_{\cqstate_n}[\cqpredicate]=\Tr(H\rho_n)$. By \Cref{thm:lsc_rho}, we obtain $\Tr(H\rho)\leq\liminf_{n\to\infty}\Tr(H\rho_n)$, which completes the proof.
\end{proof}

\subsection{The Duality Theorem}

\begin{theorem}[Kantorovich-Rubinstein Duality Theorem for Classical-Quantum Systems]
\label{thm:cq-kr-duality-appendix}
Let $\cqpredicate\in\cqbbmaps{X_1\times X_2}{\cH_1\otimes\cH_2}$, and let $\dualset{\cqpredicate}\subseteq
  \cqbndmaps{X_1}{\cH_1}\times\cqbndmaps{X_2}{\cH_2}$ such that
$(\cqpredicate_1,\cqpredicate_2)\in \dualset{\cqpredicate}$ iff
      $\cqpredicate_1 \boxplus \cqpredicate_2\sqsubseteq \cqpredicate$. Then, for any $\cqstate_1\in \cqstates{X_1}{\cH_1}$ and $\cqstate_2\in \cqstates{X_2}{\cH_2}$,
$$
 \inf_{\cqstate
  \in \couplings{\cqstate_1}{\cqstate_2}} \E_{\cqstate}[\cqpredicate] =\sup_{(\cqpredicate_1,\cqpredicate_2)\in\dualset{\cqpredicate}}
  \{\E_{\cqstate_1}[\cqpredicate_1]+\E_{\cqstate_2}[\cqpredicate_2]\}.$$
\end{theorem}

\begin{proof}
    By~\Cref{prop:inf-embeding-cq}, we have
    \begin{equation*}
        \inf_{\cqstate
  \in \couplings{\cqstate_1}{\cqstate_2}} \E_{\cqstate}[\cqpredicate]  = 
\inf_{\rho\in \couplings{\iota(\cqstate_1)}{\iota(\cqstate_2)}}\tr(\iota(\cqpredicate)\rho) .   \end{equation*}
 Combining the above with~\Cref{thm:quantum-Kantorovich-duality-unbounded}, we know
     \begin{equation*}
        \inf_{\cqstate
  \in \couplings{\cqstate_1}{\cqstate_2}} \E_{\cqstate}[\cqpredicate]  = \sup_{(A_1,A_2)\in\dualset{\iota(\cqpredicate)}}
  \{\tr(A_1\iota(\cqstate_1))+\tr(A_2\iota(\cqstate_2))\}.   
  \end{equation*}
  By~\Cref{prop:bounded-cq-expectation,prop:cq-duality-set-containment}, we know
  \begin{equation*}
      \sup_{(A_1,A_2)\in\dualset{\iota(\cqpredicate)}}
  \{\tr(A_1\iota(\cqstate_1))+\tr(A_2\iota(\cqstate_2))\} = \sup_{(A_1,A_2)\in\dualset{\iota(\cqpredicate)}}  \{\E_{\cqstate_1}[\mathcal{R}(A_1)]+\E_{\cqstate_2}[\mathcal{R}(A_2)]\} \le \sup_{(\cqpredicate_1,\cqpredicate_2)\in\dualset{\cqpredicate}}
  \{\E_{\cqstate_1}[\cqpredicate_1]+\E_{\cqstate_2}[\cqpredicate_2]\}.
  \end{equation*}
  To conclude, we have
  \begin{equation*}
      \inf_{\cqstate
  \in \couplings{\cqstate_1}{\cqstate_2}} \E_{\cqstate}[\cqpredicate] = \sup_{(A_1,A_2)\in\dualset{\iota(\cqpredicate)}}
  \{\tr(A_1\iota(\cqstate_1))+\tr(A_2\iota(\cqstate_2))\} 
  \le \sup_{(\cqpredicate_1,\cqpredicate_2)\in\dualset{\cqpredicate}}
  \{\E_{\cqstate_1}[\cqpredicate_1]+\E_{\cqstate_2}[\cqpredicate_2]\}.
  \end{equation*}

  For the other direction, by the definition of $\sqsubseteq$ and~\Cref{prop:split-expectation-cq}, for any $\cqstate\in \couplings{\cqstate_1}{\cqstate_2}$ and $(A_1,A_2)\in\dualset{\iota(\cqpredicate)}$, we have
  \begin{equation*}
     \E_{\cqstate}[\cqpredicate] \ge \E_{\cqstate}[\cqpredicate_1\boxplus \cqpredicate_2] = \E_{\cqstate_1}[\cqpredicate_1]+\E_{\cqstate_2}[\cqpredicate_2].
  \end{equation*}
  This gives 
  \begin{equation*}
       \inf_{\cqstate
  \in \couplings{\cqstate_1}{\cqstate_2}} \E_{\cqstate}[\cqpredicate] \ge \sup_{(\cqpredicate_1,\cqpredicate_2)\in\dualset{\cqpredicate}}
  \{\E_{\cqstate_1}[\cqpredicate_1]+\E_{\cqstate_2}[\cqpredicate_2]\}.
  \end{equation*}
  Combining both directions, we get the desired result.
\end{proof}

\subsection{Semantics and Weakest Preconditions}

Following~\cite{feng2020}, we introduce the restriction function \(r_b\) and denote \(\cqstate|_b \triangleq r_b(\cqstate)\), where \(\cqstate\in\cqpstates{\mathcal{S}}{\cH}\) and \(b\) is the classical assertion \(b : \mathcal{S}\rightarrow \{\true,\false\}\), as
\[
r_b(\cqstate) = \lambda a.
  \begin{cases}
     \cqstate(a), &\text{ if } b(a) = \true,\\
      0 , &\text{ if } b(a) = \false.\\
  \end{cases}
\]
\begin{lemma}
    \label{lem:prop of restriction function}
    We have the following basic properties of restriction function:
    \begin{enumerate}
        \item \(\cqstate = \cqstate|_b + \cqstate|_{\neg b}\); more generally, if $b = b_1\vee b_2$ and $b_1\wedge b_2 = \false$, then \(\cqstate|_b = \cqstate|_{b_1} + \cqstate|_{b_2}\);
        \item  If \(\E_{\cqstate}[b\mid\psi]<+\infty\), we must have \(\cqstate|_b = \cqstate\) and thus \(\E_{\cqstate}[b\mid\psi] = \E_{\cqstate|_b}[\psi]\).
    \end{enumerate}
\end{lemma}

The semantics of cqWhile is presented in~\Cref{fig:semantics-cqwhile}, following~\cite{feng2020}.
For simplicity, we show the effect of applying the semantics on singleton states for the simple commands, and the result on general states can be uniquely determined by linearity.

\begin{figure}[ht]
$$\begin{array}{rcl}	
\sem{\abort} \cqsim{\cstate}{\rho} &= & 0 \\
\sem{\skp}\cqsim{\cstate}{\rho} &= & \cqsim{\cstate}{\rho} \\
\sem{\assign{x}{e}}\cqsim{\cstate}{\rho} &= & \cqsim{\cstate[\sem{e}_\cstate/x]}{\rho} \\
\sem{\sample{x}{\mu}}\cqsim{\cstate}{\rho} &= & 
    \sum_{v\in\supp{\mu}}\mu(v)\cqsim{\cstate[v/x]}{\rho} \\
\sem{q \coloneqq \ket{v}}\cqsim{\cstate}{\rho} &= & 
    \cqsim{\cstate}{\sum_i|\sem{v}_\cstate\>_{\oq} \<i|\rho|i\>_{\oq} \<\sem{v}_\cstate|} \\
\sem{\oq \coloneqq U [\oq]}\cqsim{\cstate}{\rho} &= & 
    \cqsim{\cstate}{U_{\oq} \rho U^\dagger_{\oq}}\\
\sem{x\gets \imeas\; M[\oq]}\cqsim{\cstate}{\rho} &= & 
    \sum_{v}\cqsim{\cstate[v/x]}{M_v(\rho)M_v^{\dagger}} \\
\sem{c;c'}(\cqstate) &= & 
    \sem{c'}(\sem{c}(\cqstate)) \\ 
\sem{\ifte{b}{c_1}{c_2}}(\cqstate) &=& 
    \sem{c}(\cqstate|_b) + \sem{c'}(\cqstate|_{\neg b}) \\
\sem{\whiledtit{b}{c}{n+1}}(\cqstate) &=& 
    \cqstate|_{\neg b} + \sem{\whiledtit{b}{c}{n}}(\sem{c}(\cqstate|_b)) \\
\sem{\whiled{b}{c}}(\cqstate) &=& 
    \sum_{i=0}^\infty r_{\neg b} \circ (\sem{c} \circ r_b)^i (\cqstate)
\end{array}$$
\caption{Semantics of cqWhile.}\label{fig:semantics-cqwhile}
\end{figure}

For completeness, we present the explicit form of the weakest precondition of cqWhile in~\Cref{tab:weakest-precondition-cq}, which is proposed in~\cite[Table 3]{feng2020}. Note that all predicates here are \emph{bounded}.

\begin{table}[h]
\centering
\begin{tabular}{p{3.5cm}p{9cm}}
\toprule
\textbf{Command} & \textbf{Weakest Precondition} \\
\midrule
$\skp$ & $\wprecond{\skp}{\cqpredicate} \triangleq \cqpredicate$ \\[0.8em]

$\abort$ & $\wprecond{\abort}{\cqpredicate} \triangleq 0$ \\[0.8em]

$\assign{x}{e}$ & $\wprecond{\assign{x}{e}}{\cqpredicate} \triangleq \cqpredicate[e/x]$ \\[0.8em]

$\sample{x}{\mu}$ & $\wprecond{\sample{x}{\mu}}{\cqpredicate} \triangleq \E_{v\sim\mu}[\cqpredicate[v/x]]$ \\[0.8em]

$q \coloneqq \ket{v}$ & $\wprecond{q \coloneqq \ket{v}}{\cqpredicate} \triangleq \sum_{i} \ket{i}_{q}\bra{v}\cqpredicate\ket{v}_{q}\bra{i}$ \\[0.8em]

$\oq \coloneqq U [\oq]$ & $\wprecond{\oq \coloneqq U [\oq]}{\cqpredicate} \triangleq U^{\dagger}_{\oq}\cqpredicate U_{\oq}$ \\[0.8em]

$x\gets \imeas\; M[\oq]$ & $\wprecond{x\gets \imeas\; M[\oq]}{\cqpredicate} \triangleq \sum_i M_i^{\dagger} \cqpredicate[i/x] M_i$ \\[0.8em]

$c_1;c_2$ & $\wprecond{c_1;c_2}{\cqpredicate} \triangleq \wprecond{c_1}{(\wprecond{c_2}{\cqpredicate})}$ \\[0.8em]

$\ifte{b}{c_1}{c_2}$ & 
$\wprecond{\ifte{b}{c_1}{c_2}}{\cqpredicate} \triangleq ( \wprecond{c_1}{\cqpredicate})|_b + 
(\wprecond{c_2}{\cqpredicate})|_{\neg b}$ \\[0.8em]

$\whiled{b}{c}$ & 
$\wprecond{\whiled{b}{c}}{\cqpredicate} \triangleq
\lim_{n\rightarrow\infty} \wprecond{\whiledtit{b}{c}{n} }{\cqpredicate}$ \\

\bottomrule
\end{tabular}
\caption{Structural representations of non-relational weakest preconditions in classical-quantum setting.
}
\label{tab:weakest-precondition-cq}
\end{table}

\subsection{Program Logics}

\begin{theorem}[Soundness and Completeness of Core Rules in Classical-Quantum Programs]\label{thm:classical-quantum-complete-appendix}
A judgment is valid iff it can be derived with the rules \RuleRef{duality},
\RuleRef{conseq}, and \RuleRef{wp}.
\end{theorem}

\begin{proof}
    For soundness, we prove as follows:
    
    \RuleRef{conseq}: By definition of validity and the assumption,  
    for every $\cqstate \in \cqstates{\cstates\times \cstates}{\mathcal{H} \otimes \mathcal{H}}$, 
there exists a coupling $\cqstate'$ in 
  $$\couplings{\sem{S_1}(\tr_2(\cqstate))}{\sem{S_2}(\tr_1(\cqstate))}$$ such that
  $\E_{\cqstate}[P] \geq \E_{\cqstate'}[Q]$.
    Since $P'\sqsupseteq P$, we know $\E_{\cqstate}[P]\le \E_{\cqstate}[P']$. Also, from $Q\sqsupseteq Q'$ we know $\E_{\cqstate'}[Q]\ge \E_{\cqstate'}[Q']$. This gives $\E_{\cqstate}[P'] \geq \E_{\cqstate'}[Q']$.

   \RuleRef{duality}: We fix $\cqstate \in \cqstates{\cstates\times \cstates}{\mathcal{H} \otimes \mathcal{H}}$. By definition of validity and the assumption, for any $(Q_1,Q_2)\in \dualset{Q}$,  there exists a coupling $\cqstate'$ in 
  $\couplings{\sem{S_1}(\tr_2(\cqstate))}{\sem{S_2}(\tr_1(\cqstate))}$  such that $$\E_{\cqstate}[P] \geq \E_{\cqstate'}[Q_1\boxplus Q_2] =  \E_{\sem{S_1}(\tr_2(\cqstate))}[Q_1] + \E_{\sem{S_2}(\tr_1(\cqstate))}[Q_2] $$ by~\Cref{prop:split-expectation-cq}. This means, $$\E_{\cqstate}[P] \ge \sup_{(Q_1,Q_2)\in \dualset{Q}}  \E_{\sem{S_1}(\tr_2(\cqstate))}[Q_1] + \E_{\sem{S_2}(\tr_1(\cqstate))}[Q_2].$$ 
However, by~\Cref{thm:cq-kr-duality-appendix}, we know 
$$\sup_{(Q_1,Q_2)\in \dualset{Q}}  \E_{\sem{S_1}(\tr_2(\cqstate))}[Q_1] + \E_{\sem{S_2}(\tr_1(\cqstate))}[Q_2] = \min_{\cqstate'\in \couplings{\sem{S_1}(\tr_2(\cqstate))}{\sem{S_2}(\tr_1(\cqstate))}} \E_{\cqstate'}[Q].$$ Therefore, taking $\cqstate'$ being the minimizer of $\min_{\cqstate'\in \couplings{\sem{S_1}(\tr_2(\cqstate))}{\sem{S_2}(\tr_1(\cqstate))}} \E_{\cqstate'}[Q]$, we know $\E_{\cqstate}[P]  \ge \E_{\cqstate'}[Q] $.

\RuleRef{wp}: We fix $\cqstate \in \cqstates{\cstates\times \cstates}{\mathcal{H} \otimes \mathcal{H}}$, and let $\cqstate'=  (\sem{S_1}\otimes \sem{S_2})(\cqstate)$. By~\Cref{prop:ast-preserve-coupling-cq}, we know $$\cqstate' \in \couplings{\sem{S_1}(\tr_2(\cqstate))}{\sem{S_2}(\tr_1(\cqstate))}.$$  
By~\Cref{prop:split-expectation-cq}, we know 
$$\tr((Q_1\boxplus Q_2 )\sigma) = \tr(Q_1 \sem{S_1}(\tr_2(\rho)) )+ \tr(Q_1 \sem{S_2}(\tr_1(\rho)) ).$$ Since $ \tr(Q_1 \sem{S_1}(\tr_2(\rho)) ) = \tr(\wprecond{S_1}{Q_1}\tr_1(\rho))$ and $\tr(Q_2 \sem{S_2}(\tr_1(\rho)) ) = \tr(\wprecond{S_2}{Q_2}\tr_2(\rho))$, we know $$\tr((\wprecond{S_1}{Q_1}\boxplus \wprecond{S_2}{Q_2} )\rho ) = \tr(\wprecond{S_1}{Q_1}\tr_1(\rho)) + \tr(\wprecond{S_2}{Q_2}\tr_2(\rho))  = \tr((Q_1\boxplus Q_2 )\sigma).$$

For completeness, we prove as follows:

Assume that $\models
\rtriple{P}{S_1}{S_2}{Q}$. By definition of validity, 
we have  if for every  $\cqstate \in \cqstates{\cstates\times \cstates}{\mathcal{H} \otimes \mathcal{H}}$, 
  there exists a coupling $\cqstate'\in \couplings{\sem{S_1}(\tr_2(\cqstate))}{\sem{S_2}(\tr_1(\cqstate))}$ such that
  $\E_{\cqstate}[P] \geq \E_{\cqstate'}[Q]$. 
This means, for every $\cqstate \in \cqstates{\cstates\times \cstates}{\mathcal{H} \otimes \mathcal{H}}$, $$ \E_{\cqstate}[P]  \geq \inf_{\cqstate'\in \couplings{\sem{S_1}(\tr_2(\cqstate))}{\sem{S_2}(\tr_1(\cqstate))}} \E_{\cqstate'}[Q] .$$
By~\Cref{thm:kr-duality-appendix}, this means
for every $\cqstate \in \cqstates{\cstates\times \cstates}{\mathcal{H} \otimes \mathcal{H}}$, $$  \E_{\cqstate}[P] \geq \sup_{(Q_1,Q_2)\in \dualset{Q}} 
\E_{\sem{S_1}\tr_2(\cqstate)}[Q_1]  + \E_{\sem{S_2}(\tr_1(\cqstate))}[Q_2] .$$
By the property of weakest
precondition, and by~\Cref{prop:split-expectation-prob}, 
we have
for every $\cqstate \in \cqstates{\cstates\times \cstates}{\mathcal{H} \otimes \mathcal{H}}$, and for every $(Q_1,Q_2)\in\dualset{Q}$,
$$ \E_{\cqstate}[P]   \geq  \E_{\cqstate}[\wprecond{S_1}{Q_1}\boxplus
\wprecond{S_2}{Q_2}] .$$
By definition,
this means
$\wprecond{S_1}{Q_1}\boxplus
\wprecond{S_2}{Q_2}\sqsubseteq P$ for every $(Q_1,Q_2)\in\dualset{Q}$.  
By the \RuleRef{wp} and \RuleRef{conseq} rules, it
follows that $\vdash \rtriple{P}{S_1}{S_2}{Q_1\boxplus Q_2}$. One concludes by finally
applying the \RuleRef{duality} rule.
\end{proof}

To show the soundness of the selected two-side rules, we first note that the validity has some equivalent characterizations.

\begin{proposition}[Validity Characterizations for $\AST$ programs]
\label{prop:ast-equivalence-of-validity}
    For any $\AST$ programs $c_1$ and $c_2$, the following are equivalent:
    \begin{enumerate}
        \item \(\vDash \rtriple{b\mid \phi}{c_1}{c_2}{\psi}\);
        \item for every $\cqstate\in\cqstates{\cstates\times \cstates}{\cH\otimes \cH}$ such that $\E_{\cqstate}[b\mid\phi] <  +\infty$ (i.e., has finite expectation), there exists a
        coupling $$\cqstate'\in \couplings{\sem{c_1}(\tr_2(\cqstate))}{\sem{c_2}(\tr_1(\cqstate))}$$ such that
            \(
         \E_{\cqstate'}[\psi] \leq   \E_{\cqstate}[\phi] 
            \);
        \item for every $\cqsim{\cstate_1}{\rho_1}$ and $\cqsim{\cstate_2}{\rho_2}$ such that $\sem{b}_{(\cstate_1,\cstate_2)} = \true$, there exists a coupling $\cqstate' \in \couplings{\sem{c_1}\cqsim{\cstate_1}{\rho_1}}{\sem{c_2}\cqsim{\cstate_2}{\rho_2}}$ such that
        \[
        \inf_{\rho\in \couplings{\rho_1}{\rho_2}}\Tr(\phi(\cstate_1,\cstate_2)\rho) \ge \E_{\cqstate'}[\psi];
        \]
        \item for every \(\cstate_1\in\cstates, \cstate_2\in\cstates,\rho\in\qstate{\cH}\) such that $\sem{b}_{(\cstate_1,\cstate_2)} = \true$, 
        there exists a coupling \(\cqstate'\in \couplings{\sem{c_1}\cqsim{\cstate_1}{\tr_2(\rho)}}{\sem{c_2}\cqsim{\cstate_2}{\tr_1(\rho)}}\),
        such that \(\Tr(\phi(\cstate_1,\cstate_2)\rho) \ge \E_{\cqstate'}[\psi]\).
    \end{enumerate} 
\end{proposition}

\begin{proof}
      $(1)\Rightarrow (2)$. This is direct by definition of validity.

      $(2)\Rightarrow (3)$. Choosing $\cqstate = \cqsim{(s_1,s_2)}{\rho}$ where $\rho$ is the minimizer of  $\inf_{\rho\in \couplings{\rho_1}{\rho_2}}\tr(\phi(\cstate_1,\cstate_2)\rho)$.
      If $\E_{\cqstate}[\phi] = +\infty$, then choosing $\cqstate' = \sem{c_1}\otimes \sem{c_2} (\cqstate)$ suffices as $\E_{\cqstate'}[\psi]\le +\infty$. Otherwise, the existence of $\cqstate'$ is implied by $(2)$.

      $(3)\Rightarrow (4)$. Let $\rho_1 = \tr_2(\rho)$ and $\rho_2=\tr_1(\rho)$. By $(3)$, we know there exists a coupling \(\cqstate'\in \couplings{\sem{c_1}\cqsim{\cstate_1}{\tr_2(\rho)}}{\sem{c_2}\cqsim{\cstate_2}{\tr_1(\rho)}}\) such that
      \begin{equation*}
          \tr(\phi(\cstate_1,\cstate_2)\rho) \ge  \inf_{\rho\in \couplings{\rho_1}{\rho_2}}\tr(\phi(\cstate_1,\cstate_2)\rho) \ge \E_{\cqstate'}[\psi];
      \end{equation*}

      $(4)\Rightarrow (1)$. For any $\cqstate$, we can write it as 
      $\cqstate = \sum_{x_1,x_2} \lambda_{x_1,x_2}\cqsim{(x_1,x_2)}{\cqstate(x_1,x_2)}$, where $ \lambda_{x_1,x_2} \ge 0$ and $\sum_{x_1,x_2}  \lambda_{x_1,x_2} = 1$. 
      If there exists $x_1,x_2$ such that $\sem{b}_{(x_1,x_2)} = \false$ and $\cqstate_{x_1,x_2}\ne 0$, then $\E_{\cqstate}[b\mid \phi]= +\infty$ and we take $\cqstate' = \sem{c_1}\otimes\sem{c_2} (\cqstate)$ which is a valid coupling by~\Cref{prop:ast-preserve-coupling-cq}.
      Therefore, in the following we assume all $x_1, x_2$ satisfies $\sem{b}_{(x_1,x_2)} = \true$ .
      By $(4)$, for each $x_1, x_2$ such that $\sem{b}_{(x_1,x_2)} = \true$ and  $\cqstate(x_1,x_2)$, there exists a coupling $\cqstate'_{x_1,x_2}\in \couplings{\sem{c_1}\cqsim{x_1}{\tr_2(\cqstate(x_1,x_2))}}{\sem{c_2}\cqsim{x_2}{\tr_1(\cqstate(x_1,x_2))}}$, such that
      \(\tr(\phi(x_1,x_2)\cqstate(x_1,x_2)) \ge \E_{\cqstate'_{x_1,x_2}}[\psi]\).
      Then, taking $\cqstate' = \sum_{x_1,x_2} \lambda_{x_1,x_2}\cqstate'_{x_1,x_2}$, by~\Cref{prop:convex-combinations-couplings}, we know $\cqstate'$ is a valid coupling.
      Also, $\E_{\cqstate}[b\mid \phi]\ge \sum_{x_1,x_2} \lambda_{x_1,x_2} \E_{\cqstate'_{x_1,x_2}}[\psi] = \E_{\cqstate'}[\psi]$.
\end{proof}

\begin{theorem}[Soundness of the rules While and Sample.]
    The rules in~\Cref{fig:rules:twosided} are sound.
\end{theorem}

\begin{proof}
 We first prove the soundness of the rule \RuleRef{while}.
 By \Cref{prop:ast-equivalence-of-validity}, it is sufficient to show that for every $\cqstate\in\cqstates{X}{\cH}$ such that $\E_{\cqstate}[b_1\leftrightarrow b_2\mid\psi]<+\infty$, there exists a coupling 
  \[\cqstate'\in\couplings{\sem{\whiled{b_1}{c_1}}(\tr_2(\cqstate))}{\sem{\whiled{b_2}{c_1}}(\tr_1(\cqstate))},\]   such that
  \[
    +\infty > \E_{\cqstate}[b_1 \leftrightarrow b_2\mid\psi] \ge 
    \E_{\cqstate'} [\neg b_1 \wedge \neg b_2\mid\psi].
  \]
   We inductively construct the $\cqstate_n$ which satisfies $\E_{\cqstate_n}[b_1\leftrightarrow b_2\mid\psi] < +\infty$ by:
  \begin{itemize}
     \item $\cqstate_0 = \cqstate$ which satisfies $\E_{\cqstate_0}[b_1\leftrightarrow b_2\mid\psi] < +\infty$ by the assumption;
     \item for $n + 1$, by the assumption, we select the coupling 
      \[
        \cqstate_{n+1} \in \couplings{\sem{c_1}(\tr_2(\cqstate_{n}|_{b_1\wedge b_2}))}{\sem{c_2}(\tr_1(\cqstate_{n}|_{b_1\wedge b_2}))}
      \]
      satisfying
      \begin{align*}
        +\infty &> \E_{\cqstate_{n}}[b_1\leftrightarrow b_2\mid\psi] = 
        \E_{\cqstate_{n}|_{b_1\leftrightarrow b_2}}[\psi] \\
        &\ge \E_{\cqstate_{n}|_{b_1\wedge b_2}}[\psi] = \E_{\cqstate_{n}|_{b_1\wedge b_2}}[b_1\wedge b_2\mid \psi] \\
        &\ge \E_{\cqstate_{n+1}}[b_1 \leftrightarrow b_2\mid \psi] = \E_{\cqstate_{n+1}|_{b_1 \leftrightarrow b_2}}[\psi],
      \end{align*}
  \end{itemize}
 by recalling \Cref{lem:prop of restriction function}, and we also know $\cqstate_n = \cqstate_n|_{b_1\leftrightarrow b_2} = \cqstate_n|_{b_1\wedge b_2} + \cqstate_n|_{\neg b_1\wedge \neg b_2}$ for all $n$.
  We first prove some basic properties of $\cqstate_n$: 
  \begin{enumerate}
    \item $\tr_2(\cqstate_n|_{b_1\wedge b_2}) = \tr_2(\cqstate_n)|_{b_1}$ and $\tr_2(\cqstate_n|_{\neg b_1\wedge \neg b_2}) = \tr_2(\cqstate_n)|_{\neg b_1}$;
    \item $\tr_1(\cqstate_n|_{b_1\wedge b_2}) = \tr_1(\cqstate_n)|_{b_2}$ and $\tr_1(\cqstate_n|_{\neg b_1\wedge \neg b_2}) = \tr_1(\cqstate_n)|_{\neg b_2}$.
  \end{enumerate}
  The proofs are similar and we here only show the first part of (1). For any $\cstate_1$, if $\sem{b_1}_{\cstate_1} = \false$, then obviously that $\tr_2(\cqstate_n)|_{b_1}(\cstate_1) = 0$ and
  \[ \tr_2(\cqstate_n|_{b_1\wedge b_2})(\cstate_1) = \sum_{\cstate_2}\tr_1(\cqstate_{n}|_{b_1\wedge b_2}(\cstate_1,\cstate_2)) = \sum_{\cstate_2} 0 = 0 = \tr_2(\cqstate_n)|_{b_1}(\cstate_1); \]
  if $\sem{b_1}_{\cstate_1} = \true$, for any $\cstate_2$, if $\sem{b_2}_{\cstate_2} = \false$ then $\cqstate_n|_{b_1\wedge b_2}(\cstate_1,\cstate_2)\sqsubseteq \cqstate_n(\cstate_1,\cstate_2) = \cqstate_n|_{b_1\leftrightarrow b_2}(\cstate_1,\cstate_2) = 0$, and if $\sem{b_2}_{\cstate_2} = \true$ then $\cqstate_n|_{b_1\wedge b_2}(\cstate_1,\cstate_2) = \cqstate_n(\cstate_1,\cstate_2)$,  therefore, $\cqstate_{n}|_{b_1\wedge b_2}(\cstate_1,\cstate_2) = \cqstate_{n}(\cstate_1,\cstate_2)$ for all $\cstate_2$, and thus,
  \begin{align*}
    \tr_2(\cqstate_{n}|_{b_1\wedge b_2})(\cstate_1)
    &= \sum_{\cstate_2}\tr_2(\cqstate_{n}|_{b_1\wedge b_2}(\cstate_1,\cstate_2)) 
    = \sum_{\cstate_2}\tr_2(\cqstate_{n}(\cstate_1,\cstate_2)) \\
    &= \tr_2(\cqstate_{n})(\cstate_1)
    = \tr_2(\cqstate_{n})|_{b_1}(\cstate_1).
  \end{align*}

  We next check $\tr_2(\cqstate_n) = (\sem{c_1}\circ r_{b_1})^n(\tr_2(\cqstate))$, which is done by induction. 
  The base case $n=0$ is direct. Assuming it holds for $n$, then for $n+1$, we have
  \[
  \tr_2(\cqstate_{n+1}) = \sem{c_1}(\tr_2(\cqstate_{n}|_{b_1\wedge b_2})) = \sem{c_1}(\tr_2(\cqstate_{n})|_{b_1}) = (\sem{c_1}\circ r_{b_1})^{n+1}(\tr_2(\cqstate))
  \]
  as we desired.
  Similarly, 
  we can prove $\tr_1(\cqstate_n)= (\sem{c_2}\circ r_{b_2})^n(\tr_1(\cqstate_0))$.
  Now, we define
  \[
  \cqstate' = \sum_n \cqstate_n|_{\neg(b_1\wedge b_2)}.
  \]
  We then check
  $\cqstate'$ is a valid coupling in
  $\couplings{\sem{\whiled{b_1}{c_1}}(\tr_2(\cqstate))}{\sem{\whiled{b_2}{c_2}}(\tr_1(\cqstate))}$ by:
  \[
  \begin{aligned}
    \tr_2(\cqstate') &= \sum_n \tr_2(\cqstate_n|_{\neg(b_1\wedge b_2)})
    = \sum_n \tr_2(\cqstate_n)|_{\neg b_1} \\
    &= \sum_n r_{\neg b_1} \circ (\sem{c_1}\circ r_{b_1})^n (\tr_2(\cqstate)) \\
    &= \sem{\whiled{b_1}{c_1}}(\tr_2(\cqstate)),
  \end{aligned}
  \]
  and similarly
  $ \tr_2(\cqstate') =  \sem{\whiled{b_2}{c_2}}(\tr_1(\cqstate))$.
  Finally, we check that
  \[
  \begin{aligned}
    \E_{\cqstate}[b_1\leftrightarrow b_2\mid\psi] 
      &= \E_{\cqstate_0|_{b_1\leftrightarrow b_2}}[\psi] \\
      &= \E_{\cqstate_0|_{\neg b_1\wedge \neg b_2}}[\psi] + \E_{\cqstate_0|_{b_1\wedge b_2}}[\psi] \\
      &\ge \E_{\cqstate_0|_{\neg b_1\wedge \neg b_2}}[\psi] + \E_{\cqstate_1|_{b_1\leftrightarrow b_2}}[\psi] \\
      &= \E_{\cqstate_0|_{\neg b_1\wedge \neg b_2}}[\psi] + \E_{\cqstate_1|_{\neg b_1\wedge \neg b_2}}[\psi] + \E_{\cqstate_1|_{b_1\wedge b_2}}[\psi] \\
      &\ge \cdots \ge \\
      &\ge \sum_n\E_{\cqstate_n|_{\neg b_1\wedge \neg b_2}}[\psi] \\
      &\ge \sum_n\E_{\cqstate_n|_{\neg b_1\wedge \neg b_2}}[\neg b_1\wedge \neg b_2 \mid \psi] \\
      &= \E_{\sum_n\cqstate_n|_{\neg b_1\wedge \neg b_2}}[\neg b_1\wedge \neg b_2\mid \psi] \\
      &\ge \E_{\cqstate'}[\neg b_1\wedge \neg b_2\mid \psi]
  \end{aligned}
  \]
  where the last step is by~\Cref{prop:LSC-with-cqstate}.

   We then prove the soundness of the rule \RuleRef{sample-supp}.
   Note that the linear relation $\E_{(v, w)\sim \mu}[\psi[v/x_1, w/x_2]]$ is well defined by~\Cref{corollary:convex-convergence}.
   By~\Cref{prop:ast-equivalence-of-validity}, we only need to show,
   for every $s_1,s_2\in\cstates$ such that $\sem{\forall x_1 x_2.~\xi \rightarrow b}_{s_1,s_2} = \true$, and $\rho\in \qstate{\cH}$, 
        there exists a coupling \[\cqstate'\in \couplings{\sem{c_1}\cqsim{s_1}{\tr_2(\rho)}}{\sem{c_2}\cqsim{s_2}{\tr_1(\rho)}},\]
        such that \(\Tr(\E_{(v, w)\sim \mu}[\psi[v/x_1, w/x_2]](s_1,s_2)\rho) \ge \E_{\cqstate'}[\psi]\). For $\mu$ such that $\supp(\mu)\subseteq \xi$, we know $\sem{b}_{s_1[i/x_1],s_2[j/x_2]}$ for all $(i,j)\in \supp(\mu)$.
        Therefore, let $$\cqstate' = \sum_{i,j} \mu(i,j)\cqsim{(s_1[i/x_1],s_2[j/x_2])}{\rho}.$$
        It is direct to verify that $\cqstate'$ is a coupling and $\cqstate' = \cqstate'|_b$.
        By~\Cref{thm:quantum-monotone-convergence,corollary:convex-convergence}, we have
        \begin{equation*}
        \begin{aligned}
                        \Tr(\E_{(v, w)\sim \mu}[\psi[v/x_1, w/x_2]](s_1, s_2)\rho) &= \sum_{(v,w)\in \supp (\mu)} \mu(v,w)\tr(\psi[v/x_1, w/x_2](s_1, s_2)\rho)\\
                        &= \sum_{(v,w)\in \supp (\mu)} \mu(v,w)\tr(\psi(s_1[v/x_1], s_2[w/x_2])\rho)\\
                        &= \E_{\cqstate'} [\psi] \\
                         &= \E_{\cqstate'|_b} [\psi] \\
                           &= \E_{\cqstate'} [b\mid\psi] 
        \end{aligned}
        \end{equation*}
        as we desired.
\end{proof}
\section{More Details of Bernoulli Sampling}
\label{sec:more-details-bernoulli}

What remains to be proven is (combining the above simplified representation)
\begin{align*}
  &\left(3Q_1[0/x_1] + Q_1[1/x_1]\right) \boxplus \left(\sum_{ij}\<ij|Q_2[i/x_2'][j/x_2]|ij\>\right) \sqsubseteq 0,
\end{align*}
with $Q_1,Q_2$ satisfying if $\sem{x_1 = x_2\wedge x_2'}_{(\cstate_1, \cstate_2)} = \true$ then
\begin{equation}
  \label{eqn: bern duality1}
  Q_1(\cstate_1) \boxplus Q_2(\cstate_2) \sqsubseteq 0.
\end{equation}
To see this, by definition, it is sufficient to check for all $\cstate_1$ and $\cstate_2$, we have:
\begin{align}
\label{eqn: bern goal}
   &\left(3Q_1[0/x_1] + Q_1[1/x_1]\right)(\cstate_1) +  \left(\sum_{ij}\<ij|Q_2[i/x_2'][j/x_2]|ij\>\right)(\cstate_2) \le 0.
\end{align}
We calculate the LHS of \cref{eqn: bern goal} as follows:
\begin{align*}
  LHS =\ &\left(3Q_1(\cstate_1[0/x_1]) + Q_1(\cstate_1[1/x_1])\right) + 
  \left(\sum_{ij}\<ij|Q_2(\cstate_2[i/x_2', j/x_2])|ij\>\right) \\
  =\ &\left(Q_1(\cstate_1[0/x_1]) + \<00|Q_2(\cstate_2[0/x_2', 0/x_2])|00\> \right) \\
  &+ \left(Q_1(\cstate_1[0/x_1]) + \<01|Q_2(\cstate_2[0/x_2', 1/x_2])|01\>\right) \\
  &+ \left(Q_1(\cstate_1[0/x_1]) + \<10|Q_2(\cstate_2[1/x_2', 0/x_2])|10\>\right) \\
  &+ \left(Q_1(\cstate_1[1/x_1]) + \<11|Q_2(\cstate_2[1/x_2', 1/x_2])|11\>\right) \\
  =\ &\<00| \left(Q_1(\cstate_1[0/x_1]) \boxplus Q_2(\cstate_2[0/x_2', 0/x_2]) \right) |00\> \\
  &+ \<01| \left(Q_1(\cstate_1[0/x_1]) \boxplus Q_2(\cstate_2[0/x_2', 1/x_2]) \right) |01\> \\
  &+ \<10| \left(Q_1(\cstate_1[0/x_1]) \boxplus Q_2(\cstate_2[1/x_2', 0/x_2]) \right) |10\> \\
  &+ \<11| \left(Q_1(\cstate_1[1/x_1]) \boxplus Q_2(\cstate_2[1/x_2', 1/x_2]) \right) |11\>
\end{align*}
Recall \cref{eqn: bern duality1}, notice that 
\begin{align*}
  &\sem{x_1 = x_2\wedge x_2'}_{(\cstate_1[0/x_1], \cstate_2[0/x_2', 0/x_2])} = (0 = 0\wedge 0) = \true \\
  &\sem{x_1 = x_2\wedge x_2'}_{(\cstate_1[0/x_1], \cstate_2[0/x_2', 1/x_2])} = (0 = 0\wedge 1) = \true \\
  &\sem{x_1 = x_2\wedge x_2'}_{(\cstate_1[0/x_1], \cstate_2[1/x_2', 0/x_2])} = (0 = 1\wedge 0) = \true \\
  &\sem{x_1 = x_2\wedge x_2'}_{(\cstate_1[1/x_1], \cstate_2[1/x_2', 1/x_2])} = (1 = 1\wedge 1) = \true
\end{align*}
so we have:
\begin{align*}
  &Q_1(\cstate_1[0/x_1]) \boxplus Q_2(\cstate_2[0/x_2', 0/x_2]) \sqsubseteq 0 \\
  &Q_1(\cstate_1[0/x_1]) \boxplus Q_2(\cstate_2[0/x_2', 1/x_2]) \sqsubseteq 0 \\
  &Q_1(\cstate_1[0/x_1]) \boxplus Q_2(\cstate_2[1/x_2', 0/x_2]) \sqsubseteq 0 \\
  &Q_1(\cstate_1[1/x_1]) \boxplus Q_2(\cstate_2[1/x_2', 1/x_2]) \sqsubseteq 0
\end{align*}
and combine these into the equation we have:
\begin{align*}
   &\left(3Q_1[0/x_1] + Q_1[1/x_1]\right)(\cstate_1) +  \left(\sum_{ij}\<ij|Q_2[i/x_2'][j/x_2]|ij\>\right)(\cstate_2) \\
   \le\ & \<00| 0 |00\> +  \<01| 0 |01\> +  \<10| 0 |10\> +  \<11| 0 |11\> \\
   =\ &0
\end{align*}
which completes the proof of \cref{eqn: bern goal} and thus the original judgment.

\section{More Details of Algorithmic Stability of Quantum Neural Networks}
\label{sec:appendix-algorithmic-stability}

The proof of the judgment
$$         \rtriple
 {X_1\sim_1 X_2 \mid 4cT } { S_{T, X_1, \phi_0}} {S_{T, X_2, \phi_0}} { \lambda(m_1, m_2).
   (P\otimes I -I\otimes P)}$$
is similar to the one in~\cite{erhl}, where $X_1\sim_1 X_2$ means that $X_1$ and $X_2$ differ in exactly one element, i.e.,
$\abs{X_1}= \abs{X_2} = \abs{X_1\cap X_2} +1$.
The crucial step is to apply the two-sided \RuleRef{while} rule from~\Cref{fig:rules:twosided} with the invariant
$$
\psi \triangleq (X_1\sim_1 X_2)\wedge (t^{\ave{1}} =t^{\ave{2}}) \mid \lambda(m_1, m_2). P\otimes I- I\otimes P + (T-t^{\ave{1}})4c I
$$
and loop conditions $b_1 = (t^{\ave{1}} < T)$,  $b_2 = (t^{\ave{2}} < T)$.
After applying the \RuleRef{while} rule, it remains to prove
$$\rtriple{b_1\wedge b_2\mid \psi}{B}{B}{b_1\wedge b_2\mid \psi},$$
where $B$ denotes the loop body.
This can be done by applying the \RuleRef{sample-supp} and with 
$\mu$ being the coupling that coincides almost everywhere with the identity coupling, except for the unique pair of distinct elements $x_1\in X_1$ and $x_2\in X_2$ which are coupled together, and $\xi = \true$ and the \RuleRef{wp} rule, which proves the judgment
\begin{equation*}
    \vdash \rtriple{\E_{(v,v')\sim\mu}[t^{\ave{1}}+1 < T\wedge t^{\ave{2}}+1 <T \mid \psi']}{B}{B}{b_1\wedge b_2\mid \psi}
\end{equation*}
where 
\begin{equation*}
\begin{aligned}
        &\psi'\triangleq (X_1\sim_1 X_2) \wedge (t^{\ave{1}}+1 =t^{\ave{2}}+1) \mid  \\
        &\lambda(m_1, m_2).\;U_{v}^{\dagger}PU_{v}\otimes I- I\otimes U_{v'}^{\dagger}PU_{v'} + (T-t^{\ave{1}}-1)4c I,
\end{aligned}
\end{equation*}
Finally, we apply the \RuleRef{conseq} rule, which the entailment
\begin{equation*}
    (X_1\sim_1 X_2)\wedge b_1\wedge b_2\mid \psi \sqsupseteq \E_{(v,v')\sim\mu}[t^{\ave{1}}+1 < T\wedge t^{\ave{2}}+1 <T \mid \psi'],
\end{equation*}
which we prove via the following proposition.

\begin{proposition}
    For $c\ge 0$, a positive semidefinite operator $P$  such that $0\sqsubseteq P \sqsubseteq I$, and a unitary $U$ such that $\norm{U-I}\le c$. Then,
    \begin{equation*}
        \norm{UPU^{\dagger}-P}\le 2c.
    \end{equation*}
\end{proposition}

\begin{proof}
    By assumption, we have $\norm{P}\le 1$, $\norm{UP}\le 1$, and $\norm{U^{\dagger}-I}\le c$. Therefore, 
    we have $$\norm{UPU^{\dagger}-P} = \norm{UPU^{\dagger}- UP+ UP-P}\le \norm{UPU^{\dagger}- UP} + \norm{UP-P} \le\norm{UP} \norm{U^{\dagger}-I}+ \norm{U-I}\norm{P}\le 2c ,$$
    which yields the proof.
\end{proof}

\end{document}